\documentclass{amsart}
\usepackage[utf8]{inputenc}
\usepackage[a4paper, total={170mm, 237mm}]{geometry}

\usepackage{mathtools}
\usepackage{amsfonts}
\usepackage{pict2e}
\usepackage{amsthm}
\usepackage{subfiles}
\usepackage{amssymb}
\newtheorem{theorem}{Theorem}[section]
\newtheorem{lemma}{Lemma}[section]
\newtheorem{conjecture}{Conjecture}[section]

\begin{document}

\title{Unitary $n$-correlations with restricted support in random matrix theory }
\author{Patrik Demjan and N. C. Snaith}
%\affil{University of Bristol}
%\date{}
\begin{abstract}
We consider the $n$-correlation of eigenvalues of random unitary matrices in the alternative form that is not the tidy determinant common in random matrix theory, but rather the expression derived from averages of ratios of characteristic polynomials in a method that can be mimicked in number theoretical calculations of the correlations of zeros of $L$-functions.  This alternative form for eigenvalues of matrices from $U(N)$ was proposed by Conrey and Snaith and derived by them when certain Fourier variables associated with  the test function sum to a number less than 2, derived by Chandee and Lee for sum less than 4 and here we calculate the expression when the sum is less than 6. 
\end{abstract}

\maketitle
%\clearpage
%\tableofcontents

%\newpage
%\maketitle
\section{Introduction}
There has been interplay between the  theory of random matrices (RMT) and number theory for over 50 years, first through the observation that the statistical distribution of the relevant zeros of the Riemann zeta function and other $L$-functions coincide in the appropriate limit with that of eigenvalues of random unitary matrices \cite{kn:mont73,kn:hejhal94,kn:odlyzko89,kn:rudsar,kn:katzsarnak99a,kn:katzsarnak99b}, and then as it became clear that average values of characteristic polynomials of random matrices can be used to model average values of $L$-functions \cite{kn:keasna00a,kn:keasna00b, kn:cfkrs}. We have learned that  results in random matrix theory can predict and guide results in number theory.

The specific problem we are going to concern ourselves with has to do with the comparison between the \(n\)-point correlation function of a large unitary random matrix and the \(n\)-point correlation function of the non-trivial zeros of the Riemann zeta function or other $L$-function. In random matrix theory the \(n\)-point correlation function of eigenvalues of matrices in an ensemble such as $U(N)$ with Haar measure, can be written very neatly as an $n\times n$ determinant, but it does not seem possible to express the $n$-point correlation function of the zeros of number theoretical functions in a similar form, making it very difficult to compare results to random matrix theory.  Conrey and Snaith \cite{kn:consna06,kn:consna07, kn:consna08, kn:consna14} showed how to calculate the $n$-point correlation function of the eigenvalues of the unitary ensemble in a way that could be mimicked in number theory, using the Ratios Conjectures, and that this method in random matrix theory, although more cumbersome than the determinant, also allows for the support of the Fourier transform of the test function used in the $n$-correlation to be restricted, in exactly the way that is required in order to make the number theoretical calculations tractable.

\section{Correlation functions in random matrix theory} \label{sect:RMT}
%\subfile{sections/ch2_unitary}

If  \(U^{*}\) denotes conjugate transpose of \(U\), then a unitary matrix  \(U\) satisfies the property

\begin{equation}
\label{unitary_matrix}
    U^{*}U = UU^{*} = I,
\end{equation}
 and all of  the  eigenvalues  lie on the unit circle and are therefore of the form \(e^{i\theta_j}\), where the eigenphase \(\theta_j\) lies in the range \([0,2\pi)\).

In order to study any statistics of eigenvalues of matrices belonging to \(U(N)\), we have to equip the group with a probability measure and the natural choice is Haar measure, which is a measure that weights all matrices in this compact set equally. The joint probability density function of the eigenvalues is then:

\begin{equation}
\label{haar_meas}
 P(\theta_1,\dots,\theta_N)d\theta_1 \dots d\theta_N = \frac{1}{(2\pi)^N N!} \prod_{1\leq j<k\leq N} \left|e^{i\theta_k} - e^{i\theta_j}\right|^2 d\theta_1 \dots d\theta_N.
\end{equation}

We can now define the \(n\)-point correlation function

\begin{equation}
\label{n_corr_def}
    R_n^{U(N)}(\theta_1,\dots,\theta_n) = \frac{N!}{(N-n)!} \int_{0}^{2\pi} \dots \int_{0}^{2\pi} P(\theta_1,\dots,\theta_N) d\theta_{n+1} \dots d\theta_{N}.
\end{equation}
As the joint probability density function can also be expressed as
\begin{equation}
\label{measure_as_Sn}
    P(\theta_1,\dots,\theta_N)d\theta_1 \dots d\theta_N = \frac{1}{N!}\det_{N \times N}(S_N(\theta_k-\theta_j)) d\theta_1 \dots d\theta_N,
\end{equation}
where $S_N(\theta) = \frac{1}{2\pi} \frac{\sin(N\theta/2)}{\sin(\theta/2)} $, 
Gaudin's lemma (see \cite{kn:mehta3}, Theorem 5.1.4) can be used to write the correlation function as a determinant: 

\begin{equation}
\label{n_corr_as_Sn}
    R_n^{U(N)} (\theta_1,\dots,\theta_n) = \det_{n\times n}    (S_N(\theta_k-\theta_j)).
\end{equation}

It is often useful to scale the eigenvalues so that their mean spacing is 1 and let the size of the matrix grow. As an example, take the two-point correlation function with eigenvalues scaled so that the new variable \(x_j = N\theta_j/2\pi\), then

\begin{equation}
\label{two_corr_normalised}
    \tilde{R}_{2}^{U(N)}(x_1,x_2) = \left(\frac{2\pi}{N}\right)^2 R_{2}^{U(N)} \left(\frac{2\pi x_1}{N},\frac{2\pi x_2}{N}\right) ,
\end{equation}
where the factor \((2\pi/N)^2\) comes from normalisation and change of variables. Setting \(r= x_2-x_1\) and taking the limit as \(N\) goes to infinity, we  have

\begin{equation}
\label{two_corr_normalised_limit}
    R_2(r) := \lim\limits_{N \to \infty} \tilde{R}_{2}^{U(N)}(x_1,x_2) = 1 - \left(\frac{\sin(\pi r)}{\pi r}\right)^2.
\end{equation}

%\newpage
%\maketitle
\section{Correlation functions in number theory}
%\subfile{sections/ch3_primes}
We can write the  Riemann zeta function as

\begin{equation}
\label{zeta_def}
    \zeta(s) = \sum_{n=1}^{\infty} \frac{1}{n^s} = \prod_{p\text{ prime}} \frac{1}{1-p^{-s}},
\end{equation}
which is valid for \(\Re(s)>1\). This can be extended to the whole complex plane by analytic continuation, except for a pole at \(s=1\).
The so-called trivial zeros of \(\zeta(s)\) are at negative even integers, and can be easily determined from its functional equation. On the other hand, the non-trivial zeros lie in the critical strip of the complex plane, that is in \(0<\Re(s) <1\). The Riemann Hypothesis suggests that the non-trivial zeros lie on the critical line and have the form \(1/2 + i\gamma\), where \(\gamma\) is real.

In 1973 Montgomery \cite{kn:mont73} studied the distribution of the differences \(\gamma - \gamma'\) in terms of the following quantity

\begin{equation} \label{F_2_zeta_def} F_{2,\zeta}(\alpha,\beta;T)=\frac{2\pi}{T\log T} \#\left \{\gamma,\gamma' \in (0,T]: \alpha \leq (\gamma -\gamma') \frac{\log T}{2\pi} < \beta\right \}.\end{equation}

This measures correlations between pairs of unfolded (scaled) zeros. Specifically, the quantity in the curly brackets counts the number of pairs whose difference lies in the range \([\alpha,\beta)\). The role of the scaling factor \(\log T / 2\pi\) is to obtain unit mean spacing between zeros. This is closely related to the two-point correlation mentioned in Section \ref{sect:RMT}. The quantity can be written more explicitly as 
\begin{equation} \label{F_2_zeta_def_explicit}F_{2,\zeta}(\alpha,\beta;T)=\frac{2 \pi}{T \log T} \sum_{ \genfrac{}{}{0pt}{1}{0< \gamma, \gamma'\leq T}{2 \pi \alpha /\log T \leq \gamma-\gamma' < 2 \pi \beta /\log T}} 1.\end{equation}

Montgomery in his paper considered a test function and asymptotically evaluated the following

\begin{equation} \label{F_2_zeta_def_explicit_with_test_function} F_{2,\zeta}(f,T)=\frac{2 \pi}{T \log T} \sum_{0< \gamma,\gamma'\leq T} f\left((\gamma-\gamma')\frac{\log T}{2\pi}\right) w(\gamma-\gamma'),\end{equation}
where \(f(x)\) is a test function whose Fourier transform vanishes outside of \([-1+\delta,1-\delta]\) for any small real \(\delta\), and \(w(x) = 4 / (4 + x^2)\) is a suitable weighting function. Montgomery calculates this by integrating the Fourier transform of \(f(x)\), \(\hat{f}(u) = \int_{-\infty}^{\infty}f(x)\exp(-2\pi i x u)dx \), against  the function

\begin{equation} \label{montgomerys_F_function} F(u) = \frac{2\pi}{T \log T} \sum_{0< \gamma, \gamma'\leq T} T^{i u(\gamma-\gamma')}w(\gamma-\gamma'),\end{equation}
where \(u\) and  \(T\geq2\) are real. We see that

\begin{equation} \label{int_F_f_hat} \begin{aligned} 
\frac{T \log T}{2\pi}\int_{-\infty}^{\infty} F(u) \hat{f}(u) du  &= \int_{-\infty}^{\infty} \sum_{0< \gamma, \gamma'\leq T} w(\gamma-\gamma')T^{iu(\gamma-\gamma')}\hat{f}(u) du \\
&= \sum_{0< \gamma, \gamma'\leq T} w(\gamma-\gamma') \int_{-\infty}^{\infty} T^{iu(\gamma-\gamma')}\hat{f}(u) du \\
&= \sum_{0< \gamma, \gamma'\leq T} w(\gamma-\gamma') f\left((\gamma-\gamma')\frac{\log T}{2\pi}\right).
\end{aligned}\end{equation}

The main focus of Montgomery's paper is a theorem proving the asymptotic behaviour of \(F(u)\) in large \(T\) limit. His result holds for \( -1 < u < 1\) and the behaviour of \(F(u)\) is essentially unknown outside of this range. Therefore, we only consider \(f(x)\) such that its Fourier transform has support in this range. Using the restricted support of the Fourier transform in combination with the Plancherel theorem, it can be shown that 

\begin{equation} \label{two_corr_with_test_function} F_{2,\zeta}(f,T) \sim  f(0) + \int_{-\infty}^{\infty} f(x) \left( 1-\left(\frac{\sin{\pi x}}{\pi x}\right)^2 \right) dx,
\end{equation}
as \(T\) approaches infinity.

He also gives some heuristic arguments for the behaviour of \(F(u)\) outside of this range of $u$, and is able to conjecture the following. For fixed \(\alpha < \beta\),
\begin{equation} \label{two_corr_conjecture}F_{2,\zeta}(\alpha,\beta;T) \sim \int_{\alpha}^{\beta} \left[1- \left(\frac{\sin{\pi x}}{\pi x}\right)^2 \right]dx + \delta(\alpha,\beta) ,\end{equation}
as \(T\) goes to infinity and where \(\delta(\alpha,\beta) = 1\) if \(0 \in [\alpha,\beta]\), \(\delta(\alpha,\beta) = 0\) otherwise.

The connection between number theory and random matrix theory lies in the fact that the integrand in (\(\ref{two_corr_conjecture}\)) is exactly the pair correlation of eigenangles of a random unitary matrix of large order scaled to have a unit mean spacing between eigenvalues which we arrived at in (\ref{two_corr_normalised_limit}). This leads to the idea that the zeros of the Riemann zeta function behave locally in a manner modelled by the eigenvalues of random unitary matrices, in the appropriate scaling limit. This has also been strongly supported by extensive numerical computations done by Odlyzko in the 1980s \cite{kn:odlyzko89}. 

Recalling (\ref{n_corr_as_Sn}) and the definition of the $n$-point correlation function, it can be shown that
\begin{equation} \label{rmt_n_corr_using_gaudins_lemma}\int_{U(N)} \sum_{j_1,\dots,j_n}^{\ast} f(\theta_{j_1},\dots,\theta_{j_n}) dX =  \int_{[0,2\pi]^n} f(\theta_1,\dots,\theta_n) \underset{n \times n}{\text{det}} \, [S_N (\theta_k - \theta_j)] \, d\theta_1 \dots d\theta_n,\end{equation}
where the \(\sum^{\ast}\) denotes a sum over $n$-tuples of distinct indices and \(dX\) is the Haar measure, and where 
\[S_N(\theta)=\frac{1}{2 \pi}\frac{\sin{\frac{N\theta}{2}}}{\sin{\frac{\theta}{2}}}.\]
Note the determinant in the integrand in (\ref{rmt_n_corr_using_gaudins_lemma}) is the \(n\)-point correlation function. The \(n\)-point correlation function having a determinantal form is a feature of random matrix theory and is not something that arises in analogous calculations in number theory.

Meanwhile in number theory, Rudnick and Sarnak \cite{kn:rudsar} proved the equivalent result for \(n\)-correlation of the zeros  of more general $L$-functions. For simplicity, we will specifically focus on the Riemann zeta function. Their definition for measuring the correlations is as follows

\begin{equation} \label{rudnick_sarnak_def} F_{n,\zeta} (B_N,Q) = \frac{1}{N} \# \left\{j_1,\dots, j_n \leq N : \left((\gamma_{j_1}-\gamma_{j_2})\frac{\log T}{2\pi},\dots,(\gamma_{j_{n-1}}-\gamma_{j_n})\frac{\log T}{2\pi}\right) \in Q \right\},\end{equation}
where the \(j\) indices of \(\gamma\)s are distinct and drawn from a set \(B_N\) of size N, and \(Q\) is a box in \(\mathbb{R}^{n-1}\). This again counts the number of \(n\)-tuples of zeros such that the spacings between pairs fall in a defined range/box. They too consider a more general test function \(f\) satisfying certain properties, one of which is that the support of Fourier transform of \(f\) is restricted, just as in Montgomery's calcluation. They give the so-called smoothed correlations
\begin{equation} \label{smoothed_n_corr} F_{n,\zeta} (T,f,w) = \sum_{j_1,\dots,j_n}^{\ast} w\left(\frac{\gamma_{j_1}}{T}\right)\dots w\left(\frac{\gamma_{j_n}}{T}\right)f\left( \frac{\log T}{2\pi}\gamma_{j_1},\dots,\frac{\log T}{2\pi}\gamma_{j_n}\right),\end{equation}
where \(w(u)\) is a rapidly decreasing smooth function and \(f\) is symmetric, translation invariant and decreases rapidly in the hyperplane \(\sum_j \gamma_j =0\). As before, \(\sum^*\) denotes a sum over \(n\)-tuples of distinct indices. Note that the translation invariance of \(f\) is required because we are studying an \(n\)-tuple of zeros along the critical line, meaning that the function \(f\) has to “slide” along the line without changing. (Note that this differs from the test function  that is used in the study of families of $L$-functions and the $n$-level density of their zeros near the critical point, $1/2$, where the test function is required to decay away from the point 1/2.)  Also note that Rudnick and Sarnak require the support of the Fourier transform \(\hat{f}(\xi_1,\ldots,\xi_n)\) be restricted to \(\sum_{j}|\xi_j| < 2\). This restriction implies that the “diagonal” terms in the multiple sums involved in number theoretic calculations dominate over the “off-diagonal” terms which are more difficult to calculate.

In the first half of their paper, they consider the sum in (\ref{smoothed_n_corr}) to be unrestricted, meaning the indices need not be distinct. In the second half, they recover the restricted sum by means of a combinatorial sieving. This sieving is what represents the extra layer of difficulty when comparing the number theoretic results to the corresponding results in random matrix theory. The reason behind this is that correlation functions in random matrix theory naturally take the form of a determinant. This is not the case in number theory and very complex combinatorics is required to build the structure of a determinant. 

This problem has been revisited by Entin \cite{kn:entin14}, who reproves the conclusion of Rudnick and Sarnak, but using the beautiful result on averages of traces of unitary matrices from Diaconis and Shahshahani \cite{kn:diasha}.   

A way to avoid the complicated combinatorics of Rudnick and Sarnak is given  by the works of Conrey and Snaith on triple correlation \cite{kn:consna07} and later on \(n\)-correlation \cite{kn:consna08}. They managed to show that there is an alternative form of the correlation function in random matrix theory by means of the Ratios Theorem of which we will give a simple example:
\begin{theorem} [Example of a Ratios Theorem]
Let \(\Re{\gamma}, \Re{\delta}>0\) and let \(\Lambda_X\) and \(\Lambda_{X^{\ast}}\) be the characteristic polynomials of a matrix \(X\) and its conjugate transpose \(X^{\ast}\) from \(U(N)\) respectively, then
\begin{equation}
    \label{ratios_thm_double}
    \begin{aligned}
    R(\alpha, \beta,\gamma,\delta) &= \int_{U(N)} \frac{\Lambda_{X}(e^{-\alpha}) \Lambda_{X^{\ast}}(e^{-\beta})}{\Lambda_{X}(e^{-\gamma}) \Lambda_{X^{\ast}}(e^{-\delta})} dX \\ &=  \frac{z(\alpha + \beta)z(\gamma + \delta)}{z(\alpha + \delta) z(\beta + \gamma)} + e^{-N(\alpha + \beta)} \frac{z(-\beta - \alpha)z(\gamma + \delta)}{z(-\beta + \delta)z(-\alpha + \gamma)},
    \end{aligned}
\end{equation}
where

\begin{equation}
\label{def_of_z}
    z(x) = \frac{1}{1-e^{-x}}.
\end{equation}
\end{theorem}
More details and the proof of the Theorem can be found in \cite{kn:cfs05}.
One benefit of this approach is the fact that the method can be followed almost exactly in number theory, starting from the Ratios Conjectures for $L$-functions. A simple example involving a ratio of two zeta functions over two zeta functions is the following 
\begin{conjecture}[Example of a Ratios Conjecture] With the constraints on \(\alpha, \beta, \gamma\) and \(\delta \) as described below, we have
\begin{equation}
\label{ratios_conjecture_double}
\begin{aligned}
    R_{\zeta}(\alpha,\beta, \gamma, \delta) &= \int_0^T \frac{\zeta(1/2+it + \alpha)\zeta(1/2 - it + \beta)}{\zeta(1/2+it +\gamma)\zeta(1/2 - it + \delta)}dt \\ &= \int_0^T \bigg[\frac{\zeta(1+\alpha+\beta)\zeta(1+\gamma+\delta)}{\zeta(1+\alpha+\delta)\zeta(1+\beta +\gamma)} A_{\zeta}(\alpha,\beta, \gamma, \delta)
    \\ &+ \left(\frac{t}{2\pi}\right)^{-\alpha-\beta} \frac{\zeta(1-\alpha-\beta)\zeta(1+\gamma+\delta)}{\zeta(1-\beta+\delta)\zeta(1-\alpha +\gamma)}A_{\zeta}(-\beta,-\alpha, \gamma, \delta) \bigg] dt + O(T^{1/2+\epsilon}),
    \end{aligned}
\end{equation}
where \(A_{\zeta}\) is a product over primes, and 

\begin{equation}
\label{ratios_conjecture_double_cond1}
    -\frac{1}{4} < \Re{\alpha}, \Re{\beta} < \frac{1}{4}
\end{equation}
and
\begin{equation}
\label{ratios_conjecture_double_cond2}
    \frac{1}{\log T} \ll \Re{\gamma}, \Re{\delta} < \frac{1}{4}
\end{equation}
and
\begin{equation}
\label{ratios_conjecture_double_cond3}
    \Im{\alpha},\Im{\delta} \ll T^{1-\epsilon}
\end{equation}
for every \(\epsilon>0\).

\end{conjecture}
More details can be found for example in \cite{kn:consna06}. The benefit of this approach is that there need be no restriction on the support of \(\hat{f}(\xi)\) and one can relatively easily identify the lower order terms as well. Note the similar structure of (\ref{ratios_thm_double}) and (\ref{ratios_conjecture_double}), where the \(z\) function in the former plays the role of the zeta function in the latter expression and \(t/2\pi\) is asymptotically the same as  \(e^{N}\) if \(N\) plays the role of \(\log(t)\).  As both $z(x)$ and $\zeta(1+x)$ have poles at $x=0$ with residue 1, the structure of these two formulae are identical and the arithmetic terms are a number theoretical feature that decorate the RMT structure.

Another requirement in comparing the result of Rudnick and Sarnak to the results in random matrix theory is to restrict the support of the test function in the random matrix result, in the same way as is necessary for the number theory calculations to be tractable.   In \cite{kn:consna14}, Conrey and Snaith restrict the support of the test function to \(\sum_n |\xi_n| < 2\), and using the Ratios Theorem, they are able to easily identify which terms survive the restriction as we will see below.

%Alternative formula for n-correlation
The following notation will be useful in understanding the statement of Theorem 4 in \cite{kn:consna14}, which is derived from the Ratios Theorem, and which will be the starting point for our present calculation. It shows how to compute a similar quantity to that in expression  (\ref{rmt_n_corr_using_gaudins_lemma}), but with the sum being unrestricted, which essentially means that the combinatorial arguments in the work of Rudnick and Sarnak would no longer be necessary.

For two sets \(A,B \subset \mathbb{C}\) we will consider a sum over subsets \(S \subset A\) and \(T \subset B\) with \(|S| = |T|\) (this is a different $T$ from above). Let \(\hat{\alpha}\) denote a generic member of \(S\) and \(\hat{\beta}\) denote a generic member of \(T\). The elements \(\alpha\) and \(\beta\) will be either in \(A\) and \(B\) respectively, or \(A-S\) and \(B-T\) respectively, depending on the context. Here and elsewhere we are using \(A - T\) to mean the complement of \(T\) in \(A\).  Also \(S^-=\{-\hat{\alpha}: \hat{\alpha} \in S\}\).  Finally, the sum $K+L+M$ is over set partitions of $\{1,2,\dots,n\}$ into 3 parts, but empty sets are allowed. The conditions on the test function $F$ in the following theorem are described in more detail in the discussion following the statement of the theorem.

\begin{theorem} [Theorem 4 in \cite{kn:consna14}]
\label{insupthm}
Let \(\delta>0\) and let \(\int_{(\delta)}\) denote an integration along the vertical path from \(\delta -i\infty\) to \(\delta + i \infty\). Suppose that \(F(x_1,\dots,x_n)\) is a holomorphic function which decays rapidly in each variable in horizontal strips. Then, for any \(\delta>0\),
\begin{equation}
\label{n_corr_with_J_star}
\begin{aligned}
    \int_{U(N)} \sum_{j_1,\dots,j_n} F(\theta_{j_1},\dots,\theta_{j_n}) dX &= \frac{1}{(2\pi i)^n}\sum_{K+L+M  = \{1,\dots,n\}} (-1)^{|L|+|M|}N^{|M|} \\ & \times \int_{(\delta)^{|K|}} \int_{(-\delta)^{|L|}} \int_{(0)^{|M|}} J^*(z_K;-z_L) F(iz_1,\dots,iz_n) dz_1\dots dz_n,
\end{aligned}
\end{equation}
where \(z_K=\{z_k: k\in K\}, -z_L=\{-z_l: l \in L\}\) and \(\int_{(\delta)^{|K|}} \int_{(-\delta)^{|L|}} \int_{(0)^{|M|}}\) means that we are integrating all the variables in \(z_K\) along the \((\delta)\) path, all of the variables in \(z_L\) along the \((-\delta)\) path and all of the variables in \(z_M\) along the \((0)\) path. And the function

\begin{equation}
\label{J_star_def}
    J^*(A,B):= \sum_{\genfrac{}{}{0pt}{1}{S \subset A, T \subset B}{|S| = |T|}} e^{-N(\sum_{\hat{\alpha} \in S} \hat{\alpha} +\sum_{\hat{\beta} \in T} \hat{\beta})} \frac{Z(S,T)Z(S^-,T^-)}{Z^{\dagger}(S,S^-) Z^{\dagger}(T,T^-)} \sum_{\genfrac{}{}{0pt}{1}{\genfrac{}{}{0pt}{1}{(A-S)+(B-T)}{=W_1+\dots+W_Y}}{|W_y|\leq 2}} \prod_{y=1}^{Y} H_{S,T}(W_y),
\end{equation}
where 
 
\begin{equation}
\label{H_def}
    H_{S,T}(W)=\begin{cases}
    \sum_{\hat{\alpha}\in S} \frac{z'}{z}(\alpha-\hat{\alpha}) - \sum_{\hat{\beta}\in T} \frac{z'}{z}(\alpha + \hat{\beta})       & \quad \text{if } W = \{\alpha\} \subset A-S \\
    \sum_{\hat{\beta}\in T} \frac{z'}{z}(\beta-\hat{\beta}) - \sum_{\hat{\alpha}\in S} \frac{z'}{z}(\beta + \hat{\alpha})  &       \quad \text{if } W = \{\beta\} \subset B-T \\
    \left( \frac{z'}{z} \right)'(\alpha+\beta)   & \quad \text{if } W = \{\alpha,\beta\} \text{with} \genfrac{}{}{0pt}{1}{\alpha \in A-S}{\beta \in B-T} \\
    0  &\quad \text{otherwise}.
  \end{cases}
\end{equation}
Functions \(z(x)\) and \(Z(A,B)\) are defined as

\begin{equation}
\label{z_and_Z_def}
\begin{aligned}
    z(x) &= \frac{1}{1-e^{-x}}, \\ \\
    Z(A,B) &= \prod_{\genfrac{}{}{0pt}{1}{\alpha \in A}{\beta \in B}} z(\alpha + \beta).
\end{aligned}
\end{equation}
The dagger in \(Z^{\dagger}\) means that the factor \(z(x)\) is to be omitted if its argument is zero.
\end{theorem}

Replacing the function \(F\) in the theorem by the test function used by Rudnick and Sarnak (see equation (1.4) in \cite{kn:rudsar})

\begin{equation}
\label{F_def}
    F(iz_1,\dots, iz_n) = f\left( \frac{iNz_{1}}{2\pi}, \dots, \frac{iN z_{n}}{2\pi} \right) h_1\left(\frac{iz_{1}}{\mathcal{T}}\right) \dots h_n\left( \frac{iz_{n}}{\mathcal{T}} \right),
\end{equation}
where \(N\) is sufficiently large in terms of fixed \(n\), and \(\mathcal{T}\) is sufficiently large in terms of \(N\) and 
\begin{equation}
    \label{h_j_defined}
    h_j(x) = \int_{\mathbb{R}}g_j(t)e^{ixt}dt
\end{equation}
with \(g_j\) being smooth and compactly supported, Conrey and Snaith \cite{kn:consna14} are able to show that the restriction on the support of \(\hat{f}(\xi_1,\ldots,\xi_n)\)  to \(|\xi_1|+\dots + |\xi_n| \leq 2q-\epsilon \) for some positive \(\epsilon\) results in simplifying the outermost sum in (\ref{J_star_def}) to the sum over subsets \(S\) and \(T\) of size less than \(q\), as we will now illustrate. 

We start by replacing \(f\) with its Fourier transform
\begin{equation}
    \label{f_replaced_by_phi}
    f\left( \frac{iNz_{1}}{2\pi}, \dots, \frac{iN z_{n}}{2\pi} \right) = 
    \int_{\mathbb{R}^n} \Phi(\xi_1,\dots,\xi_n)\delta\left(\sum_{j =  1}^n \xi_j\right)
    e^{N\left(\sum_{j=1}^n z_j \xi_j\right)}d\xi,
\end{equation}
where \(\Phi\) is smooth, even, and compactly supported. The delta function in the integrand provides the needed translation invariance. Consider
\begin{equation}
    \label{translation_invariance}
    f\left( \frac{iN(z_{1}+t)}{2\pi}, \dots, \frac{iN (z_{n}+t)}{2\pi} \right) = 
    \int_{\mathbb{R}^n} \Phi(\xi_1,\dots,\xi_n)\delta\left(\sum_{j =  1}^n \xi_j\right)
    e^{N\left(\sum_{j=1}^n z_j \xi_j\right)}
    e^{Nt\left(\sum_{j=1}^n \xi_j\right)}d\xi,
\end{equation}
but the delta function forces the sum of \(\xi\)s to be \(0\), therefore the last exponential is equal to \(1\) for any real \(t\). 

Note that the contours of integration restrict the modulus of the real part of the \(z\) variables in $J^*$ in (\ref{n_corr_with_J_star})  to \(\delta\).  Also, if we assume that the Fourier transform \(\Phi\) of \(f\) has support on  \(|\xi_1|+\dots + |\xi_n| \leq 2q-\epsilon \) for some positive \(\epsilon\), this allows us to bound the exponential factor
\begin{equation}
    \label{exp_factor_bound}
    \Bigg|e^{N\left(\sum_{j=1}^n z_j \xi_j\right)} \Bigg|\leq e^{N\delta(2q-\epsilon)}.
\end{equation}
Another exponential factor appears in \(J^*(A,B)\)
\begin{equation}
    \label{j_*_exp_factor}
    \Bigg|e^{-N\left(\sum_{{\hat{\alpha}} \in S} \hat{\alpha} + \sum_{{\hat{\beta}} \in T} \hat{\beta}\right)} \Bigg|\leq e^{-2N\delta q},
\end{equation}
where all of the \(\alpha\)'s and \(\beta\)'s have real parts equal to \(\delta\) and \(|S|=|T|\geq q\). And so the product of these two bounds is \(\leq e^{-N\delta \epsilon}\). As we move the contours away from the imaginary axis, the integrand tends to zero uniformly as the other factors do not interfere with this bound. For example the \(z\) function can be bounded to be 
\begin{equation}
    \label{z_bound}
    |z(\alpha+\beta)| \leq \frac{1}{1- e^{-2\delta}} \to 1
\end{equation}
as \(\delta \to \infty\). And the logarithmic derivative of \(z\) can also be bounded
\begin{equation}
    \label{z_reciprocal_bound}
    \Bigg|\frac{z'}{z}(\alpha+\beta)\Bigg| =\Bigg|1- \frac{1}{1-e^{-(\alpha+\beta)}}\Bigg| = |1 - z(\alpha+\beta)| \leq 2.
\end{equation}
Lastly the \(h\) functions can be bounded
\begin{equation}
    \label{h_bound}
    \Bigg|h\left(\frac{iz}{\mathcal{T}}\right) \Bigg| \leq \int_{-\Delta}^{\Delta} |g(t)|e^{-\delta t/\mathcal{T}} dt
    < e^{\Delta \delta / \mathcal{T}}\int_{-\Delta}^{\Delta} |g(t)| dt = C e^{\Delta \delta / \mathcal{T}},
\end{equation}
as \(g\) is smooth and compactly supported on \([-\Delta, \Delta]\) and \(C\) is a constant. From this we see that any terms with \(|S|=|T|\geq q\) are 0 as \(N \to \infty\) (and \(\mathcal{T} \to \infty\) faster than $N$, which is an assumption made in previous works). Therefore we have the following restriction on the sum

\begin{equation}
\label{J_restricted}
    J^*(A,B) = J^*_q(A,B) := \sum_{\genfrac{}{}{0pt}{1}{S \subset A, T\subset B}{|S| = |T| <q}}\dots.
\end{equation}

With this simplification, Conrey and Snaith \cite{kn:consna14} show that setting \(q=1\) matches the result of Rudnick and Sarnak \cite{kn:rudsar} (see Theorem 3.1 in \cite{kn:rudsar}) without having to consider the rather complicated combinatorics and obtain the following theorem.

\begin{theorem}[Theorem 3 in \cite{kn:consna14}]
\label{insup_result}
Suppose the condition on the support of the Fourier transform of \(f\), \(\Phi\), is \(|\xi_1|+\dots + |\xi_n| < 2\), then the following holds
\begin{equation}
\label{in_supp_result}
\begin{aligned}
\int_{U(N)} & \sum_{j_1,\dots,j_n} h_1\left(\frac{\theta_{j_1}}{\mathcal{T}}\right) \dots h_n\left(\frac{\theta_{j_n}}{\mathcal{T}}\right) f\left(\frac{N\theta_{j_1}}{2\pi},\dots,\frac{N\theta_{j_n}}{2\pi} \right) dX \\ 
&= \kappa(\mathbf{h}) \frac{N\mathcal{T}}{2\pi} \sum_{\genfrac{}{}{0pt}{1}{K+L+M = \{1,\dots,n\}}{|K| = |L|}} \sum_{\sigma \in S_{|K|}} \int_{\xi_{k_j}>0} \xi_{k_1}\dots \xi_{k_{|K|}} \Phi 
\left(\sum_{j=1}^{|K|} \xi_{k_j} \mathbf{e}_{{k_j},{l_{\sigma(j)}}}\right) d\xi + O(N),
\end{aligned}
\end{equation}
where \(\kappa(\mathbf{h})=\int_{\mathbb{R}} h_1(u)\dots h_n(u)du\)  and \(\mathbf{e}_{i,j} = \mathbf{e}_i - \mathbf{e}_j\);  \(\mathbf{e}_i\) being the \(i\)th standard basis vector \((0,\dots,1,\dots,0)\).
\end{theorem}

%%%%%%%%%%%%%%%%%%%%%%%%%%%%%%%%%%%%%%%%%%%%%%%%%%%%%%%%%

\section{Support of the test function in the study of zeros in the number theoretical literature}

We will briefly summarise the literature for various families of $L$-functions where $n$-correlations or $n$-level densities are compared with random matrix theory. We will not give full explanation of all the families here, but refer readers to the original papers. 

After the work of Rudnick and Sarnak \cite{kn:rudsar}  proving that in the appropriate scaling limit, the $n$-correlation of the zeros of the Riemann zeta function agrees, if the  Fourier transform of the test function is supported on $\sum_j|\xi_j|<2$, with the $n$-correlation of eigenvalues of matrices chosen randomly from $U(N)$ with Haar measure (and they show very similar results for other individual $L$-functions), another big step in the random matrix theory/number theory connection was Katz and Sarnak \cite{kn:katzsarnak99a,kn:katzsarnak99b} showing that for function field zeta-functions, in an appropriate limit, natural families of zeta-functions have zero statistics that agree with eigenvalue statistics of matrices chosen with respect to Haar measure not only from $U(N)$  but also from other classical compact groups.  This philosophy extends to expecting the same behaviour from natural families of number field $L$-functions. 

Thus, parallel to the study of the non-trivial zeros of \(\zeta(s)\) is the study of the low-lying non-trivial zeros of families of $L$-functions. By low-lying non-trivial zeros we mean the zeros that lie in proximity to the critical point - the point where the critical line crosses the real axis. In analogy to the \(n\)-point correlation function of the zeros of \(\zeta(s)\) high on the critical line matching the \(n\)-point correlation of eigenvalues from \(U(N)\) with large $N$, there are a variety of results on the \(n\)-level density of low-lying zeros of families of $L$-functions matching the same statistic of eigenvalues of matrices from $U(N)$, $USp(2N)$, $SO(2N)$ or $SO(2N+1)$. Here we are using the term “\(n\)-level density” as opposed to the \(n\)-point correlation function to distinguish the study of zeros on large portion of the critical line of one $L$-function ($n$-correlation) and zeros close to the critical point across a family of $L$-functions ($n$-level density). In the case of \(n\)-point correlation function we also require that the test function is translation-invariant, and so focuses on the relative distance between eigenvalues.  

For families of $L$-functions, the problem is again to show the equivalence of the leading order terms in the \(n\)-level density statistics in random matrix theory and number theory. This was done for a family of Dirichlet $L$-functions associated with quadratic characters by Entin, Roditty-Gershon and Rudnick \cite{kn:entrodrud} with support \(\sum_{i=1}^{n} |u_i| <2\) by using comparison with the equivalent problem for function field zeta functions.  On the other hand, the method that is closer to this current paper uses the idea that it is easier to make this comparison using an alternative formula for the \(n\)-level density, derived from Ratios Theorems, than to try to compare number theory to the determinantal expressions for $n$-level densities. This was demonstrated by Mason and Snaith, firstly in  \cite{kn:massna16} where they first derived the alternative formula for the \(n\)-level density of eigenvalues from \(USp(2N)\) and then they used that to match their result to results in number theory \cite{kn:massna16b}, which again requires a restriction on the support of the Fourier transform of the test function. In particular, Mason and Snaith are able to match their result for the $n$-level density to that of Rubinstein \cite{kn:rub01}, for a family of Dirichlet $L$-functions associated with quadratic characters, with the support being \(\sum_{i=1}^{n} |u_i| <1\), and to Gao's result \cite{kn:gao14} for the same family with support \(\sum_{i=1}^{n} |u_i| <2\). Previously, the best result was to match up to the $7$-level density with random matrix theory, in the paper by Levinson and Miller \cite{kn:levmil13}.  Furthermore, Mason and Snaith manage to write down the formula with restriction \(\sum_{i=1}^{n} |u_i| <3\) in the hopes that future work in number theory might further extend the range of support.  As it relates also to the same family, we note here the precursor to the results of \cite{kn:rub01} and \cite{kn:gao14} by {\"O}zl{\"u}k and Snyder \cite{kn:ozlsny99} on the one-level density of quadratic Dirichlet $L$-functions.

There is a range of literature comparing $n$-level densities, in the appropriate asymptotic limit, of various families of $L$-functions to random matrix theory, much of it considering just the one-level density, where the identification with random matrix theory is easier to make.    Iwaniec, Luo and Sarnak \cite{kn:ILS99} look at some families of $L$-functions associated with automorphic forms on $GL_2$ and $GL_3$. They show, largely assuming the Generalised Riemann Hypothesis, that the one-level density agrees with random matrix theory for families with either orthogonal or symplectic symmetry,  for various ranges of support depending on the exact family and scaling regime, mostly being within (-2,2), but extending beyond this with the help of a further hypothesis.  Generalised and similar results extending the support or looking at slightly different families followed by Liu and Miller \cite{kn:milliu}, Barrett, Burkhardt, DeWitt, Dorward and Miller \cite{kn:bbddm}, Lesesvre \cite{kn:lesesvre} and Hamieh and Wong \cite{kn:hamwong}, with support in the $(-2,2)$ window.  In addition, Hughes and Miller \cite{kn:hugmil07}, Li and Miller \cite{kn:limil} and Cohen et al. \cite{kn:cohen_ea} look at moments of the one-level density and compare with random matrix theory.  One-level densities are also considered in other families of $L$-functions, such as Alpoge and Miller \cite{kn:alpmil14}, Fouvry and Iwaniec \cite{kn:fouiwa03}, G{\"u}lo{\u g}lu \cite{kn:guloglu05}, Young \cite{kn:young}, Miller and Peckner \cite{kn:milpec12} and Shin and Templier \cite{kn:shitem16}.  

2-level densities are included in the work of Miller \cite{kn:mil04}, Due{\~n}ez and Miller \cite{kn:duemil06}, Ricotta and Royer \cite{kn:ricroy11} and Alpoge et al.~\cite{kn:aailmz} to determine symmetry type when the one-level density is not sufficient, and in the work of Mauro, Miller and Miller \cite{kn:maumilmil}, where they use the idea of \cite{kn:ILS99} to push the support beyond 2 by assuming an additional hypothesis. 

There is also a series of papers that look at what the Ratios Conjecture predicts for the 1-level density - that is using the Ratios Conjecture to derive terms lower than the leading order asymptotic, down to the constant term \cite{kn:goesetal,kn:huymilmor,kn:mil07,kn:mil09,kn:milmon}.
Fiorilli and Miller \cite{kn:fiomil15}, in the case of Dirichlet $L$-functions with character modulo $q$, look at the one-level density and obtain lower order terms below those which can be predicted by the Ratios Conjecture.

In \cite{kn:balchali}, Baluyot, Chandee and Li look at the one-level density of the family of $L$-functions associated with Hecke newforms of level $q$ averaged over a range of $q\asymp Q$ and show, assuming the Generalised Riemann Hypothesis, that as $Q\rightarrow \infty$ the limiting random matrix prediction for an ensemble of orthogonal matrices \cite{kn:katzsarnak99a} holds when the Fourier transform of the test function has support in $(-4,4)$. It is the additional average over $q$ that allows them to extend the range of support. 

One-level densities where each $L$-function is weighted according to the size of the $L$-function at the central value, where the critical line crosses the real axis, are compared to random matrix theory in \cite{kn:fazzari24,kn:betfaz24}.

The paper closest to this current work is by Chandee and Lee \cite{chandee}, who look at a family of primitive Dirichlet $L$-functions that display unitary symmetry - that is the zeros near the critical point behave like the eigenvalues of matrices averaged over the full unitary group $U(N)$. They study the $n$-level density, which makes comparison with random matrix theory a much more difficult problem than the 1- or 2-level densities.  They solve this problem using the alternative form of the $n$-level density proposed in \cite{kn:consna14}  and \cite{kn:massna16b} but extend the support to \(\sum_{i=1}^{n} |u_i| <4\), or in the notation of this current paper, $q=2$. They develop the random matrix theory $n$-level density alternative formula for $q=2$ in their Section 6 and match this to the number theoretical calculation.  This is a convincing example of the usefulness of the alternative $n$-level density random matrix formula that comes from the Ratios Theorem. 

In the following section we carry out a $q=2$ calculation that is very similar to that of Chandee and Lee, except that as we are computing an $n$-correlation function, our test function is slightly different, as it includes the delta function that incorporates the translation invariance of the correlation functions.  As the result of this procedure is fairly complicated, there are many ways of expressing the resulting sums.  We choose to use a form more like that of Rubinstein \cite{kn:rub01} and Gao \cite{kn:gao14}, so it looks somewhat different from Chandee and Lee, although the steps to arrive at the final formula are similar. 

In the final section, we extend this result to $q=3$.  We don't know of a number theoretical result yet that requires this range of support, but the previous paragraphs are evidence of the volume of activity on this type of problem and we hope that the steps laid  out here will help with comparison to future number theoretical results of $n$-correlation or $n$-level density. 

Our main results are Theorems \ref{q_2_thm} and \ref{q_3_thm}.

%%%%%%%%%%%%%%%%%%%%%%%%%%%%%%%%%%%%%%%%%%%%%%%%%%%%%%%%%%
%\newpage
%\maketitle
\section{Extension of the support of \(\Phi\) to q = 2}
%\subfile{sections/ch5_q_2}
In this section and the next, we are going to generalise the work of Conrey and Snaith \cite{kn:consna14} by extending the support of \(\hat{f}(u_1,\ldots,u_n)\) to \(\sum_j |u_j| < 4\) and \(\sum_j |u_j| < 6\), effectively setting \(q=2\) and \(3\). This means that we will have to consider more terms in the sum as well as the exponential factor. 

The following theorem is a result of extending the previous result of Conrey and Snaith \cite{kn:consna14} and extends the support of the Fourier transform of the test function \(f\) to be \(\sum_j |\xi_j| < 4\).  The sum over \((A:B)\) denotes the sum over all possible pairs of elements from \(A\) and \(B\) respectively. Note that here we are extending the eigenvalues so that \(N+1\)st eigenvalue is the first eigenvalue \(+2\pi\). Therefore the indices go on to infinity.

\begin{theorem}
\label{q_2_thm}
Let \(h_j\), for \(1\leq j \leq n\), be rapidly decaying functions with
\begin{equation}
    \label{q_2_thm_h_def}
    h_j(x)= \int_{\mathbb{R}} g_j(t)e^{ixt}dt,
\end{equation}
where \(g_j\) is smooth and compactly supported. Let \(\delta\) be the Dirac \(\delta\)-function. Suppose that
\begin{equation}
    \label{f_trans_def_q_2}
    f(x_1,\dots,x_n) = \int_{\mathbb{R}^n} \Phi(\xi_1,\dots,\xi_n)\delta(\xi_1 + \dots + \xi_n)e^{-2\pi i x_1\xi_1 -\dots-2\pi i x_n\xi_n}d\xi_1 \dots d\xi_n,
\end{equation}
where \(\Phi\) is smooth, even, and compactly supported in such a way that
\begin{equation}
    \label{phi_0}
    \Phi(\xi_1,\dots,\xi_n) =0
\end{equation}
whenever \(|\xi_1|+\dots + |\xi_n| > 4 - \epsilon\) for some \(\epsilon >0\). Then the following holds, with
\begin{equation}
    F(iz_1,\dots, iz_n) = f\left( \frac{iNz_{1}}{2\pi}, \dots, \frac{iN z_{n}}{2\pi} \right) h\left(\frac{iz_{1}}{\mathcal{T}}\right) \dots h\left( \frac{iz_{n}}{\mathcal{T}} \right),
\end{equation}
where \(N\) is sufficiently large in terms of fixed \(n\), and  \(\mathcal{T}\) goes to infinity faster than \(N\).
\begin{align}
    \label{q_2_result_in_thm}
    \int_{U(N)} &\sum_{j_1,\dots,j_n}F(\theta_{j_1},\dots,\theta_{j_n})dX = \notag\\ \notag\\
    &\kappa(\mathbf{h}) \frac{N\mathcal{T}}{2\pi} \sum_{K+L+M = \{1,\dots,n\}} \Bigg\{ \sum_{(K:L)} \int_{\mathbb{R}^{|K|}_{\xi_{k_j}>0}} \prod_{j=1}^{|K|}(\xi_{k_j}) \Phi 
\left(\sum_{j=1}^{|K|} \xi_{k_j}\mathbf{e}_{k_j} - \sum_{j=1}^{|K|} \xi_{k_j}\mathbf{e}_{l_j}\right) d\xi \notag \\ \notag \\ 
&
 +(-1)^{|L|}  \sum_{\genfrac{}{}{0pt}{1}{R \subsetneq K, Q \subsetneq L}{|R|=|Q|}} \sum_{(R:Q)} \int_{\mathbb{R}_{\xi>0}^{|Q|}} 
        \prod_{q_j \in Q}(-\xi_{q_j}) \notag\\ \notag\\
        &\times
        \sum_{R_1\subsetneq R^c, Q_1 \subsetneq Q^c} \sum_{k \in R_1^c, l \in Q_1^c} (-1)^{|R_1^{>k} \cup Q_1^{>l}|} \int_{\mathbb{R}_{\xi_l}}\int_{\mathbb{R}_{\xi>0}^{|R^c\backslash\{k\} \cup Q^c\backslash\{l\}|}} \left(\xi_l + 1 + \sum_{j \in R_1^c\backslash\{k\} \cup Q_1^{<l}} \xi_j - \sum_{j \in Q_1^{>l}} \xi_j\right) \notag\\ \notag\\
        &\times \Phi\Bigg(\sum_{j \in R^c\backslash\{k\} \cup L} \xi_j \mathbf{e}_j + \sum_{j=1}^{|R|} \xi_{q_j}\mathbf{e}_{r_j} + \left(\sum_{j \in S_2} \xi_j - \sum_{j \in S_1\backslash\{k\}} \xi_j\right)\mathbf{e}_k \Bigg) \quad d\mathbf{\xi}\Bigg\} + O(N).
     \end{align}
where \(\kappa(\mathbf{h})=\int_{\mathbb{R}} h_1(u)\dots h_n(u)du\)  and there is a condition on the \(\xi\) integrals: \(\xi_l + \sum_{j \in R_1^c\backslash\{k\}} \xi_j +\sum_{j \in Q_1^{<l}} \xi_j +1 < \sum_{j \in Q_1^{>l}} \xi_j\).  \(K+L + M = \{1,\dots,n\}\) denotes partitioning \(\{1,\dots,n\}\)  into three disjoint sets. \(R^c\)  denotes the complement of \(R\) in $K$, $R_1^c$ the complement of $R_1$ in $R^c$,  \(Q^c\)  denotes the complement of \(Q\) in $L$ and $Q_1^c$ the complement of $Q_1$ in $Q^c$ . The sum denoted \((R:Q)\) indicates summing over all possible ways to pair one element from $R$ with one element of $Q$.  In particular, the sum $\sum_{(K:L)} $ is zero unless $|K|=|L|$.  The set $R_1^{>k}$ contains all elements that are in $R_1$ and are larger than the particular $k$ picked out by the sum over $k\in R_1^c, l\in Q_1^c$. Sets \(S_1\) and \(S_2\) are \(R_1^c \cup R_1^{>k}\cup Q_1^{<l}\cup\{l\}\) and \(R_1^{<k} \cup Q_1^{>l} \cup Q_1^c\) respectively.
    
\end{theorem}

{\it Remark:}  As an example, if $R=\{1,3\}$ and $Q=\{2,4\}$ then the sum $\sum_{(R:Q)}$ sums over two terms: $(r_1,q_1)=(1,2),(r_2,q_2)=(3,4)$ and $(r_1,q_1)=(1,4), (r_2,q_2)=(3,2)$. Note that pairs $(r_j,q_j)$ are all treated identically, so it doesn't matter if we label as $(r_1,q_1)=(1,4), (r_2,q_2)=(3,2)$ or $(r_2,q_2)=(1,4), (r_1,q_1)=(3,2)$; we want each pairing just once, so we could decide, for example, that $r_1<r_2<\cdots <r_{|R|}$.

We begin the calculation with generic \(q\) by replacing the function \(F\) in Theorem \ref{insupthm} by 

\begin{equation}
\label{F_def_again}
    F(iz_1,\dots, iz_n) = f\left( \frac{iNz_{1}}{2\pi}, \dots, \frac{iN z_{n}}{2\pi} \right) h\left(\frac{iz_{1}}{\mathcal{T}}\right) \dots h\left( \frac{iz_{n}}{\mathcal{T}} \right),
\end{equation}
where \(N\) is sufficiently large in terms of fixed \(n\), and \(\mathcal{T}\) is sufficiently large in terms of \(N\), that is \(\mathcal{T}\) goes to infinity faster than \(N\). Also \(h_j(x) = \int_{\mathbb{R}} g_j(t) e^{ixt} dt\) are rapidly decaying functions and \(g_j\) are smooth and compactly supported. We also replace \(f\) by its Fourier Transform
\begin{equation}
    \label{f_replaced_by_phi_pt2_a}
    f\left( \frac{iNz_{1}}{2\pi}, \dots, \frac{iN z_{n}}{2\pi} \right) = 
    \int_{\mathbb{R}^n} \Phi(\xi_1,\dots,\xi_n)\delta\left(\sum_{j =  1}^n \xi_j\right)
    e^{N\left(\sum_{j=1}^n z_j \xi_j\right)}d\xi,
\end{equation}
where \(\Phi\) is smooth, even, and compactly supported on \(\sum_{j}|\xi_j| <2q\). As explained and derived previously, this restriction simplifies the form of \(J^*\) to sum over subsets of size less than \(q\). That is

\begin{equation}
\label{J_star_restricted}
    J_q^*(A,B) = \sum_{\genfrac{}{}{0pt}{1}{S \subset A, T\subset B}{|S| = |T| <q}}\dots.
\end{equation}
Replacing \(F\) by (\ref{F_def_again}), \(f\) by its Fourier transform (\ref{f_replaced_by_phi_pt2_a}), and \(J^*\)
by \(J_q^*\) (\ref{J_star_restricted}) in Theorem \ref{insupthm} we have

\begin{equation}
\label{n_corr_with_J_star1}
\begin{aligned}
    \int_{U(N)} \sum_{j_1,\dots,j_n} F(\theta_{j_1},\dots,\theta_{j_n}) dX &= \frac{1}{(2\pi i)^n}\sum_{K+L+M  = \{1,\dots,n\}} (-1)^{|L|+|M|}N^{|M|} \\ & \times \int_{(\delta)^{|K|}} \int_{(-\delta)^{|L|}} \int_{(0)^{|M|}} J^*(z_K;-z_L) F(iz_1,\dots,iz_n) dz_1\dots dz_n \\
    &= \frac{1}{(2\pi i)^n}\sum_{K+L+M  = \{1,\dots,n\}} (-1)^{|L|+|M|}N^{|M|} \\ & \times \int_{(\delta)^{|K|}} \int_{(-\delta)^{|L|}} \int_{(0)^{|M|}} J_q^*(z_K;-z_L) \, \prod_{j=1}^{n} h\left( \frac{iz_{j}}{\mathcal{T}} \right)\\
    &\times \int_{\mathbb{R}^n} \Phi(\xi_1,\dots,\xi_n)\delta\left(\sum_{j =  1}^n \xi_j\right)
    e^{N\left(\sum_{j=1}^n z_j \xi_j\right)}d\xi \, dz_1\dots dz_n.
\end{aligned}
\end{equation}
We extend the previous result by Conrey and Snaith \cite{kn:consna14} Theorem 3, which we state here again for clarity
\begin{equation}
\label{in_supp_result_chpt_q_2}
\begin{aligned}
\int_{U(N)} & \sum_{j_1,\dots,j_n} h_1\left(\frac{\theta_{j_1}}{T}\right) \dots h_n\left(\frac{\theta_{j_n}}{T}\right) f\left(\frac{N\theta_{j_1}}{2\pi},\dots,\frac{N\theta_{j_n}}{2\pi} \right) dX \\ 
&= \kappa(\mathbf{h}) \frac{NT}{2\pi} \sum_{\genfrac{}{}{0pt}{1}{K+L+M = \{1,\dots,n\}}{|K| = |L|}} \sum_{\sigma \in S_{|K|}} \int_{\xi_{k_j}>0} \xi_{k_1}\dots \xi_{k_{|K|}} \Phi 
\left(\sum_{j=1}^{|K|} \xi_{k_j} \mathbf{e}_{{k_j},{l_{\sigma(j)}}}\right) d\xi + O(N),
\end{aligned}
\end{equation}
where the support of \(\Phi\) was assumed to be \(\sum_{j}|\xi_j|<2\), which means \(q=1\). Here we use the notation
\begin{eqnarray}
\left\{
\begin{array}{ll}
\mathbf e_{i,j} &= \mathbf e_i -\mathbf e_j\\
\mathbf e_i & = (0,\dots,1, \dots ,0) \mbox{ the $i$th standard basis vector.} 
\end{array}
\right.
\end{eqnarray}

We now extend this to consider the support of \(\hat{f}\) to be \(\sum_j |\xi_j| < 4\), setting \(q=2\). This means that in the first sum in (\ref{J_star_def}), we consider subset \(S\) and \(T\) to be empty or having one element at most, whereas in the Conrey Snaith result above, they only had to consider \(S\) and \(T\) being empty. Since the sum over empty sets has already been worked out in \cite{kn:consna14}, we will consider the sum over singletons.
Considering a generic choice of  sets \(z_K=\{\alpha_k: k\in K\}, -z_L=\{-\beta_l: l \in L\}\), we get

\begin{equation}
\label{J_star_2_def}
J_2^*(z_K,-z_L) = J_{\emptyset,\emptyset}^* + \sum_{\genfrac{}{}{0pt}{1}{\{\alpha\} \subset z_K}{\{-\beta\}{\subset -z_L}}} J_{\alpha,-\beta}^*,
\end{equation}
where 
\begin{equation}
    \label{j_empty_empty}
    J_{\emptyset, \emptyset}^* = \sum_{(z_K:-z_L)}\prod_{j=1}^{|K|}\left(\frac{z'}{z}\right)' (\alpha_{k_j} -\beta_{l_j}),
\end{equation}
\begin{equation}
    \label{J_star_a_b_def}
    J_{\alpha,-\beta}^* = e^{-N(\alpha - \beta)}z(\alpha - \beta)z(-\alpha + \beta) \sum_{\genfrac{}{}{0pt}{1}{\genfrac{}{}{0pt}{1}{(z_K\backslash \{\alpha\})+(-z_L\backslash  \{-\beta\})}{=W_1+\dots+W_Y}}{|W_y|\leq 2}} \prod_{y=1}^{Y} H_{\{\alpha\},\{-\beta\}}(W_y)
\end{equation}

and

\begin{equation}
\label{H_a_b_def}
    H_{\{\alpha\},\{-\beta\}}(W)=\begin{cases}
    \frac{z'}{z}(\alpha_k-\alpha) - \frac{z'}{z}(\alpha_k - \beta)       & \quad \text{if } W = \{\alpha_k\} \subset z_K\backslash \{\alpha\} \\
    \frac{z'}{z}(-\beta_l+\beta) - \frac{z'}{z}(-\beta_l + \alpha)  &       \quad \text{if } W = \{-\beta_l\} \subset -z_L\backslash \{-\beta\} \\
    \left( \frac{z'}{z} \right)'(\alpha_k-\beta_l)   & \quad \text{if } W = \{\alpha_k,-\beta_l\} \text{with} \genfrac{}{}{0pt}{1}{\alpha_k \in z_K\backslash \{\alpha\}}{-\beta_l \in -z_L\backslash \{- \beta\}} \\
    0  &\quad \text{otherwise},
  \end{cases}
\end{equation}
with
\begin{equation}
    z(x) = \frac{1}{1-e^{-x}}.
\end{equation}

To see an example of (\ref{J_star_a_b_def}), let \(z_K=\{\alpha_1, \alpha_2,\alpha_3\}\) and \( -z_L=\{-\beta_1,-\beta_2\}\), which means that \(|K|=3\) and \(|L|=2\). Let \(S=\{\alpha_1\}\) and \(T=\{-\beta_1\}\). Then we consider the terms of the sum in (\ref{J_star_a_b_def}), that is different arrangements \(W\)s of the elements of sets of \(z_K\backslash\{\alpha\}\) and \(-z_L\backslash\{-\beta\}\) into sets of one element or two elements from \(z_K\backslash\{\alpha\}\) and \(-z_L\backslash\{-\beta\}\) respectively:

\begin{equation}
    \label{W_arrangement_example}
            (z_K\backslash \{\alpha_1\}) + (-z_L \backslash \{-\beta_1\}) = \{\alpha_2,\alpha_3\} + \{-\beta_2\} 
\end{equation}      
    Thus the partitions into subsets $W_1 + \dots + W_Y $ could be one of the following three options: 
    
    \begin{equation}
       \begin{aligned}
         &\{\alpha_2\} + \{\alpha_3\} + \{-\beta_2\} \\
        &  \{\alpha_2,-\beta_2\} + \{\alpha_3\} \\
         & \{\alpha_3,-\beta_2\} + \{\alpha_2\}
    \end{aligned}
\end{equation}
and so 
\begin{equation}
\label{H_product_example_1}
\begin{aligned}
            &H_{\{\alpha_1\},\{-\beta_1\}} (\{\alpha_2\}) \times H_{\{\alpha_1\},\{-\beta_1\}} (\{\alpha_3\})\times H_{\{\alpha_1\},\{-\beta_1\}}(\{-\beta_2\})\\ \\ &=\left(\frac{z'}{z}(\alpha_2 - \alpha_1) - \frac{z'}{z}(\alpha_2- \beta_1)\right) \\ 
            & \times \left(\frac{z'}{z}(\alpha_3 - \alpha_1) - \frac{z'}{z}(\alpha_3- \beta_1)\right) \\ 
            &\times \left(\frac{z'}{z}(-\beta_2 + \beta_1) - \frac{z'}{z}(-\beta_2 + \alpha_1)\right), \\
\end{aligned}
\end{equation}

\begin{equation}
\label{H_product_example_2}
\begin{aligned}
            H_{\{\alpha_1\},\{-\beta_1\}}(\{\alpha_2,-\beta_2\}) \times H_{\{\alpha_1\},\{-\beta_1\}}(\{\alpha_3\}) &=\left(\left(\frac{z'}{z}\right)'(\alpha_2-\beta_2)\right) \\
            &\times \left(\frac{z'}{z}(\alpha_3 - \alpha_1) - \frac{z'}{z}(\alpha_3 - \beta_1)\right),
\end{aligned}
\end{equation}

\begin{equation}
\label{H_product_example_3}
\begin{aligned}
            H_{\{\alpha_1\},\{-\beta_1\}}(\{\alpha_3,-\beta_2\}) \times H_{\{\alpha_1\},\{-\beta_1\}}(\{\alpha_2\}) &=\left(\left(\frac{z'}{z}\right)'(\alpha_3-\beta_2)\right) \\
            & \times \left(\frac{z'}{z}(\alpha_2 - \alpha_1) - \frac{z'}{z}(\alpha_2 - \beta_1)\right).
\end{aligned}
\end{equation}

From here on we are going to use the set \(K\) as the indexing sum of \(z_K\), \(L\) instead of \(z_L\) and so on for convenience. In a more general choice of sets \(z_K\) and \(-z_L\) we have

\begin{equation}
    \label{J_star_2_generic}
    \begin{aligned}
    J_2^{*}(z_K,-z_L) &= J_{\emptyset,\emptyset}^* + \sum_{\genfrac{}{}{0pt}{1}{\{z_k\} \subset z_K}{\{-z_l\}{\subset -z_L}}} J_{z_k,-z_l}^* \\ \\
    &= 
    \sum_{\genfrac{}{}{0pt}{1}{(K:L)}{|K|=|L|}} \left(\prod_{j=1}^{|K+L|/2} H_{\emptyset,\emptyset} (z_{k_j},-z_{l_j}) \right)
    \\ \\
    &+ \sum_{k \in K,\, l \in L} e^{-N(z_k - z_l)}z(z_k - z_l)z(-z_k+z_l) \\ \\
    &\times \sum_{\genfrac{}{}{0pt}{1}{R\subset K \backslash\{k\},\, Q \subset L \backslash \{l\}}{|R| = |Q|}}
    \left(\prod_{s \in R^c\backslash\{k\}} H_{z_k,-z_l}(z_s)\right)\left(\prod_{t\in Q^c\backslash\{l\}} H_{z_k,-z_l}(-z_t)\right) \\ \\
    &\times \left(\sum_{(R:Q)} \prod_{j=1}^{|R+Q|/2} H_{z_k,-z_l}(z_{r_j},-z_{q_j})\right),
    \end{aligned}
\end{equation}
where \(R^c\) is the complement of \(R\) in \(K\), and similarly for \(Q^c\) and 

\begin{equation}
    \label{H_generic_as_z}
    \begin{aligned}
    H_{\emptyset,\emptyset}(z_k,-z_l) &= \left(\frac{z'}{z}\right)'(z_k-z_l), \\ \\
    H_{z_k,-z_l}(z_s)&=\frac{z'}{z}(z_s-z_k) - \frac{z'}{z}(z_s-z_l), \\ \\
    H_{z_k,-z_l}(-z_t)&=\frac{z'}{z}(-z_t+z_l) - \frac{z'}{z}(-z_t+z_k), \\ \\
    H_{z_k,-z_l}(z_r,-z_q)&=H_{\emptyset,\emptyset}(z_r,-z_q).
    \end{aligned}
\end{equation}

We are using the symbol \((K:L)\) under the sum to denote a sum over all possible pairings of elements indexed by \(K\) with elements indexed by \(L\). Similar notation appears in \cite{kn:rub01}. Note that the first term in \(J_2^*\), \(J_{\emptyset,\emptyset}^*\), only exists when the two sets are of equal size. In the second term, we are essentially taking one element (indexed by \(k\) and \(l\)) away from each of the sets \(K\) and \(L\) and splitting the rest into two sets of equal size (\(R\) and \(Q\)) elements of which can be paired up (corresponding to the innermost sum over (\(R:Q\))) and anything that is left after the pairing is a stand alone product over those elements indexed by \(R^c\) and \(Q^c\). Substituting (\ref{H_generic_as_z}) into (\ref{J_star_2_generic}), expanding the products over \(R^c\) and \(Q^c\) 

\begin{equation}
    \label{prod_expansion}
    \begin{aligned}
    &\left(\prod_{s \in R^c\backslash\{k\}} H_{z_k,-z_l}(z_s)\right)\left(\prod_{t\in Q^c\backslash\{l\}} H_{z_k,-z_l}(-z_t)\right) \\\\
    &= \left(\prod_{s \in R^c\backslash\{k\}} \frac{z'}{z}(z_s-z_k) - \frac{z'}{z}(z_s-z_l)\right)\left(\prod_{t\in Q^c\backslash\{l\}} \frac{z'}{z}(-z_t+z_l) - \frac{z'}{z}(-z_t+z_k)\right) \\\\
    &= \sum_{R_1 \subsetneq R^c\backslash\{k\},\, Q_1 \subsetneq Q^c\backslash\{l\}}\left(\prod_{s \in R_1}\frac{z'}{z}(z_s-z_k)\right)
    \left(\prod_{s \in R_1^c\backslash\{k\}} -\frac{z'}{z}(z_s-z_l)
    \right) \\ \\
    &\times \left(\prod_{t \in Q_1}\frac{z'}{z}(-z_t+z_l)\right)
    \left(\prod_{t \in Q_1^c\backslash\{l\}} -\frac{z'}{z}(-z_t+z_k)\right),
    \end{aligned}
\end{equation}
where we are multiplying out the brackets and $R_1$ and $Q_1$ keep track of which terms in the result come from the first term in the bracket.  We have

\begin{equation}
    \label{J_star_2_generic_with_z}
    \begin{aligned}
    J_2^{*}(z_K,-z_L) &= 
    \sum_{\genfrac{}{}{0pt}{1}{(K:L)}{|K|=|L|}} \left(\prod_{j=1}^{|K+L|/2} \left(\frac{z'}{z}\right)' (z_{k_j}-z_{l_j}) \right)
    \\ \\
    &+  \sum_{\genfrac{}{}{0pt}{1}{R\subsetneq K,\, Q \subsetneq L}{|R| = |Q|}} \left(\sum_{(R:Q)} \prod_{j=1}^{|R+Q|/2} \left(\frac{z'}{z}\right)' (z_{r_j} - z_{q_j})\right)  \\ \\
    &\times 
     \sum_{R_1 \subsetneq R^c,\, Q_1 \subsetneq Q^c} \, \sum_{k \in R_1^c,\, l \in Q_1^c}  e^{-N(z_k - z_l)}z(z_k - z_l)z(-z_k+z_l) \\ \\ 
     & \times \left(\prod_{s \in R_1}\frac{z'}{z}(z_s-z_k)\right)
    \left(\prod_{s \in R_1^c\backslash\{k\}} -\frac{z'}{z}(z_s-z_l)
    \right) \\ \\
    & \times  \left(\prod_{t \in Q_1}\frac{z'}{z}(-z_t+z_l)\right)
    \left(\prod_{t \in Q_1^c\backslash\{l\}} -\frac{z'}{z}(-z_t+z_k)\right),
    \end{aligned}
\end{equation}
where the first line corresponds to \(J_{\emptyset,\emptyset}^*\) and lines 2 to 5 correspond to \(\sum J_{z_k,-z_l}^*\).

In the above we are simply rearranging the order of sums. In lines 2 to 5, instead of starting by picking the special indices $k$ and $l$, we start by splitting the sets \(K\) and \(L\) into subsets \(R\) and its complement in \(K\), \(R^c\), and \(Q\) and its complement in \(L\), \(Q^c\), where \(R\) and \(Q\) are equal in size and must be proper subsets. We sum over all possible choices of such subsets and all possible choices of pairings. Next we split what is left of each set \(R^c\) (and \(Q^c\)) into two complementing sets \(R_1\) (\(Q_1\)) and its complement \(R_1^c\) (\(Q_1^c\)) in \(R^c\) ($Q^c$),  and sum over all possible choices of such sets. Then we pick an element \(z_k\) (\(-z_l\)) indexed by \(R_1^c\) (\(Q_1^c\)) and any other element of $R^c$ ($Q^c$) is then either in a product over \(R_1\) (\(Q_1\)) or \(R_1^c\backslash\{k\}\) (\(Q_1^c\backslash\{l\}\)) paired with \(z_k\) or with \(-z_l\). We obtain the two products over \(R_1\) and \(R_1^c\) by multiplying out the product over \(R^c\) in (\ref{J_star_2_generic}). We do the same procedure for \(Q^c\). The next step is to slightly spread out and order the contours for variables with indices in \(R^c = R_1 \cup R_1^c\) so that \(\delta_1<\delta_2<\dots < \delta_{|R^c|}\) are in the positive half of the complex plane (i.e. \(Re(z)>0\)), and \(Q^c = Q_1 \cup Q_1^c\) so that \(-\delta_{1}>-\delta_{2}>\dots>-\delta_{|Q^c|}\) are in the negative half of the complex plane as in Figure \ref{contours}. Note that the indices in \(R^c\) 
are not necessarily \((1, \dots, |R^c|)\) but we can always map them to this set. This spreading of contours can be done even though there appear to be poles when $z_s=z_k$ or $z_t=z_l$,  but these poles cancel as can be seen in the example below.

Consider first \(k=1\) and \(s=2\), then

\begin{align}
    \label{residue_cancel_example_1}
\underset{z_2 = z_1}{\text{Res}} & \Bigg[f\left(\frac{Niz_1}{2\pi},\frac{Niz_2}{2\pi},\dots,\frac{Niz_n}{2\pi}\right) h_1\left(\frac{iz_1}{\mathcal{T}}\right)
h_l\left(\frac{iz_l}{\mathcal{T}}\right) e^{-N(z_1-z_l)}z(z_1-z_l)z(-z_1+z_l) \notag \\
&\times \prod_{s' \in R_1\backslash\{1,2\}} \frac{h_{s'}\left(\frac{iz_{s'}}{\mathcal{T}}\right)}{-(z_{s'}-z_1)}\times \frac{h_2\left(\frac{iz_2}{\mathcal{T}}\right)}{-(z_2-z_1)} \notag\\
&\times \prod_{s' \in R_1^c} \frac{-h_{s'}\left(\frac{iz_{s'}}{\mathcal{T}}\right)}{-(z_{s'}-z_l)}
\times \prod_{t' \in Q_1\backslash\{l\}} \frac{h_{t'}\left(\frac{iz_{t'}}{\mathcal{T}}\right)}{-(z_{t'}-z_l)}
\times \prod_{t' \in Q_1^c} \frac{-h_{t'}\left(\frac{iz_{t'}}{\mathcal{T}}\right)}{-(-z_{t'}+z_1)} \Bigg]  \notag\\
&= - f\left(\frac{Niz_1}{2\pi},\frac{Niz_1}{2\pi},\dots,\frac{Niz_n}{2\pi}\right) h_1\left(\frac{iz_1}{\mathcal{T}}\right)
h_l\left(\frac{iz_l}{\mathcal{T}}\right) e^{-N(z_1-z_l)}z(z_1-z_l)z(-z_1+z_l) \notag\\
&\times \prod_{s' \in R_1\backslash\{1,2\}} \frac{h_{s'}\left(\frac{iz_{s'}}{\mathcal{T}}\right)}{-(z_{s'}-z_1)}\times h_2\left(\frac{iz_1}{\mathcal{T}}\right) \notag\\
&\times \prod_{s' \in R_1^c} \frac{-h_{s'}\left(\frac{iz_{s'}}{\mathcal{T}}\right)}{-(z_{s'}-z_l)}
\times \prod_{t' \in Q_1\backslash\{l\}} \frac{h_{t'}\left(\frac{iz_{t'}}{\mathcal{T}}\right)}{-(z_{t'}-z_l)}
\times \prod_{t' \in Q_1^c} \frac{-h_t\left(\frac{iz_{t'}}{\mathcal{T}}\right)}{-(-z_{t'}+z_1)},
\end{align}

now for \(k=2\) and \(s=1\), we have

\begin{equation}
    \label{residue_cancel_example_2}
\begin{aligned}
\underset{z_2 = z_1}{\text{Res}} & \Bigg[f\left(\frac{Niz_1}{2\pi},\frac{Niz_2}{2\pi},\dots,\frac{Niz_n}{2\pi}\right) h_2\left(\frac{iz_2}{\mathcal{T}}\right)
h_l\left(\frac{iz_l}{\mathcal{T}}\right) e^{-N(z_2-z_l)}z(z_2-z_l)z(-z_2+z_l) \\
&\times \prod_{s' \in R_1 \backslash\{1,2\}} \frac{h_{s'}\left(\frac{iz_{s'}}{T}\right)}{-(z_{s'}-z_2)}\times \frac{h_1\left(\frac{iz_1}{\mathcal{T}}\right)}{-(z_1-z_2)} \\
&\times \prod_{s' \in R_1^c} \frac{-h_{s'}\left(\frac{iz_{s'}}{\mathcal{T}}\right)}{-(z_{s'}-z_l)}
\times \prod_{t' \in Q_1\backslash\{l\}} \frac{h_{t'}\left(\frac{iz_{t'}}{\mathcal{T}}\right)}{-(z_{t'}-z_l)}
\times \prod_{t' \in Q_1^c} \frac{-h_{t'}\left(\frac{iz_{t'}}{\mathcal{T}}\right)}{-(-z_{t'}+z_2)} \Bigg] \\
&=f\left(\frac{Niz_1}{2\pi},\frac{Niz_1}{2\pi},\dots,\frac{Niz_n}{2\pi}\right) h_2\left(\frac{iz_1}{\mathcal{T}}\right)
h_l\left(\frac{iz_l}{\mathcal{T}}\right) e^{-N(z_1-z_l)}z(z_1-z_l)z(-z_1+z_l) \\
&\times \prod_{s' \in R_1 \backslash\{1,2\}} \frac{h_{s'}\left(\frac{iz_{s'}}{\mathcal{T}}\right)}{-(z_{s'}-z_1)}\times h_1\left(\frac{iz_1}{\mathcal{T}}\right) \\
&\times \prod_{s' \in R_1^c} \frac{-h_{s'}\left(\frac{iz_{s'}}{\mathcal{T}}\right)}{-(z_{s'}-z_l)}
\times \prod_{t' \in Q_1\backslash\{l\}} \frac{h_{t'}\left(\frac{iz_{t'}}{\mathcal{T}}\right)}{-(z_{t'}-z_l)}
\times \prod_{t' \in Q_1^c} \frac{-h_{t'}\left(\frac{iz_{t'}}{\mathcal{T}}\right)}{-(-z_{t'}+z_1)}.
\end{aligned}
\end{equation}

These two residues differ only in the sign. Same argument can be made for when \(z_t=z_l\).

We break down our expression
\begin{equation}
\label{our_expression_to_break_down}
\begin{aligned}
    \int_{U(N)} &\sum_{j_1,\dots,j_n} F(\theta_{j_1},\dots,\theta_{j_n}) dX 
    = \frac{1}{(2\pi i)^n}\sum_{K+L+M  = \{1,\dots,n\}} (-1)^{|L|+|M|}N^{|M|} \\ & \times \int_{(\delta)^{|K|}} \int_{(-\delta)^{|L|}} \int_{(0)^{|M|}} \Bigg(J_{\emptyset,\emptyset}^* + \sum_{\genfrac{}{}{0pt}{1}{\{z_k\} \subset z_K}{\{-z_l\}{\subset -z_L}}} J_{z_k,-z_l}^*\Bigg) \prod_{j=1}^{n}h \left(\frac{iz_j}{\mathcal{T}}\right)\\
    &\times \int_{\mathbb{R}^n} \Phi(\xi_1,\dots,\xi_n)\delta\left(\sum_{j =  1}^n \xi_j\right)
    e^{N\left(\sum_{j=1}^n z_j \xi_j\right)}d\xi \, dz_1\dots dz_n
\end{aligned}
\end{equation}
 into two separate integrals for clarity. One we call $I_{0,0}$, which has $J_{\emptyset,\emptyset}^*$ in the integrand, and which is evaluated in \cite{kn:consna14} (see Section 5 of that paper):
\begin{align}
    \label{I_0_0_def}
    I_{0,0}
    &=\frac{1}{(2 \pi i)^n} \sum_{K+L+M=\{1,\dots,n\}}(-1)^{|L|+|M|}N^{|M|} \notag\\ \notag\\
    &\times \int_{\mathbb{R}^n} \Phi(\xi_1,\dots,\xi_n)\delta(\xi_1+\dots + \xi_n) 
 \prod_{m\in M} \int_{(0)} h_m\left(\frac{i z_m}{\mathcal{T}}\right)e^{N z_m \xi_m} dz_m \notag\\ \notag\\
    &\times \sum_{\genfrac{}{}{0pt}{1}{(K:L)}{|K|=|L|}} \prod_{j=1}^{|K+L|/2} 
    \int_{(\delta)}\int_{(-\delta)} h_{k_j}\left(\frac{i z_{k_j}}{\mathcal{T}}\right) e^{N z_{k_j} \xi_{k_j}} h_{l_j}\left(\frac{i z_{l_j}}{\mathcal{T}}\right) e^{N z_{l_j} \xi_{l_j}} 
    \left(\frac{z'}{z}\right)' (z_{k_j}-z_{l_j}) \, dz_{l_j} dz_{k_j} \notag\\ \notag\\
    &d\xi_1 \dots d\xi_n.
\end{align}

We will focus on the second integral
\begin{align}
    \label{I_1_1_def} 
    I&_{1,1}=\frac{1}{(2\pi i)^n} \sum_{K+L+M=\{1,\dots,n\}}(-1)^{|L|+|M|}N^{|M|} \notag \\ \notag \\
    &\times \int_{(\delta)^{|K|}}\int_{(-\delta)^{|L|}}\int_{(0)^{|M|}}\Bigg( \sum_{\genfrac{}{}{0pt}{1}{\{z_k\} \subset z_K}{\{-z_l\}{\subset -z_L}}} J_{z_k,-z_l}^*(z_K,-z_L)\Bigg)\prod_{j=1}^{n} h_j\left(\frac{iz_j}{\mathcal{T}}\right) \notag\\ \notag\\
     &\times \int_{\mathbb{R}^n}\Phi(\xi_1,\dots,\xi_n)\delta(\xi_1 + \dots + \xi_n)
     e^{(N z_1\xi_1+\dots+N z_n\xi_n)} \, d\xi_1\dots d\xi_n \, dz_1\dots dz_n, \notag\\ \notag\\
     \displaybreak[0] 
    &=\frac{1}{(2 \pi i)^n} \sum_{K+L+M=\{1,\dots,n\}}(-1)^{|L|+|M|}N^{|M|} \notag\\ \notag\\
  &\times \int_{\mathbb{R}^n} \Phi(\xi_1,\dots,\xi_n)\delta(\xi_1+\dots + \xi_n) 
 \prod_{m\in M} \int_{(0)} h_m\left(\frac{i z_m}{\mathcal{T}}\right)e^{N z_m \xi_m} dz_m \notag\\ \notag\\
    & \sum_{\genfrac{}{}{0pt}{1}{R \subsetneq K, Q \subsetneq L}{|R|=|Q|}} \Bigg[ \sum_{(R:Q)} \prod_{j=1}^{|R+Q|/2} \int_{(\delta)} \int_{(-\delta)} h_{r_j}\left(\frac{i z_{r_j}}{T}\right) e^{N z_{r_j}\xi_{r_j}} h_{q_j}\left(\frac{i z_{q_j}}{\mathcal{T}}\right) e^{N z_{q_j}\xi_{q_j}} \left(\frac{z'}{z}\right)'(z_{r_j} - z_{q_j}) dz_{r_j} dz_{q_j}\Bigg] \notag\\ \notag\\
    &\times \Bigg[\sum_{R_1 \subsetneq R^c,\, Q_1 \subsetneq Q^c} \, \sum_{k\in R_1^c,\, l\in Q_1^c}  \int_{(\delta_k)} \int_{(-\delta_l)} h_{k}\left(\frac{i z_k}{\mathcal{T}}\right) e^{N z_k \xi_k} h_l\left(\frac{i z_l}{\mathcal{T}}\right) e^{N z_l \xi_l} e^{-N(z_k-z_l)}z(z_k-z_l)z(-z_k+z_l) \notag\\ \notag\\ 
    &\times  
     \left(\prod_{s \in R_1} \int_{(\delta_s)} h_s\left(\frac{i z_s}{\mathcal{T}}\right) e^{N z_s \xi_s} \frac{z'}{z}(z_s-z_k) dz_s\right) 
     \left(\prod_{s \in R_1^c \backslash \{k\}} \int_{(\delta_s)} -h_s\left(\frac{i z_s}{\mathcal{T}}\right) e^{N z_s \xi_s} \frac{z'}{z}(z_s-z_l) dz_s \right) 
     \notag\\ \notag \\ 
    & \times 
     \left(\prod_{t \in Q_1} \int_{(-\delta_t)} h_t\left(\frac{i z_t}{\mathcal{T}}\right) e^{N z_t \xi_t} \frac{z'}{z}(-z_t+z_l) dz_t\right)\left(\prod_{t \in Q_1^c\backslash \{l\} } \int_{(-\delta_t)} -h_t\left(\frac{i z_t}{\mathcal{T}}\right) e^{N z_t \xi_t} \frac{z'}{z}(-z_t+z_k) dz_t \right) \notag\\ \notag\\
    &  dz_l dz_k \Bigg] \, d\xi_1 \dots d\xi_n.
\end{align}

 In this expression there appear several different types of integrals which we can summarise in the following lemmas:

\begin{lemma}
\label{lemma_1}
\begin{equation}
\label{thm:lmm_over_M}
    \begin{aligned}
    \prod_{m \in M} \int_{(0)}  h_m\left(\frac{iz_m}{\mathcal{T}}\right) e^{Nz_m \xi_m} dz_m  &= \prod_{m \in M} \left(-2\pi i \mathcal{T} g_m(N\mathcal{T}\xi_m)\right).
    \end{aligned}
\end{equation}

\end{lemma}

\begin{lemma}
\label{lemma_2}
Suppose \(R=\{r_1,r_2,\ldots,r_{|R|}\}\) and \(Q=\{q_1,q_2,\ldots,q_{|Q|}\}\) are two disjoint indexing sets of equal size then
\begin{equation}
\label{lmm_over_R_Q}
    \begin{aligned}
    \prod_{j=1}^{|R+Q|/2}&\int_{(\delta)}\int_{(-\delta)} h_{r_j}\left(\frac{i z_{r_j}}{\mathcal{T}}\right) e^{N z_{r_j} \xi_{r_j}} h_{q_j}\left(\frac{i z_{q_j}}{\mathcal{T}}\right) e^{N z_{q_j} \xi_{q_j}}
    \left(\frac{z'}{z}\right)' (z_{r_j}-z_{q_j}) dz_{q_j} dz_{r_j}  \\ 
    &= (1+O(1/\mathcal{T})) (2 \pi i)^{|R\cup Q|}\prod_{j=1}^{|R+Q|/2} \left(N\mathcal{T}\xi_{q_j}\right) \\ 
    &\times \int_{-\infty}^{\infty} g_{q_j}(u_{q_j})g_{r_j}(N\mathcal{T}(\xi_{r_j} + \xi_{q_j})-u_{q_j})du_{q_j},
    \end{aligned}
\end{equation}
with the condition that \(\xi_{q_j}<0\) for \(q_j \in Q\), otherwise the expression is zero.
\end{lemma}

These two integrals appear in \cite{kn:consna14} already but we show the proof for the sake of completeness. The only integral which is new is for the extension \(\sum_{j}|\xi|<2\) is the following

\begin{lemma}
\label{lemma_3}
Suppose \(R_1\), \(R_1^c\), \(Q_1\), \(Q_1^c\) are four disjoint indexing sets and \(k \in R_1^c\)  and \(l \in Q_1^c\). Let \(R_1^{<k}\) denote the set of indices from \(R_1\) that are smaller than \(k\), and similarly \(R_1^{>k}\) be the set of indices that are larger than k, and the same concept applies to $Q$ and $l$.  Then
\begin{equation}
\label{lmm_over_complements}
    \begin{aligned}
& \int_{(\delta_k)} \int_{(-\delta_l)} h_{k}\left(\frac{i z_k}{\mathcal{T}}\right) e^{N z_k \xi_k} h_l\left(\frac{i z_l}{\mathcal{T}}\right) e^{N z_l \xi_l} e^{-N(z_k-z_l)}z(z_k-z_l)z(-z_k+z_l) \\
    &\times  
     \left(\prod_{s \in R_1} \int_{(\delta_s)} h_s\left(\frac{i z_s}{\mathcal{T}}\right) e^{N z_s \xi_s} \frac{z'}{z}(z_s-z_k) ds\right)  \left(\prod_{s \in R_1^c \backslash\{k\}} \int_{(\delta_s)} -h_s\left(\frac{i z_s}{\mathcal{T}}\right) e^{N z_s \xi_s} \frac{z'}{z}(z_s-z_l) dz_s \right) 
     \\
    & \times 
     \left(\prod_{t \in Q_1} \int_{(-\delta_t)} h_t\left(\frac{i z_t}{\mathcal{T}}\right) e^{N z_t \xi_t} \frac{z'}{z}(-z_t+z_l) dz_t\right)\left(\prod_{t \in Q_1^c\backslash \{l\} } \int_{(-\delta_t)} -h_t\left(\frac{i z_t}{\mathcal{T}}\right) e^{N z_t \xi_t} \frac{z'}{z}(-z_t+z_k) dz_t \right)
     \\& dz_l dz_k \\ 
     &= (2\pi i)^{|R^c| + |Q^c|} (-1)^{|R_1^{>k}| + |Q_1^{>l}|}\left(1+O(1/\mathcal{T})\right)N\mathcal{T}\left(\xi_l + \sum_{j \in R_1^c\backslash\{k\} \cup Q_1^{<l}} \xi_j + \sum_{j \in Q_1^{>l}} \xi_j +1\right) \\ 
&\times \int_{-\infty}^{\infty} \dots \int_{-\infty}^{\infty} \prod_{j \in R^c\backslash\{k\} \cup Q^c} g_j(u_j) \\
&\times g_k\Big(N\mathcal{T}\Big(\sum_{j \in R_1^c\backslash \{k\} \cup R_1^{>k} \cup Q_1^{<l}} \xi_j
+ \sum_{j \in Q_1^{>l} \cup Q_1^c \backslash \{l\}\cup R_1^{<k}} \xi_j + \xi_k+  \xi_l \Big)-\sum_{j \in R^c \backslash\{k\}\cup Q^c } u_j\Big) d\mathbf{u},
    \end{aligned}
\end{equation}
with the conditions that \(\xi_j >0\) for  $j \in R_1^c\backslash \{k\} \cup R_1^{>k} \cup Q_1^{<l}$, and $\xi_j<0$ for $j \in Q_1^{>l} \cup Q_1^c \backslash \{l\}\cup R_1^{<k}$   and that \(\xi_l + \sum_{j \in R_1^c\backslash \{k\}} \xi_j +\sum_{j \in Q_1^{<l}} \xi_j +\sum_{j \in Q_1^{>l}} \xi_j+1 < 0\).  However, $\xi_k$ and $\xi_l$ can be any real number.  If any of these conditions on the $\xi$ variables does not hold, then the expression is zero. 

\end{lemma}

\renewcommand\qedsymbol{$\blacksquare$}
\begin{proof}[Proof of Lemma \ref{lemma_1}]

    We can integrate out the \(z_M\) variables as in \cite{kn:consna14}, expression (48), in the following way

\begin{equation}
\label{lmm_over_M_proof}
    \begin{aligned}
    \int_{(0)^{|M|}} \prod_{m \in M} h_m\left(\frac{iz_m}{\mathcal{T}}\right) e^{Nz_m \xi_m} dz_m &= \prod_{m \in M} \int_{(0)} h_m \left(\frac{iz_m}{\mathcal{T}}\right) e^{Nz_m \xi_m}dz_m \\
    &= \prod_{m \in M} \int_{\mathbb{R}} (-2\pi i\mathcal{T}) h(2\pi z_m) e^{-2\pi i\mathcal{T}Nz_m\xi_m} dz_m \\
    &= \prod_{m \in M} \left(-2\pi i \mathcal{T} g_m(N\mathcal{T}\xi_m)\right),
    \end{aligned}
\end{equation}
where in line two we made a change of variables \(z_m \to -2\pi i\mathcal{T}z_m\), in line three we used the fact that \(h(2\pi z) = \int_{\mathbb{R}}g(\xi)e^{2\pi i z \xi}d\xi\) and therefore \(h(2 \pi z)\) is the inverse Fourier transform of \(g(\xi)\).
\end{proof}

\renewcommand\qedsymbol{$\blacksquare$}
\begin{proof}[Proof of Lemma \ref{lemma_3}]

We consider a generic choice of \(k\) and \(l\), meaning that \(z_k\) can lie on any of the contours on the right half of the complex plane and \(z_l\) can lie on any of the contours on the left half of the complex plane (see Figure \ref{contours}). Let us introduce some useful notation that we will use in the rest of the calculation. Consider the set \(R_1\). Let \(R_1^{<k}\) denote the set of indices from \(R_1\) that are smaller than \(k\), and similarly \(R_1^{>k}\) be the set of indices that are larger than k. We remind the reader that the coutours for the \(z\) variables indexed by \(R_1\) lie on the vertical lines spread around \(\delta>0\).
\setlength{\unitlength}{1cm}
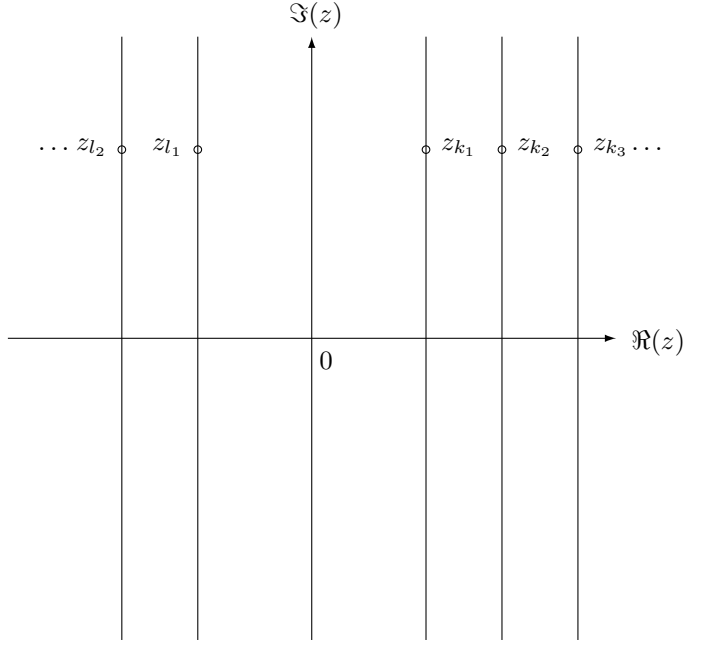
\begin{figure}
\begin{picture}(10,10)

\put(4,5){\vector(1,0){8}}
\put(12.2,4.85){$\Re(z)$}
\put(8,1){\vector(0,1){8}}
\put(7.7,9.2){$\Im(z)$}
\put(8.1,4.6){$0$}

\put(9.5,1){\line(0,1){8}}
\put(9.7,7.5){$z_{k_1}$}
\put(9.5,7.5){\circle{0.1}}
\put(10.5,1){\line(0,1){8}}
\put(10.7,7.5){$z_{k_2}$}
\put(10.5,7.5){\circle{0.1}}
\put(11.5,1){\line(0,1){8}}
\put(11.7,7.5){$z_{k_3}$}
\put(11.5,7.5){\circle{0.1}}
\put(12.2,7.5){$\dots$}

\put(6.5,1){\line(0,1){8}}
\put(5.9,7.5){$z_{l_1}$}
\put(6.5,7.5){\circle{0.1}}
\put(5.5,1){\line(0,1){8}}
\put(4.9,7.5){$z_{l_2}$}
\put(5.5,7.5){\circle{0.1}}
\put(4.4,7.5){$\dots$}

\end{picture}
\caption{Visualisation of the placement of contours. The \(z_{k_1}\) variable lies on the contour (\({\delta}_{1}\)), \(z_{k_2}\) lies on the contour (\({\delta}_{2}\)), and so on. Also \(k_1<k_2 < k_3<\dots\)  and $\delta_1<\delta_2<\cdots$. The variable $z_{l_1}$ is integrated on $(-\delta_1)$, and so forth.}\label{contours}
\end{figure}

It is also worth pointing out that the function \(z'/z(x)\) behaves like \(-1/x\) when \(x\) is small. Now consider the product of integrals $z_s$ with $s$ in \(R_1\) first. Suppose \(\xi_s < 0\), then closing the \((\delta_s)\) contour to the far right will result in zero for any s in \(R_1^{>k}\) since \(z_k\) is left of any such \(z_s\) and \(e^{Nz_s\xi_s}\) will be very small when \(z_s\) has large positive real part and \(\xi_s\) is negative. For any s in \(R_1^{<k}\) we will pick up a single pole resulting in a residue with an extra minus sign as the orientation of the contour is reversed. If \(\xi_s> 0\), then closing the contour to the far left will result in no contribution for any s in \(R_1^{<k}\) because \(z_k\) is right of any such \(z_s\). And any s in \(R_1^{>k}\) we will pick up a simple pole at \(z_s = z_k\). Therefore, we have

\begin{equation}
\label{R_1_int}
\begin{aligned}
\prod_{s \in R_1} \int_{(\delta_s)}h_s\left(\frac{iz_s}{\mathcal{T}}\right)e^{Nz_s\xi_s} \frac{z'}{z}(z_s-z_k)dz_s &= (2\pi i)^{|R_1|}\prod_{s \in R_1} \underset{z_s = z_k}{\text{Res}} 
\left(\frac{h_s\left(\frac{iz_s}{\mathcal{T}}\right)e^{Nz_s\xi_s}}{-(z_s-z_k)}\right) \\
&= (2\pi i)^{|R_1|}\prod_{s \in R_1^{<k}}  h_s\left(\frac{iz_k}{\mathcal{T}}\right)e^{Nz_k\xi_s} \\ 
&\times\prod_{s \in R_1^{>k}} - h_s\left(\frac{iz_k}{\mathcal{T}}\right)e^{Nz_k\xi_s},
\end{aligned}
\end{equation}
with the condition that \(\xi_s <0\) for all \(s \in R_1^{<k}\) and \(\xi_s >0\) for all \(s \in R_1^{>k}\). If the condition does not hold for any one of the \(\xi_s\) then whole expression is \(0\). That is 
\begin{equation}
\label{R_1_int_zero}
\begin{aligned}
\prod_{s \in R_1\backslash\{k\}} \int_{(\delta_s)}h_s\left(\frac{iz_s}{\mathcal{T}}\right)e^{Nz_s\xi_s} \frac{z'}{z}(z_s-z_k)dz_s &= 0,
\end{aligned}
\end{equation}
if \(\xi_s \geq 0\) for any \(s \in R_1^{<k}\) or \(\xi_s \leq 0\) for any \(s \in R_1^{>k}\).

Similarly for \(s \in R_1^c\backslash\{k\}\), suppose \(\xi_s <0\), then closing the \((\delta_s)\) contour to the far right will result in no contribution as the pole is at $z_l$, which lies in the left half-plane. If \(\xi_s > 0\), then closing the contour to the far left picks up the pole at \(z_s=z_l\):
\begin{equation}
\label{R_1_c_int}
\begin{aligned}
\prod_{s \in R_1^c\backslash\{k\}} \int_{(\delta_s)}-h_u\left(\frac{iz_s}{\mathcal{T}}\right)e^{Nz_s\xi_s} \frac{z'}{z}(z_s-z_{l})dz_s &= (2\pi i)^{|R_1^c\backslash\{k\}|}\prod_{s \in R_1^c} \underset{z_s = z_{l}}{\text{Res}} 
\left(\frac{-h_s\left(\frac{iz_s}{\mathcal{T}}\right)e^{Nz_s\xi_s}}{-(z_s-z_{l})}\right) \\
&= (2\pi i)^{|R_1^c\backslash\{k\}|} \prod_{s \in R_1^c} h_s\left(\frac{iz_{l}}{\mathcal{T}}\right)e^{Nz_{l}\xi_s},
\end{aligned}
\end{equation}
with the condition that \(\xi_s >0\) for all \(s \in R_1^c\backslash\{k\}\). Here again the expression is zero if the condition does not hold for any one of the  \(\xi_s\). 

The procedure is the same for \(t \in Q_1\). Suppose \(\xi_t>0\), then closing the \((-\delta_t)\) contour to the far left results in no contribution for any \(t \in Q_1^{>l}\) and for any \(t \in Q_1^{<l}\) we pick up a simple pole. If \(\xi_t <0\), then closing the contour to the far right results in no contribution for \(t \in Q_1^{<l}\) and picks up a pole at \(z_t=z_l\) with an extra minus sign as the orientation of the contour is reversed for \(t \in Q_1^{>l}\):
\begin{equation}
\label{Q_1_int}
\begin{aligned}
\prod_{t \in Q_1} \int_{(-\delta_t)}h_t\left(\frac{iz_t}{\mathcal{T}}\right)e^{Nz_t\xi_t} \frac{z'}{z}(-z_t+z_{l})dz_t &= (2\pi i )^{|Q_1|}\prod_{t \in Q_1} \underset{z_t = z_{l}}{\text{Res}} 
\left(\frac{h_t\left(\frac{iz_t}{\mathcal{T}}\right)e^{Nz_t\xi_t}}{-(-z_t+z_{l})}\right) \\
&= (2\pi i )^{|Q_1|} \prod_{t \in Q_1^{<l}} h_t\left(\frac{iz_{l}}{\mathcal{T}}\right)e^{Nz_{l}\xi_t}\times \prod_{t \in Q_1^{>l}} -h_t\left(\frac{iz_{l}}{\mathcal{T}}\right)e^{Nz_{l}\xi_t},
\end{aligned}
\end{equation}
with the condition that \(\xi_t > 0\) for \(t \in Q_1^{<l}\) and \(\xi_t <0\) for \(t \in Q_1^{>l}\).  The product is zero if this condition is not met for some $\xi_t$.

Lastly, consider \(t \in  Q_1^c\backslash\{l\}\) and suppose \(\xi_t > 0\), then closing the contour to the far left  result in no contribution. If \(\xi_t < 0 \), then closing the contour to the far right picks up a pole at \(z_t = z_k\) and again an extra minus sign from the clockwise contour enclosing the simple pole. And so we have
\begin{equation}
\label{Q_1_c_int}
\begin{aligned}
\prod_{t \in Q_1^c\backslash\{l\}} \int_{(-\delta_t)}-h_t\left(\frac{iz_t}{\mathcal{T}}\right)e^{Nz_t\xi_t} \frac{z'}{z}(-z_t+z_k)dz_t &= (2\pi i)^{|Q_1^c|} \prod_{t \in  Q_1^c\backslash\{l\}} \underset{z_t = z_k}{\text{Res}} 
\left(\frac{-h_t\left(\frac{iz_t}{\mathcal{T}}\right)e^{Nz_t\xi_t}}{-(-z_t+z_k)}\right) \\
&= (2\pi i)^{|Q_1^c|} \prod_{t \in  Q_1^c\backslash\{l\}} h_t\left(\frac{iz_k}{\mathcal{T}}\right)e^{Nz_k\xi_t}.
\end{aligned}
\end{equation}

Now we deal with the integral with respect to \(z_l\)

\begin{equation}
\label{int_combined}
\begin{aligned}
(2\pi i)&^{|R^c| + |Q^c| -2} (-1)^{|R_1^{>k}| + |Q_1^{>l}|} \int_{(\delta_k)} \int_{(-\delta_l)}  \prod_{s \in R^c} h_s\left(\frac{iz_k}{\mathcal{T}}\right) e^{Nz_k\left(\xi_k + \sum_{j \in R_1^{>k} }\xi_j +\sum_{j \in R_1^{<k}} \xi_j + \sum_{j \in Q_1^c\backslash\{l\}} \xi_j -1 \right)} \\ \\
&\times \prod_{t \in Q^c} h_t\left(\frac{iz_l}{\mathcal{T}}\right)e^{Nz_l\left(\xi_l + \sum_{j \in R_1^c\backslash\{k\}} \xi_j +\sum_{j \in Q_1^{<l}} \xi_j + \sum_{j \in Q_1^{>l}} \xi_j +1\right)}
\, z(z_k - z_l)z(-z_k+z_l) dz_l dz_k. 
\end{aligned}
\end{equation}
Now suppose \(\xi_l + \sum_{j \in R_1^c\backslash\{k\}} \xi_j +\sum_{j \in Q_1^{<l}} \xi_j + \sum_{j \in Q_1^{>l}} \xi_j +1\) is positive, then moving the \((-\delta_l)\) contour to the far left will give us no contribution as there are no poles. But if the sum of \(\xi\)s is negative and we move our contour to the far right, we pick up a double pole at \(z_l = z_k\)

\begin{equation}
\label{z_l_residue_with_small_term}
\begin{aligned}
\underset{z_l = z_k}{\text{Res}} &\left[\frac{1}{-(z_l -z_k)^2}\prod_{t \in Q^c} h_t\left(\frac{iz_l}{\mathcal{T}}\right)e^{Nz_l\left(\xi_l + \sum_{j \in R_1^c\backslash\{k\}} \xi_j +\sum_{j \in Q_1^{<l}} \xi_j + \sum_{j \in Q_1^{>l}} \xi_j +1\right)}  \right] \\ \\
&= \lim\limits_{z_l \to z_k} \frac{d}{dz_l} \left[-\prod_{t \in Q^c} h_t\left(\frac{iz_l}{\mathcal{T}}\right)e^{Nz_l\left(\xi_l + \sum_{j \in R_1^c\backslash\{k\}} \xi_j +\sum_{j \in Q_1^{<l}} \xi_j + \sum_{j \in Q_1^{>l}} \xi_j +1\right)}\right] \\ \\
&= -N\left(\xi_l + \sum_{j \in R_1^c\backslash\{k\} }\xi_j +\sum_{j \in Q_1^{<l}} \xi_j + \sum_{j \in Q_1^{>l}} \xi_j +1\right) e^{Nz_k\left(\xi_l + \sum_{j \in R_1^c\backslash\{k\}} \xi_j +\sum_{j \in Q_1^{<l}} \xi_j + \sum_{j \in Q_1^{>l}} \xi_j +1\right)} \\ \\
&\times \prod_{t \in Q^c} h_t\left(\frac{iz_k}{\mathcal{T}}\right) - \lim\limits_{z_l \to z_k} \frac{d}{dz_l} \left[ \prod_{t \in Q^c} h_t\left(\frac{iz_l}{\mathcal{T}}\right) \right] e^{Nz_k\left(\xi_l + \sum_{j \in R_1^c\backslash\{k\} } \xi_j +\sum_{j \in Q_1^{<l}} \xi_j + \sum_{j \in Q_1^{>l}} \xi_j +1\right)},
\end{aligned}
\end{equation}
but for large \(\mathcal{T}\), the second term contains a factor of \(1/\mathcal{T}\) after differentiation because of the \(1/\mathcal{T}\) in the argument and the fact that the Fourier transform of $h_t$ is compactly supported. 

Therefore we have

\begin{equation}
\label{z_l_residue}
\begin{aligned}
\underset{z_l = z_k}{\text{Res}} &\left[\frac{1}{-(z_l -z_k)^2}\prod_{t \in Q^c} h_t\left(\frac{iz_l}{\mathcal{T}}\right)e^{Nz_l\left(\xi_l + \sum_{j \in R_1^c\backslash\{k\} } \xi_j +\sum_{j \in Q_1^{<l}} \xi_j + \sum_{j \in Q_1^{>l}} \xi_j +1\right)}  \right] \\ \\
&= -\left(1+O(1/\mathcal{T})\right)N\left(\xi_l + \sum_{j \in R_1^c\backslash\{k\}  } \xi_j +\sum_{j \in Q_1^{<l}} \xi_j + \sum_{j \in Q_1^{>l}} \xi_j +1\right) \\ \\
&\times e^{Nz_k\left(\xi_l + \sum_{j \in R_1^c \backslash\{k\} } \xi_j +\sum_{j \in Q_1^{<l}} \xi_j + \sum_{j \in Q_1^{>l}} \xi_j +1\right)}\prod_{t \in Q^c} h_t\left(\frac{iz_k}{\mathcal{T}}\right).
\end{aligned}
\end{equation}

Inserting this into expression (\ref{int_combined}), we have

\begin{equation}
\label{int_combined_residue}
\begin{aligned}
-(2\pi i)&^{|R^c| + |Q^c|-1} (-1)^{|R_1^{>k}| + |Q_1^{>l}|}\left(1+O(1/\mathcal{T})\right)N\left(\xi_l + \sum_{j \in R_1^c\backslash\{k\} } \xi_j +\sum_{j \in Q_1^{<l}} \xi_j + \sum_{j \in Q_1^{>l}} \xi_j +1\right) \\ \\
&\times \int_{(\delta_k)}  \prod_{s \in R^c \cup Q^c} h_{s}\left(\frac{iz_k}{\mathcal{T}}\right) \\
&\qquad\qquad \times e^{Nz_k\left(\xi_k + \xi_l + \sum_{j \in R_1^c\backslash\{k\} } \xi_j + \sum_{j \in R_1^{>k}} \xi_j + \sum_{j \in Q_1^{<l}} \xi_j + \sum_{j \in Q_1^{>l}} \xi_j + \sum_{j \in R_1^{<k}} \xi_j + \sum_{j \in Q_1^c\backslash\{l\} } \xi_j \right)} dz_k. 
\end{aligned}
\end{equation}

For the integral over (\(\delta_k\)) above, we will need the generalised version of the convolution theorem. The convolution theorem \cite{fourier} states that for two functions \(f\) and \(g\) with Fourier transforms \(\mathcal{F}\{f\}\) and \(\mathcal{F}\{g\}\) respectively, the following statement is true

\begin{equation}
    \label{conv_thrm_def}
\begin{aligned}
    \mathcal{F}\{f \cdot g\} = \mathcal{F}\{f\} \ast \mathcal{F}\{g\},
\end{aligned}
\end{equation}
where \(f \cdot g\) is the ordinary pointwise product of two functions and the convolution of two functions \(f \ast g\) is defined as
\begin{equation}
    \label{conv_def}
\begin{aligned}
    \{f \ast g\} (x)= \int_{-\infty}^{\infty} g(\tau) h(x - \tau) d\tau = \int_{-\infty}^{\infty} g(x-\tau) h(\tau) d\tau.
\end{aligned}
\end{equation}

Since we want to deal with an integral of the form

\begin{equation}
    \label{prod_over_S}
\begin{aligned}
    \int_{(\delta)} \prod_{j \in S} h_j\left(\frac{iz}{\mathcal{T}}\right)e^{Nz\left(\sum_{j \in S} \xi_j\right)} dz,
\end{aligned}
\end{equation}
where \(S\) is in place of \(R^c \cup Q^c\) with \(n\) elements, \(z_k\) shortened to \(z\) and we will ignore the minus signs with \(\xi\)s for now as that is not important for the following result. Also since there are no poles left, we can move the \((\delta)\) contour to the imaginary axis and we will let \(z \to -2\pi i \mathcal{T}x\)
to transform our integral into an integral over the real line

\begin{equation}
    \label{h_to_h_bar}
\begin{aligned}
    \int_{(\delta)} \prod_{j \in S} h_j\left(\frac{iz}{\mathcal{T}}\right)e^{Nz\left(\sum_{j \in S} \xi_j\right)} dz = (-2\pi i \mathcal{T}) \int_{-\infty}^{\infty} \prod_{j \in S} \hat{h}_j(x) e^{-2\pi ixN\mathcal{T}\left(\sum_{j \in S} \xi_j\right)} dx,
\end{aligned}
\end{equation}
where we denoted \(h(2\pi x)\) by \(\hat{h}(x)\) for convenience. The first step is to take the inverse Fourier transform of the Fourier transform of the product of all but one \(h\) function and combine the exponential that came with the transform with the exponential of \(\xi\)s.

\begin{equation}
    \label{inverse_fourier_of_fourier}
\begin{aligned}
    (-2\pi i \mathcal{T}) \int_{-\infty}^{\infty}\int_{-\infty}^{\infty} \mathcal{F}\Big\{\prod_{j \in S\backslash\{n\}} \hat{h}_j\Big\} (\hat{u}_{n-1})\, \hat{h}_n(x)  \, e^{-2\pi i x\left(N\mathcal{T} \left(\sum \xi_j\right) - \hat{u}_{n-1}\right)} \, d\hat{u}_{n-1}\,dx.
\end{aligned}
\end{equation}
Next we can change the order of integrals and move the integral with respect to \(x\) inside
\begin{equation}
    \label{change_order_of_ints}
\begin{aligned}
    (-2\pi i \mathcal{T}) \int_{-\infty}^{\infty} \mathcal{F}\Big\{\prod_{j \in S\backslash\{n\}} \hat{h}_j\Big\} (\hat{u}_{n-1})\int_{-\infty}^{\infty} \, \hat{h}_n(x)e^{-2\pi i x\left(N\mathcal{T} \left(\sum \xi_j\right) - \hat{u}_{n-1}\right)} dx \,d\hat{u}_{n-1},
\end{aligned}
\end{equation}
identifying the Fourier transform of \(\hat{h}_n\) we get
\begin{equation}
    \label{fourier_transform_identified}
\begin{aligned}
    (-2\pi i \mathcal{T}) \int_{-\infty}^{\infty} \mathcal{F}\Big\{\prod_{j \in S\backslash\{n\}} \hat{h}_j\Big\} (\hat{u}_{n-1}) \, \mathcal{F}\{\hat{h}_n\}\left(N\mathcal{T}\left(\sum \xi_j\right) -\hat{u}_{n-1}\right) \,d\hat{u}_{n-1}.
\end{aligned}
\end{equation}

Now, notice that we can repeat the same steps for the Fourier transform of the product to get
\begin{equation}
    \label{second_iteration}
\begin{aligned}
    \mathcal{F}\Big\{\prod_{j \in S\backslash\{n\}} \hat{h}_j\Big\}(\hat{u}_{n-1}) &= \int_{-\infty}^{\infty} \Big(\prod_{j \in S\backslash\{n\}} \hat{h}_j\Big) \, e^{-2\pi i \hat{u}_{n-1} x} \, dx \\ \\
    &= \int_{-\infty}^{\infty} \mathcal{F}\Big\{\prod_{j\in S\backslash\{n, n-1\}} \hat{h}_j \Big\}(\hat{u}_{n-2}) \mathcal{F}\{\hat{h}_{n-1}\}(\hat{u}_{n-1} - \hat{u}_{n-2}) \, d\hat{u}_{n-2}.
\end{aligned}
\end{equation}
We can repeat the process until we exhaust all the indices in the product over \(S\)
\begin{equation}
    \label{all_iterations}
\begin{aligned}
   \int_{-\infty}^{\infty} \dots &\int_{-\infty}^{\infty} \mathcal{F}\{\hat{h}_1\} (\hat{u}_1) \mathcal{F}\{\hat{h}_2\} (\hat{u}_2 -\hat{u}_1) \dots 
   \mathcal{F}\{\hat{h}_{n-1}\} (\hat{u}_{n-1} -\hat{u}_{n-2}) \\ \\ 
   &\times \mathcal{F}\{\hat{h}_{n}\} (N\mathcal{T}(\sum \xi_j) -\hat{u}_{n-1}) d\hat{u}_1 \dots d \hat{u}_{n-1}.
\end{aligned}
\end{equation}
We can now perform a change of variables \(\hat{u}_1 = u_1, \hat{u}_2 - \hat{u}_1 = u_2\) and so on to get that \(\hat{u}_1 = u_1, \hat{u}_2 = u_2 + u_1, \hat{u}_3 = u_3 + u_2 + u_1\) and so on. The Jacobian for this change of variables is just a lower triangular matrix where every entry is \(1\) and the determinant is obviously \(1\). Therefore, we finally have that
\begin{equation}
    \label{change_of_vars_result}
\begin{aligned}
&\int_{(\delta)} \prod_{j \in S} h_j\left(\frac{iz}{\mathcal{T}}\right)e^{Nz\left(\sum_{j \in S} \xi_j\right)} dz \\ \\
  &= (-2\pi i T) \int_{-\infty}^{\infty} \dots \int_{-\infty}^{\infty} g_1(u_1) \dots g_{n-1}(u_{n-1}) g_n\left(N\mathcal{T}\left(\sum \xi_j\right)-u_{n-1} - u_{n-2} - \dots - u_1\right) du_1 \dots du_{n-1},
\end{aligned}
\end{equation}
where \(g = \mathcal{F}\{\hat{h}\}\). Back to our expression (\ref{int_combined_residue}), let \(h_k\) play the role of \(h_n\) above. That is, we will take \(h_k\) out of the product first and use our result (\ref{change_of_vars_result}) to get

\begin{equation}
    \label{conv_thrm_used}
\begin{aligned}
-(2\pi i)&^{|R^c| + |Q^c|-1} (-1)^{|R_1^{>k}| + |Q_1^{>l}|}\left(1+O(1/\mathcal{T})\right)N\left(\xi_l + \sum_{j \in R_1^c\backslash\{k\}  } \xi_j +\sum_{j \in Q_1^{<l}} \xi_j + \sum_{j \in Q_1^{>l}} \xi_j +1\right) \\ \\
&\times\int_{(\delta_k)}  \prod_{s \in R^c \cup Q^c} h_{s}\left(\frac{iz_k}{\mathcal{T}}\right) \\
&\times e^{Nz_k\left(\xi_k + \xi_l + \sum_{j \in R_1^c\backslash\{k\} } \xi_j + \sum_{j \in R_1^{>k}} \xi_j + \sum_{j \in Q_1^{<l}} \xi_j + \sum_{j \in Q_1^{>l}} \xi_j + \sum_{j \in R_1^{<k}} \xi_j + \sum_{j \in Q_1^c\backslash\{l\} } \xi_j \right)} dz_k \\ \\ 
&= (2\pi i)^{|R^c| + |Q^c|} (-1)^{|R_1^{>k}| + |Q_1^{>l}|}\left(1+O(1/\mathcal{T})\right)N\mathcal{T}\left(\xi_l + \sum_{j \in R_1^c\backslash\{k\} \cup Q_1^{<l}} \xi_j + \sum_{j \in Q_1^{>l}} \xi_j +1\right) \\ \\
&\times \int_{-\infty}^{\infty} \dots \int_{-\infty}^{\infty} \prod_{j \in R^c\backslash\{k\} \cup Q^c} g_j(u_j) \\ \\ 
&\times g_k\Big(N\mathcal{T}\Big(\sum_{j \in R_1^c\backslash\{k\}  \cup R_1^{>k} \cup Q_1^{<l}} \xi_j
+ \sum_{j \in Q_1^{>l} \cup Q_1^c\backslash\{l\}  \cup R_1^{<k}} \xi_j + \xi_k + \xi_l \Big)-\sum_{j \in R^c \backslash\{k\}\cup Q^c } u_j\Big) d\mathbf{u},
\end{aligned}
\end{equation}
with the conditions that \(\xi_j >0\) for  $j \in R_1^c\backslash \{k\} \cup R_1^{>k} \cup Q_1^{<l}$, and $\xi_j<0$ for $j \in Q_1^{>l} \cup Q_1^c \backslash \{l\}\cup R_1^{<k}$   and that \(\xi_l + \sum_{j \in R_1^c\backslash \{k\}} \xi_j +\sum_{j \in Q_1^{<l}} \xi_j +\sum_{j \in Q_1^{>l}} \xi_j+1 < 0\). If any of these conditions don't hold then the whole expression is zero

Notice that we have somewhat artificially kept separate sums over the $\xi$ in different sets.  These could be combined into simpler sums, but when we use this lemma we will treat the positive and negative $\xi$ differently so it is clearer in later steps to keep them separate. 
\end{proof}

\renewcommand\qedsymbol{$\blacksquare$}
\begin{proof}[Proof of Lemma \ref{lemma_2}]  
Now we can tackle the product of double integrals where the \(z\) variables come paired in an argument of the derivative of the logarithmic derivative of \(z(x)\).
\begin{equation}
\label{prod_dble_ints}
    \begin{aligned}
    \prod_{j=1}^{|R+Q|/2}\int_{(\delta)}\int_{(-\delta)} h_{r_j}\left(\frac{i z_{r_j}}{\mathcal{T}}\right) e^{N z_{r_j} \xi_{r_j}} h_{q_j}\left(\frac{i z_{q_j}}{\mathcal{T}}\right) e^{N z_{q_j} \xi_{q_j}}
    \left(\frac{z'}{z}\right)' (z_{r_j}-z_{q_j}) dz_{q_j} dz_{r_j}.
    \end{aligned}
\end{equation}

Note that for small values of \(x\) the function \((\frac{z'}{z})'(x)\) can be expanded into \(1/x^2\) \(+\) a constant \(+\) \(O(x)\). 

Looking at just one factor in the product,  suppose that \(\xi_q >0\), then moving \(-(\delta)\) contour to the far left will result in no contribution as there are no poles. If \(\xi_q <0\) then moving the contour to the right, however, will pick up a double pole at \(z_q = z_r\).

\begin{equation}
\label{double_ints_residue}
    \begin{aligned}
    \underset{z_q = z_r}{\text{Res}} \left[h_q\left(\frac{iz_q}{\mathcal{T}}\right)e^{Nz_q\xi_q}\frac{1}{(z_r-z_q)^2} \right] &= \lim\limits_{z_q \to z_r}
    \frac{d}{dz_q} \left(h_q\left(\frac{iz_q}{\mathcal{T}}\right)e^{Nz_q\xi_q}\right) \\ \\
    &= N\xi_q e^{Nz_r\xi_q}h_q\left(\frac{iz_r}{\mathcal{T}}\right) + \lim\limits_{z_q \to z_r}e^{Nz_q\xi_q}\frac{d}{dz_q} \left(h_q\left(\frac{i z_q}{\mathcal{T}}\right)\right) \\ \\
    &= N\xi_qe^{Nz_r\xi_q}h_q\left(\frac{iz_r}{\mathcal{T}}\right)(1+O(1/\mathcal{T})).
    \end{aligned}
\end{equation}
The next step is to use the convolution theorem once again to get rid of the \(z_r\) variable. So using (\ref{change_of_vars_result})

\begin{equation}
    \label{double_ints_conv_thrm}
    \begin{aligned}
  \int_{(\delta)}& h_r\left(\frac{iz_r}{T}\right)h_q\left(\frac{iz_r}{\mathcal{T}}\right)e^{Nz_r(\xi_r+\xi_q)}dz_r \\ \\
    &=- (2\pi i \mathcal{T})\int_{-\infty}^{\infty} g_q(u_q)g_r(N\mathcal{T}(\xi_r+\xi_q)-u_q)du_q,
    \end{aligned}
\end{equation}
and so our product becomes
\begin{equation}
\label{double_ints_result}
    \begin{aligned}
    \prod_{j=1}^{|R+Q|/2}&\int_{(\delta)}\int_{(-\delta)} h_{r_j}\left(\frac{i z_{r_j}}{\mathcal{T}}\right) e^{N z_{r_j} \xi_{r_j}} h_{q_j}\left(\frac{i z_{q_j}}{\mathcal{T}}\right) e^{N z_{q_j} \xi_{q_j}}
    \left(\frac{z'}{z}\right)' (z_{r_j}-z_{q_j}) dz_{q_j} dz_{r_j} \\ 
    &= (1+O(1/\mathcal{T})) (2 \pi i)^{|R\cup Q|}\prod_{j=1}^{|R+Q|/2} \left(N\mathcal{T}\xi_{q_j}\right) \\ 
    &\times \int_{-\infty}^{\infty} g_{q_j}(u_{q_j})g_{r_j}(N\mathcal{T}(\xi_{r_j} + \xi_{q_j})-u_{q_j})du_{q_j},
    \end{aligned}
\end{equation}
 as long as \(\xi_{q_j} <0\) for all $q_j\in Q$.  If this condition does not hold then the whole expression is zero.  Notice that there was an extra minus sign from the residue because the contour around the double pole was oriented clockwise and that cancels with the minus sign in (\ref{double_ints_conv_thrm}).  
 \end{proof}
\renewcommand\qedsymbol{$\blacksquare$}
\begin{proof}[Proof of Theorem \ref{q_2_thm}]
We can now use all of our lemmas to evaluate the \(I_{1,1}\) integral by applying them in (\ref{I_1_1_def})

\begin{equation}
    \label{I_1_1_lmms_used}
    \begin{aligned}
        I_{1,1} &= (1 + O(1/\mathcal{T}))\sum_{K+L+M = \{1,\dots,n\}} (-1)^{|L|+|M|}N^{|M|}  \int_{\mathbb{R}_\xi^{|M|}}\prod_{m \in M} (-\mathcal{T}g_m(N\mathcal{T}\xi_m)) \\ \\
        & \times \sum_{\genfrac{}{}{0pt}{1}{R \subsetneq K, Q \subsetneq L}{|R|=|Q|}} \sum_{(R:Q)} \int_{\mathbb{R}_{\xi}^{|R|}} \int_{\mathbb{R}_{\xi>0}^{|Q|}} \prod_{q_j \in Q} \left[(-N\mathcal{T}\xi_{q_j}) \int_{\mathbb{R}_{u_{q_j}}} g_{q_j}(u_{q_j}) g_{r_j}(N\mathcal{T}(\xi_{r_j}-\xi_{q_j}) - u_{q_j}) du_{q_j} \right] \\ \\
        &\times
        \sum_{R_1\subsetneq R^c, Q_1 \subsetneq Q^c} \sum_{k \in R_1^c, l \in Q_1^c} (-1)^{|R_1^{>k} \cup Q_1^{>l}|} \int_{\mathbb{R}_{\xi_k}} \int_{\mathbb{R}_{\xi_l}}\int_{\mathbb{R}_{\xi>0}^{|R^c\backslash\{k\} \cup Q^c\backslash\{l\}|}} \\ \\
        &\times N\mathcal{T}\left(\xi_l + 1 + \sum_{j \in R_1^c\backslash\{k\} \cup Q_1^{<l}} \xi_j - \sum_{j \in Q_1^{>l}} \xi_j\right) \\ \\
        &\times \int_{\mathbb{R}_{u}^{|R^c\backslash\{k\} \cup Q^c|}} \prod_{j \in R^c\backslash\{k\} \cup Q^c} g_j(u_j)g_k\left(N\mathcal{T}\left(\sum_{j \in S_1} \xi_j - \sum_{j \in S_2} \xi_j \right) -\sum_{j \in R^c\backslash\{k\} \cup Q^c} u_j\right) d\mathbf{u} \\ \\
        &\times \Phi(\xi_1,\dots,\xi_n) \delta\left(\sum_{j \in \{K \cup L \cup M\}\backslash\{S_2 \cup Q\}} \xi_j -\sum_{j \in S_2 \cup Q} \xi_j \right) d\xi_1 \dots d\xi_n,
    \end{aligned}
\end{equation}
where we denote the integral over the \(\xi\) variables by \(\int_{\mathbb{R}_{\xi}}\) and similarly for the \(u\) variables we use \(\int_{\mathbb{R}_u}\). By \(\int_{\mathbb{R}_{\xi>0}}\) we mean integral over \((0,\infty)\). Note that we have made a change of variables \(\xi_{q_j} \to -\xi_{q_j}\) for $q_i\in Q$ and \(\xi_{j} \to -\xi_{j}\) for   \(j \in R_1^{<k} \cup Q_1^{>l} \cup Q_1^c\backslash\{l\}\) to have all the \(\xi\) integrals with respect to these variables be on the positive real half line. Notice there is also a condition on the integral with respect to \(\xi_l\) which is that \(\xi_l + \sum_{j \in R_1^c\backslash\{k\}} \xi_j +\sum_{j \in Q_1^{<l}} \xi_j +1 < \sum_{j \in Q_1^{>l}} \xi_j\). Sets \(S_1\) and \(S_2\) are \(R_1^c \cup R_1^{>k}\cup Q_1^{<l}\cup\{l\}\) and \(R_1^{<k} \cup Q_1^{>l} \cup Q_1^c\backslash\{l\}\) respectively.

Now we apply the following change of variables

\begin{equation}
    \label{I_1_1_change_of_vars_def}
    \begin{aligned}
        N\mathcal{T}\xi_m &= w_m \quad\text{for } m \in M, \\
        u_{q_j} &= w_{q_j} \quad \text{for } q_j \in Q, \\
        N\mathcal{T}(\xi_{r_j} - \xi_{q_j}) - u_{q_j} &= w_{r_j} \quad \text{for } r_j \in R, \\
        N\mathcal{T}\left(\sum_{j \in S_1} \xi_j  - \sum_{j \in S_2} \xi_j\right) - \sum_{j \in R^c\backslash\{k\} \cup Q^c} u_j &= w_k, \\
        u_j &= w_j \quad \text{for } j \in R^c\backslash\{k\}\cup Q^c.
    \end{aligned}
\end{equation}
From the above we can also deduce that

\begin{equation}
    \label{xi_k_as_sum}
    \begin{aligned}
        \xi_{r_j} &= \xi_{q_j} + \frac{w_{r_j} + w_{q_j}}{N\mathcal{T}} \quad \text{for } r_j \in R \\ \\
        \xi_k &= \sum_{j \in S_2} \xi_j - \sum_{j \in S_1\backslash\{k\} } \xi_j + \sum_{j \in R^c \cup Q^c} \frac{w_j}{N\mathcal{T}}.
    \end{aligned}
\end{equation}

After this change of variables \(I_{1,1}\) becomes

    \begin{align}
        \label{I_1_1_change_of_vars}
     I_{1,1} &= (1 + O(1/\mathcal{T})) \sum_{K+L+M = \{1,\dots,n\}} (-1)^{|L|} \int_{\mathbb{R}_w^{n}} \prod_{j=1}^{n} g_j(w_j)\, \delta\left(\sum_{j=1}^{n} \frac{w_j}{N\mathcal{T}}\right)  \sum_{\genfrac{}{}{0pt}{1}{R \subsetneq K, Q \subsetneq L}{|R|=|Q|}} \sum_{(R:Q)} \int_{\mathbb{R}_{\xi>0}^{|Q|}} 
        \prod_{q_j \in Q}(-\xi_{q_j}) \notag\\ \notag\\
        &\times
        \sum_{R_1\subsetneq R^c, Q_1 \subsetneq Q^c} \sum_{k \in R_1^c, l \in Q_1^c} (-1)^{|R_1^{>k} \cup Q_1^{>l}|} \int_{\mathbb{R}_{\xi_l}}\int_{\mathbb{R}_{\xi>0}^{|R^c\backslash\{k\} \cup Q^c\backslash\{l\}|}} \left(\xi_l + 1 + \sum_{j \in R_1^c\backslash\{k\}  \cup Q_1^{<l}} \xi_j - \sum_{j \in Q_1^{>l}} \xi_j\right) \notag\\ \notag\\
        &\times \Phi\Bigg(\sum_{j \in \{R^c\backslash\{k\} \cup L\}} \xi_j \mathbf{e}_j + \sum_{m \in M} \frac{w_m}{N\mathcal{T}}\mathbf{e}_m + \sum_{r_j \in R} \left(\left(\frac{w_{r_j}+w_{q_j}}{N\mathcal{T}}\right)+ \xi_{q_j}\right) \mathbf{e}_{r_j} \notag\\ \notag\\
        &+ \left(\sum_{j \in S_2} \xi_j - \sum_{j \in S_1\backslash\{k\}} \xi_j + \sum_{j \in R^c \cup Q^c} \frac{w_j}{N\mathcal{T}}\right)\mathbf{e}_k \Bigg) \quad d\mathbf{\xi} \, d\mathbf{w},
    \end{align}
where  \(\mathbf{e}_i\) is the \(i\)th standard basis vector \((0,\dots,1,\dots,0)\). We can now Taylor expand \(\Phi\) just as in \cite{kn:consna14} to get

    \begin{align}
        \label{I_1_1_taylor_expand}
     I_{1,1} &= (1 + O(1/\mathcal{T})) \sum_{K+L+M = \{1,\dots,n\}} (-1)^{|L|} \int_{\mathbb{R}_w^{n}} \prod_{j=1}^{n} g_j(w_j)\, \delta\left(\sum_{j=1}^{n} \frac{w_j}{N\mathcal{T}}\right)  \sum_{\genfrac{}{}{0pt}{1}{R \subsetneq K, Q \subsetneq L}{|R|=|Q|}} \sum_{(R:Q)} \int_{\mathbb{R}_{\xi>0}^{|Q|}} 
        \prod_{q_j \in Q}(-\xi_{q_j}) \notag\\ \notag\\
        &\times
        \sum_{R_1\subsetneq R^c, Q_1 \subsetneq Q^c} \sum_{k \in R_1^c, l \in Q_1^c} (-1)^{|R_1^{>k} \cup Q_1^{>l}|} \int_{\mathbb{R}_{\xi_l}}\int_{\mathbb{R}_{\xi>0}^{|R^c\backslash\{k\} \cup Q^c\backslash\{l\}|}} \left(\xi_l + 1 + \sum_{j \in R_1^c\backslash\{k\}  \cup Q_1^{<l}} \xi_j - \sum_{j \in Q_1^{>l}} \xi_j\right) \notag\\ \notag\\
        &\times \Phi\Bigg(\sum_{j \in R^c\backslash\{k\} \cup L} \xi_j \mathbf{e}_j + \sum_{r_j \in R} \xi_{q_j}\mathbf{e}_{r_j} + \left(\sum_{j \in S_2} \xi_j - \sum_{j \in S_1\backslash\{k\} } \xi_j\right)\mathbf{e}_k \Bigg) \quad d\mathbf{\xi} \, d\mathbf{w}.
    \end{align}

Note that the \(w\) variables only appear inside the \(g\) functions and delta function, which we can rewrite in the following way

    \begin{align}
        \label{prod_of_gs}
        \int_{\mathbb{R}^{n}} \prod_{j=1}^n g_j(w_j) \, \delta\left(\sum_{j =1}^n \frac{w_j}{N\mathcal{T}}\right) d\mathbf{w} &= 
         N\mathcal{T} \int_{\mathbb{R}^{n}} \prod_{j=1}^n g_j(w_j) \, \delta\left(\sum_{j =1}^n w_j\right) d\mathbf{w} \notag\\ \notag \\
         &= N\mathcal{T} \int_{\mathbb{R}^{n-1}} \prod_{j=1}^{n-1} g_j(w_j) \int_{\mathbb{R}} g_n(w_n)\, \delta\left(\sum_{j =1}^n w_j\right) d\mathbf{w},
    \end{align}
using the sifting property of the delta function we find

    \begin{align}
        \label{prod_of_gs_as_kappa}         
         N\mathcal{T} \int_{\mathbb{R}^{n-1}} \prod_{j=1}^{n-1} g_j(w_j) g_n\left(-\sum_{j=1}^{n-1} w_j\right) d\mathbf{w}
        &= N\mathcal{T} \int_{\mathbb{R}^{n-1}} \prod_{j=1}^{n-1} g_j(w_j) \frac{1}{2\pi} \int_{\mathbb{R}} h_n(t)e^{it(\sum_{j=1}^{n-1} w_j)} dt \, d\mathbf{w} \notag\\ \notag\\
        &= \frac{N\mathcal{T}}{2\pi} \int_{\mathbb{R}} h_n(t) \prod_{j=1}^{n-1} \left(\int_{\mathbb{R}} g_j(w_j)e^{itw_j} dw_j\right) \, dt  \notag\\ \notag\\
        &= \frac{N\mathcal{T}}{2\pi} \int_{\mathbb{R}} \prod_{j=1}^n h_j(t) dt \notag\\ \notag\\
        &=: \frac{N\mathcal{T}}{2\pi} \kappa (\mathbf{h}).
    \end{align}

Inserting this back into expression (\ref{I_1_1_taylor_expand}), we finally have

    \begin{align}
        \label{I_1_1_result}
     I_{1,1} &= \frac{N\mathcal{T}}{2\pi} \kappa (\mathbf{h}) \sum_{K+L+M = \{1,\dots,n\}} (-1)^{|L|}  \sum_{\genfrac{}{}{0pt}{1}{R \subsetneq K, Q \subsetneq L}{|R|=|Q|}} \sum_{(R:Q)} \int_{\mathbb{R}_{\xi>0}^{|Q|}} 
        \prod_{q_j \in Q}(-\xi_{q_j}) \notag\\ \notag\\
        &\times
        \sum_{R_1\subsetneq R^c, Q_1 \subsetneq Q^c} \sum_{k \in R_1^c, l \in Q_1^c} (-1)^{|R_1^{>k} \cup Q_1^{>l}|} \int_{\mathbb{R}_{\xi_l}}\int_{\mathbb{R}_{\xi>0}^{|R^c\backslash\{k\} \cup Q^c\backslash\{l\}|}} \left(\xi_l + 1 + \sum_{j \in R_1^c\backslash\{k\}  \cup Q_1^{<l}} \xi_j - \sum_{j \in Q_1^{>l}} \xi_j\right) \notag\\ \notag\\
        &\times \Phi\Bigg(\sum_{j \in R^c\backslash\{k\} \cup L} \xi_j \mathbf{e}_j + \sum_{r_j \in R} \xi_{q_j}\mathbf{e}_{r_j} + \left(\sum_{j \in S_2} \xi_j - \sum_{j \in S_1\backslash\{k\} } \xi_j\right)\mathbf{e}_k \Bigg) \quad d\mathbf{\xi} + O(N),
    \end{align}
where the condition on the \(\xi\) integrals is \(\xi_l + \sum_{j \in R_1^c\backslash\{k\}} \xi_j +\sum_{j \in Q_1^{<l}} \xi_j +1 < \sum_{j \in Q_1^{>l}} \xi_j\).

Putting \(I_{0,0}\) (the result of which comes from \cite{kn:consna14}) and \(I_{1,1}\) together, we get that 

\begin{align}
    \label{result}
    \int_{U(N)} &\sum_{j_1,\dots,j_n}F(\theta_{j_1},\dots,\theta_{j_n})dX = \notag\\ \notag\\
    &\kappa(\mathbf{h}) \frac{N\mathcal{T}}{2\pi} \sum_{K+L+M = \{1,\dots,n\}} \Bigg\{ \sum_{(K:L)} \int_{\mathbb{R}^{|K|}_{\xi_{k_j}>0}} \prod_{j=1}^{|K|}(\xi_{k_j}) \Phi 
\left(\sum_{j=1}^{|K|} \xi_{k_j}\mathbf{e}_{k_j} - \sum_{j=1}^{|K|} \xi_{k_j}\mathbf{e}_{l_j}\right) d\xi \notag \\ \notag \\ 
&
+ (-1)^{|L|}  \sum_{\genfrac{}{}{0pt}{1}{R \subsetneq K, Q \subsetneq L}{|R|=|Q|}} \sum_{(R:Q)} \int_{\mathbb{R}_{\xi>0}^{|Q|}} 
        \prod_{q_j \in Q}(-\xi_{q_j}) \notag\\ \notag\\
        &\times
        \sum_{R_1\subsetneq R^c, Q_1 \subsetneq Q^c} \sum_{k \in R_1^c, l \in Q_1^c} (-1)^{|R_1^{>k} \cup Q_1^{>l}|} \int_{\mathbb{R}_{\xi_l}}\int_{\mathbb{R}_{\xi>0}^{|R^c\backslash\{k\} \cup Q^c\backslash\{l\}|}} \left(\xi_l + 1 + \sum_{j \in R_1^c\backslash\{k\} \cup Q_1^{<l}} \xi_j - \sum_{j \in Q_1^{>l}} \xi_j\right) \notag\\ \notag\\
        &\times \Phi\Bigg(\sum_{j \in R^c\backslash\{k\} \cup L} \xi_j \mathbf{e}_j + \sum_{j=1}^{|R|} \xi_{q_j}\mathbf{e}_{r_j} + \left(\sum_{j \in S_2} \xi_j - \sum_{j \in S_1\backslash\{k\} } \xi_j\right)\mathbf{e}_k \Bigg) \quad d\mathbf{\xi}\Bigg\} + O(N),
     \end{align}
  where the sum $\sum_{(K:L)} $ only exists if $|K|=|L|$ and is otherwise zero. 
\end{proof}

%%%%%%%%%%%%%%%%%%%%%%%%%%%%%%%%%%%%%%%%%%%%
%\newpage
%\maketitle
\section{Extension of the support of \(\Phi\) to q = 3}
%\subfile{sections/ch6_q_3}
The following theorem is a result of extending the support of the Fourier transform of the test function \(f\) to be \(\sum_j |\xi_j| < 6\).  Clearly the expression becomes more complicated each time the support is extended and more terms survive.  However, it should still provide easier comparison  or guidance for number theoretical $n$-correlations than the determinant form. 

The notation is as defined in the previous sections. The sum over \((R:Q)\) means that we are summing over all ways of pairing an element from the set \(R\) with an element from the set \(Q\), where is understood that these two sets are of the same size. The set \(R_1^{>k_1}\) denotes all indices in \(R_1\) greater than \(k_1\). The integral \(\int_{\mathbb{R}_{\xi_q>0}^{|Q|}}\) means that we are integrating over the range \((0,\infty)\) for all \(\xi_q\) where \(q \in Q\). The vector is defined as \(\mathbf{e}_{i,j} = \mathbf{e}_i - \mathbf{e}_j\), where \(\mathbf{e}_i\) is the \(i\)th standard basis vector \((0,\dots,1,\dots,0)\).

Furthermore, the sets \(K, L, M\)  are disjoint. The sets \(R^c\) (the complement of \(R\) in \(K\)) and \(Q^c\) (the complement of \(Q\) in \(L\)) are each split into four disjoint sets \(R_1, R_1^c, R_2, R_2^c\) and  \(Q_1, Q_1^c, Q_2, Q_2^c\). The sets \(R_1, R_2, Q_1,Q_2\) cannot be empty.

\begin{theorem}
\label{q_3_thm}
We consider the n-point correlation function for test functions with a restricted range of support. In particular the test function \(F\) is defined as
\begin{equation}
\label{F_def_q_3}
    F(iz_1,\dots, iz_n) = f\left( \frac{iNz_{1}}{2\pi}, \dots, \frac{iN z_{n}}{2\pi} \right) h\left(\frac{iz_{1}}{\mathcal{T}}\right) \dots h\left( \frac{iz_{n}}{\mathcal{T}} \right),
\end{equation}
where  \(h_j(x) = \int_{\mathbb{R}} g_j(t) e^{ixt} dt\) are rapidly decaying functions and \(g_j\) are smooth and compactly supported. And \(\Phi\) is assumed to be smooth, even and the  inverse Fourier transform of \(f\) 
\begin{equation}
    \label{f_replaced_by_phi_pt2}
    f\left( \frac{iNz_{1}}{2\pi}, \dots, \frac{iN z_{n}}{2\pi} \right) = 
    \int_{\mathbb{R}^n} \Phi(\xi_1,\dots,\xi_n)\delta\left(\sum_{j =  1}^n \xi_j\right)
    e^{N\left(\sum_{j=1}^n z_j \xi_j\right)}d\xi,
\end{equation}
with compact support on \(\sum_{j}|\xi_j| <6\). Also suppose that \(\mathcal{T}\) goes to infinity faster than \(N\). Then the following result holds
\begin{align}
    \label{q_2_result_in_thm1}
    \int_{U(N)} &\sum_{j_1,\dots,j_n}F(\theta_{j_1},\dots,\theta_{j_n})dX = 
\frac{N\mathcal{T}}{2\pi} \kappa (\mathbf{h}) \sum_{K+L+M=\{1,\dots,n\}}(-1)^{|L|} \notag\\ \notag\\
 &\Bigg\{\sum_{(K:L)} \int_{\mathbb{R}^{|K|}_{\xi_{k_j}>0}} \prod_{j=1}^{|K|}(\xi_{k_j}) \Phi 
\left(\sum_{j=1}^{|K|} \xi_{k_j}\mathbf{e}_{k_j} - \sum_{j=1}^{|K|} \xi_{k_j}\mathbf{e}_{l_j}\right) d\xi_{K}\displaybreak[0]  \notag\\ \notag\\
        &+ \sum_{\genfrac{}{}{0pt}{1}{R \subsetneq K, Q \subsetneq L}{|R|=|Q|}} \sum_{(R:Q)} \int_{\mathbb{R}_{\xi>0}^{|Q|}} 
        \prod_{q_j \in Q}(-\xi_{q_j}) \notag\\ \notag\\
        &\times
        \sum_{R_1\subsetneq R^c, Q_1 \subsetneq Q^c} \sum_{k \in R_1^c, l \in Q_1^c} (-1)^{|R_1^{>k} \cup Q_1^{>l}|} \int_{\mathbb{R}_{\xi_l}}\int_{\mathbb{R}_{\xi>0}^{|R^c\backslash\{k\} \cup Q^c\backslash\{l\}|}} \left(\xi_l + 1 + \sum_{j \in R_1^c\backslash\{k\} \cup Q_1^{<l}} \xi_j - \sum_{j \in Q_1^{>l}} \xi_j\right)  \notag\\ \notag\\
        &\times \Phi\Bigg(\sum_{j \in R^c\backslash\{k\} \cup L} \xi_j \mathbf{e}_j + \sum_{j=1}^{|R|} \xi_{q_j}\mathbf{e}_{r_j} + \left(\sum_{j \in S_2} \xi_j - \sum_{j \in S_1} \xi_j\right)\mathbf{e}_k \Bigg) \, d\mathbf{\xi}_{R^c\backslash\{k\} \cup L} \displaybreak[0] \notag\\ \notag\\
    &+ \sum_{\genfrac{}{}{0pt}{1}{R \subset K\backslash\{k_1,k_2\}, Q \subset L\backslash\{l_1,l_2\}}{|R|=|Q|}} \sum_{(R:Q)}
    \int_{\mathbb{R}_{\xi_q>0}^{|Q|}} \prod_{q_j \in Q}
    (-\xi_{q_j})\notag\\ \notag\\
    &\times \sum_{\genfrac{}{}{0pt}{1}{R_1 \cup R_1^c \cup R_2 \cup R_2^c = R^c}{Q_1 \cup Q_1^c \cup Q_2 \cup Q_2^c = Q^c}} \, \sum_{\genfrac{}{}{0pt}{1}{k_1 \in R_1,\, k_2 \in R_2}{l_1 \in Q_1,\, l_2 \in Q_2}} 
    (-1)^{|R_{1}^{>k_1} \cup R_{2}^{>k_2} \cup Q_{1}^{>l_1} \cup Q_{2}^{>k_2}|}  
 \notag\\ \notag\\
    &\times \Bigg[ \Bigg\{\int_{\mathbb{R}_{\xi_j>0}^{|R^c \cup Q^c\backslash\{k_1,k_2,l_1,l_2\}|}}
    \int_{\mathbb{R}_{l_2}}\int_{\mathbb{R}_{l_1}}\left(1+ \sum_{j \in Q_{1}^{\leq l_1} \cup R_1^c}\xi_j - \sum_{j \in Q_{1}^{>l_1}}\xi_j\right)\left(1+ \sum_{j \in Q_{2}^{\leq l_2} \cup R_2^c}\xi_j - \sum_{j \in Q_{2}^{>l_2}}\xi_j\right)  \notag\\ \notag\\
    &\Phi\Bigg(\sum_{j \in R^c \backslash\{k_1,k_2\} \cup L} \xi_j \mathbf{e}_j + \sum_{j=1}^{|R|} \left(\xi_{q_j}\right) \mathbf{e}_{r_j}+ \Bigg(\sum_{j\in R_{1}^{<k_1} \cup Q_1^c \cup Q_2^{>l_2}} \xi_j - \sum_{j\in R_1^{>k_1} \cup Q_2^{\leq l_2} \cup R_2^c} \xi_j\Bigg) \mathbf{e}_{k_1}\notag\\ \notag\\
    &+ \Bigg(\sum_{j\in R_{2}^{<k_2} \cup Q_2^c \cup Q_1^{>l_1}} \xi_j - \sum_{j\in R_2^{>k_2} \cup Q_1^{\leq l_1} \cup R_1^c} \xi_j\Bigg) \mathbf{e}_{k_2} \Bigg) \displaybreak[0]  \notag\\ \notag\\
    &+ 
\left(1+ \sum_{j \in Q_{1}^{\leq l_1} \cup R_1^c}\xi_j - \sum_{j \in Q_{1}^{>l_1}}\xi_j\right)\left(1+ \sum_{j \in Q_{2}^{\leq l_2} \cup R_2^c}\xi_j - \sum_{j \in Q_{2}^{>l_2}}\xi_j\right)\notag\\ \notag\\
    &\Phi\Bigg(\sum_{j \in R^c \backslash\{k_1,k_2\} \cup L} \xi_j \mathbf{e}_j + \sum_{j=1}^{|R|} \left(\xi_{q_j}\right) \mathbf{e}_{r_j}+ \Bigg(\sum_{j\in R_{1}^{<k_1} \cup Q_1^c \cup Q_1^{>l_1}} \xi_j - \sum_{j\in R_1^{>k_1} \cup Q_1^{\leq l_1} \cup R_1^c} \xi_j\Bigg) \mathbf{e}_{k_1} \notag\\ \notag\\
    &+ \Bigg(\sum_{j\in R_{2}^{<k_2} \cup Q_2^c \cup Q_2^{>l_2}} \xi_j - \sum_{j\in R_2^{>k_2} \cup Q_2^{\leq l_2} \cup R_2^c} \xi_j\Bigg) \mathbf{e}_{k_2} \Bigg)
d{\xi}_{R^c \cup Q^c \backslash\{k_1, k_2\}}
\Bigg\}\displaybreak[0] \notag\\ \notag\\
    &-  \int_{\mathbb{R}_{\xi_j>0}^{|R^c \cup Q^c\backslash\{k_1,k_2,l_1,l_2\}|}} \int_{\mathbb{R}_{l_2}}\int_{\mathbb{R}_{l_1}} \int_{\mathbb{R}_{k_2}} \left(\sum_{j \in R_2^{\geq k_2}} \xi_j - 1 - \sum_{j \in R_2^{< k_2} \cup Q_2^c} \xi_j\right) \notag\\ \notag\\
    &\Phi\Bigg(\sum_{j \in R^c \backslash\{k_1\} \cup L} \xi_j \mathbf{e}_j + \sum_{j=1}^{|R|} \left(\xi_{q_j}\right) \mathbf{e}_{r_j}+ \Bigg(\sum_{j\in A} \xi_j - \sum_{j\in B \cup \{k_2, l_1, l_2\}} \xi_j\Bigg) \mathbf{e}_{k_1}\Bigg)  d{\xi}_{R^c \cup Q^c \backslash\{k_1\}} 
\displaybreak[0]\notag\\ \notag\\
    &+  \int_{\mathbb{R}_{\xi_j>0}^{|R^c \cup Q^c\backslash\{k_1,k_2,l_1,l_2\}|}} \int_{\mathbb{R}_{l_2}}\int_{\mathbb{R}_{l_1}} \int_{\mathbb{R}_{k_2}} \left(\sum_{j \in R_2^{\geq k_2} \cup R_1^c \cup Q_1^{\leq l_1}} \xi_j - \sum_{j \in R_2^{< k_2} \cup Q_2^c  \cup Q_1^{> l_1}} \xi_j\right) \notag\\ \notag\\
    &\Phi\Bigg(\sum_{j \in R^c \backslash\{k_1\} \cup L} \xi_j \mathbf{e}_j + \sum_{j=1}^{|R|} \left(\xi_{q_j}\right) \mathbf{e}_{r_j}+ \Bigg(\sum_{j\in A} \xi_j - \sum_{j\in B \cup \{k_2, l_1, l_2\}} \xi_j\Bigg) \mathbf{e}_{k_1}\Bigg) d{\xi}_{R^c \cup Q^c \backslash\{k_1\}}
\displaybreak[0]\notag\\ \notag\\
    &+   \int_{\mathbb{R}_{\xi_j>0}^{|R^c \cup Q^c\backslash\{k_1,k_2,l_1,l_2\}|}} \int_{\mathbb{R}_{l_2}}\int_{\mathbb{R}_{l_1}} \int_{\mathbb{R}_{k_2}} \left(\sum_{j \in R_2^{\geq k_2} \cup R_2^c \cup Q_2^{\leq l_2}} \xi_j- \sum_{j \in R_2^{< k_2} \cup Q_2^c  \cup Q_2^{> l_2}} \xi_j\right) \notag\\ \notag\\
    &\Phi\Bigg(\sum_{j \in R^c \backslash\{k_1\} \cup L} \xi_j \mathbf{e}_j + \sum_{j=1}^{|R|} \left(\xi_{q_j}\right) \mathbf{e}_{r_j}+ \Bigg(\sum_{j\in A} \xi_j - \sum_{j\in B \cup \{k_2, l_1, l_2\}} \xi_j\Bigg) \mathbf{e}_{k_1}\Bigg) d{\xi}_{R^c \cup Q^c \backslash\{k_1\}} \displaybreak[0]\notag\\ \notag\\
    &-  \int_{\mathbb{R}_{\xi_j>0}^{|R^c \cup Q^c\backslash\{k_1,k_2,l_1,l_2\}|}} \int_{\mathbb{R}_{l_2}}\int_{\mathbb{R}_{l_1}} \int_{\mathbb{R}_{k_2}} \left(1 +\sum_{j \in R_2^{\geq k_2} \cup R_1^c \cup Q_1^{\leq l_1} } \xi_j- \sum_{j \in R_2^{< k_2} \cup Q_2^c  \cup Q_1^{> l_1} } \xi_j\right) \displaybreak[0]  \notag\\ \notag\\
    &\Phi\Bigg(\sum_{j \in R^c \backslash\{k_1\} \cup L} \xi_j \mathbf{e}_j + \sum_{j=1}^{|R|} \left(\xi_{q_j}\right) \mathbf{e}_{r_j}+ \Bigg(\sum_{j\in A} \xi_j - \sum_{j\in B \cup \{k_2, l_1, l_2\}} \xi_j\Bigg) \mathbf{e}_{k_1}\Bigg) d{\xi}_{R^c \cup Q^c \backslash\{k_1\}} \Bigg]\Bigg\} d\xi_{Q}
 + O(\mathcal{T})
     \end{align}
where \(d\xi_{Q}\) denotes integration over \(\xi\)s with indices in \(Q\) and \(\kappa(\mathbf{h})=\int_{\mathbb{R}} h_1(u)\dots h_n(u)du\) and the conditions on the \(\xi\) integrals in the square brackets are  \(1+ \sum_{j \in Q_{1}^{\leq l_1} \cup R_1^c}\xi_j - \sum_{j \in Q_{1}^{>l_1}}\xi_j < 0\)  and \(1+ \sum_{j \in Q_{2}^{\leq l_2} \cup R_2^c}\xi_j - \sum_{j \in Q_{2}^{>l_2}}\xi_j < 0\). The sets \(K_1 = R_1\cup Q_1^c\), \(K_2=R_2\cup Q_2^c\), \(L_1=Q_1 \cup R_1^c\), and \(L_2=Q_2\cup R_2^c\). The last four sets of integrals have each an extra condition where the factor in front of the \(\Phi\) function has to be positive. The first of the last four integrals exists under the condition that \(\sum_{j \in R_2^{\geq k_2}} \xi_j > 1 + \sum_{j \in R_2^{< k_2} \cup Q_2^c} \xi_j\), and is zero otherwise. The second of the last four integrals exists if  \(\sum_{j \in R_2^{\geq k_2} \cup R_1^c \cup Q_1^{\leq l_1}} \xi_j> \sum_{j \in R_2^{< k_2} \cup Q_2^c  \cup Q_1^{> l_1}} \xi_j\), and is zero to otherwise. The third of the last four integrals exists if \(\sum_{j \in R_2^{\geq k_2} \cup R_2^c \cup Q_2^{\leq l_2}} \xi_j> \sum_{j \in R_2^{< k_2} \cup Q_2^c  \cup Q_2^{> l_2}} \xi_j\), and is zero otherwise. The last integral exists if \(1 +\sum_{j \in R_2^{\geq k_2}  \cup R_1^c \cup Q_1^{\leq l_1}} \xi_j> \sum_{j \in R_2^{< k_2} \cup Q_2^c  \cup Q_1^{> l_1} } \xi_j\), and is zero otherwise.
\end{theorem}

Extending the support of \(\Phi\) to \(q=3\) adds yet another term to the previous calculation in the sum in $J^*$ (see (\ref {J_restricted}) and (\ref{J_star_def})) over subsets \(S\) and \(T\) of size less than \(q\),

\begin{equation}
\label{J_star_restricted_again}
    J_q^*(A,B) = \sum_{\genfrac{}{}{0pt}{1}{S \subset A, T\subset B}{|S| = |T| <q}}\dots.
\end{equation}
This means that we consider subsets \(S\) and \(T\) to be both empty or both having one element or both having two elements. Since the first two scenarios have now been worked out in the previous section, we will consider the sum over sets with two elements.
Considering a generic choice of  sets \(z_K=\{z_k: k\in K\}\) and \(-z_L=\{-z_l: l \in L\}\), we get

\begin{equation}
\label{J_star_3_def}
J_3^*(z_K,-z_L) = J^*_{\emptyset,\emptyset} + \sum_{\genfrac{}{}{0pt}{1}{\{z_k\} \subset z_K}{\{-z_l\}{\subset -z_L}}} J^*_{z_k,-z_l}
+ \sum_{\genfrac{}{}{0pt}{1}{\{z_{k_1},z_{k_2}\} \subset z_K}{\{-z_{l_1},-z_{l_2}\}{\subset -z_L}}} J^*_{z_{k_1},z_{k_2},-z_{l_1},-z_{l_2}}.
\end{equation}
Note that the sum is over sets $\{z_{k_1},z_{k_2}\} \subset z_K$ and $\{-z_{l_1},-z_{l_2}\}{\subset -z_L}$ and therefore we do not want to count sets twice by counting $\{z_{k_1}, z_{k_2}\}$ separately from $\{z_{k_2}, z_{k_1}\}$.  To avoid this, we will specify in what follows that $k_2>k_1$ and $l_2>l_1$.  

In (\ref{J_star_3_def}),  $J^*_{\emptyset,\emptyset}$ and $J^*_{z_k,-z_l}$ are exactly the same as in the previous section and

\begin{equation}
\begin{aligned}
    \label{J_star_a_b_c_d_def}
    J^*_{z_{k_1},z_{k_2},-z_{l_1},-z_{l_2}} &= e^{-N(z_{k_1}+z_{k_2} - z_{l_1} - z_{l_2})}
    \frac{z(z_{k_1}-z_{l_1})z(z_{k_1}-z_{l_2})z(z_{k_2}-z_{l_1})z(z_{k_2}-z_{l_2})}{z(z_{k_1}-z_{k_2})z(z_{k_2}-z_{k_1})}
    \\ \\
    &\times\frac{z(z_{l_1}-z_{k_1})z(z_{l_1}-z_{k_2})z(z_{l_2}-z_{k_1})z(z_{l_2}-z_{k_2})}{z(z_{l_1}-z_{l_2})z(z_{l_2}-z_{l_1})}\\ \\
     &\times\sum_{\genfrac{}{}{0pt}{1}{\genfrac{}{}{0pt}{1}{(z_K\backslash \{z_{k_1},z_{k_2}\})+(-z_L\backslash  \{-z_{l_1},-z_{l_2}\})}{=W_1+\dots+W_Y}}{|W_y|\leq 2}} \prod_{y=1}^{Y} H_{2,2}(W_y)
\end{aligned}
\end{equation}
and \( H_{2,2}(W)\) for \(W = \{z_{k}\} \subset z_K\backslash \{z_{k_1},z_{k_2}\}\) is defined as 

\begin{equation}
\label{H_2_first_case}
    H_{2,2}(W)=
    \frac{z'}{z}(z_{k}-z_{k_1}) + \frac{z'}{z}(z_{k}-z_{k_2}) - \frac{z'}{z}(z_{k} - z_{l_1}) - \frac{z'}{z}(z_{k} - z_{l_2}),
\end{equation}
for \(W = \{z_{l}\} \subset z_L\backslash \{z_{l_1},z_{l_2}\}\) defined as

\begin{equation}
    \label{H_2_second_case}
    H_{2,2}(W)= \frac{z'}{z}(-z_{l}+z_{l_1}) + \frac{z'}{z}(-z_{l}+z_{l_2}) - \frac{z'}{z}(-z_{l} + z_{k_1}) - \frac{z'}{z}(-z_{l} + z_{k_2}),
\end{equation}
for \(W = \{z_{k},z_l\}\) with  \( \{z_k\} \subset z_K\backslash \{z_{k_1},z_{k_2}\}\) and \( \{z_l\} \subset z_L\backslash \{z_{l_1},z_{l_2}\}\) defined as

\begin{equation}
    \label{H_2_third case}
    H_{2,2}(W)=\left( \frac{z'}{z} \right)'(z_k-z_l),
\end{equation}
and \(0\) otherwise. We replace \(J^*\) in theorem \ref{insupthm} by (\ref{J_star_3_def}) we have
\begin{equation}
\label{n_corr_with_three_terms}
\begin{aligned}
    \int_{U(N)} \sum_{j_1,\dots,j_n} &F(\theta_{j_1},\dots,\theta_{j_n}) dX 
    = \frac{1}{(2\pi i)^n}\sum_{K+L+M  = \{1,\dots,n\}} (-1)^{|L|+|M|}N^{|M|} \\ & \times \int_{(\delta)^{|K|}} \int_{(-\delta)^{|L|}} \int_{(0)^{|M|}}J_3^{*}(z_K,-z_L) \prod_{j=1}^{n} h\left( \frac{iz_{j}}{\mathcal{T}} \right)\\
    &\times \int_{\mathbb{R}^n} \Phi(\xi_1,\dots,\xi_n)\delta\left(\sum_{j =  1}^n \xi_j\right)
    e^{N\left(\sum_{j=1}^n z_j \xi_j\right)}d\xi \, dz_1\dots dz_n \\ 
    &=\frac{1}{(2\pi i)^n}\sum_{K+L+M  = \{1,\dots,n\}} (-1)^{|L|+|M|}N^{|M|} \\ & \times \int_{(\delta)^{|K|}} \int_{(-\delta)^{|L|}} \int_{(0)^{|M|}} \Bigg(J^*_{\emptyset,\emptyset} + \sum_{\genfrac{}{}{0pt}{1}{\{z_k\} \subset z_K}{\{-z_l\}{\subset -z_L}}} J^*_{z_k,-z_l}
+ \sum_{\genfrac{}{}{0pt}{1}{\{z_{k_1},z_{k_2}\} \subset z_K}{\{-z_{l_1},-z_{l_2}\}{\subset -z_L}}} J^*_{z_{k_1},z_{k_2},-z_{l_1},-z_{l_2}}\Bigg)\\
    &\times \prod_{j=1}^{n} h\left( \frac{iz_{j}}{\mathcal{T}} \right) \int_{\mathbb{R}^n} \Phi(\xi_1,\dots,\xi_n)\delta\left(\sum_{j =  1}^n \xi_j\right)
    e^{N\left(\sum_{j=1}^n z_j \xi_j\right)}d\xi \, dz_1\dots dz_n.
\end{aligned}
\end{equation}
$J^*_{\emptyset,\emptyset}$ and $J^*_{z_k,-z_l}$ are given at  (\ref{J_star_2_generic_with_z}) (line 1 and lines 2-5, respectively) and $J^*_{z_{k_1},z_{k_2},-z_{l_1},-z_{l_2}}$ is
    \begin{align}
        \label{J_2_generic_with_z}
  &  J^*_{z_{k_1},z_{k_2},-z_{l_1},-z_{l_2}}= 
    \sum_{\genfrac{}{}{0pt}{1}{R\cup R^c=K,\, Q \cup Q^c =L}{|R| = |Q|,\; k_1,k_2 \in R^c,\; l_1,l_2\in Q^c}} \Bigg\{\left(\sum_{(R:Q)} \prod_{j=1}^{|R+Q|/2} \left(\frac{z'}{z}\right)' (z_{r_j} - z_{q_j})\right)   \notag\\ \notag\\
    &\times \sum_{\genfrac{}{}{0pt}{1}{R_1 \cup R_1^c \cup R_2 \cup R_2^c = R^c,\;k_1 \in R_1,\, k_2 \in R_2}{Q_1 \cup Q_1^c \cup Q_2 \cup Q_2^c = Q^c,\; l_1 \in Q_1,\, l_2 \in Q_2}} \,e^{-N(z_{k_1}+z_{k_2} - z_{l_1} - z_{l_2})} \frac{z(z_{k_1}-z_{l_1})z(z_{k_1}-z_{l_2})z(z_{k_2}-z_{l_1})z(z_{k_2}-z_{l_2})}{z(z_{k_1}-z_{k_2})z(z_{k_2}-z_{k_1})}
    \notag\\ \notag\\
    &\times\frac{z(z_{l_1}-z_{k_1})z(z_{l_1}-z_{k_2})z(z_{l_2}-z_{k_1})z(z_{l_2}-z_{k_2})}{z(z_{l_1}-z_{l_2})z(z_{l_2}-z_{l_1})}\notag\\ \notag\\
     & \times \left(\prod_{k \in R_1\backslash \{k_1\}}\frac{z'}{z}(z_k-z_{k_1})\right)
    \left(\prod_{k \in R_1^c} -\frac{z'}{z}(z_k-z_{l_1})
    \right)\ 
     \left(\prod_{k \in R_2\backslash \{k_2\}}\frac{z'}{z}(z_k-z_{k_2})\right) \notag\\ \notag\\
    &\times \left(\prod_{k \in R_2^c} -\frac{z'}{z}(z_k-z_{l_2})
    \right)
     \left(\prod_{l \in Q_1\backslash \{l_1\}}\frac{z'}{z}(-z_l+z_{l_1})\right)
    \left(\prod_{l \in Q_1^c} -\frac{z'}{z}(-z_l+z_{k_1})
    \right)\ \notag\\ \notag\\
    &\times \left(\prod_{l \in Q_2\backslash \{l_2\}}\frac{z'}{z}(-z_l+z_{l_2})\right)
    \left(\prod_{l \in Q_2^c} -\frac{z'}{z}(-z_l+z_{k_2})
    \right)\Bigg\}.
    \end{align}
In the above sums, we split the sets \(K\) and \(L\) into \(R \cup R^c\) and \(Q \cup Q^c\) with \(R\) and \(Q\) being the same size. Furthermore, we split the set \(R^c\) into four disjoint sets \(R_1\), \( R_1^c\) and \(R_2 \), \( R_2^c\). Note that here the `complement' notation is used to mimic the $q=2$ case, rather than because $R_1^c$ is the complement of $R_1$ in any meaningful way.  Similarly we split  the set \(Q^c\) into \(Q_1 \), \(Q_1^c\) and \(Q_2 \), \( Q_2^c\). Note that the sets \(R_1, R_2, Q_1\) and \(Q_2\) cannot be empty.  The two elements, $k_1$ and $k_2$, that we choose for the set \(S\) come from \(R_1\) and \(R_2\) respectively, and the two elements, $l_1$ and $l_2$,  for \(T\) come from \(Q_1\) and \(Q_2\) respectively. We want to work out the following

\begin{align}
    \label{I_2_2_def} 
    I&_{2,2}=
    \frac{1}{(2\pi i)^n}\sum_{K+L+M  = \{1,\dots,n\}} (-1)^{|L|+|M|}N^{|M|}  \int_{(\delta)^{|K|}} \int_{(-\delta)^{|L|}} \int_{(0)^{|M|}} \Bigg(\sum_{\genfrac{}{}{0pt}{1}{\{z_{k_1},z_{k_2}\} \subset z_K}{\{-z_{l_1},-z_{l_2}\}{\subset -z_L}}} J^*_{z_{k_1},z_{k_2},-z_{l_1},-z_{l_2}}\Bigg)\notag\\ \notag\\
    &\times \prod_{j=1}^{n} h\left( \frac{iz_{j}}{\mathcal{T}} \right)\int_{\mathbb{R}^n} \Phi(\xi_1,\dots,\xi_n)\delta\left(\sum_{j = 1}^n \xi_j\right)
    e^{N\left(\sum_{j=1}^n z_j \xi_j\right)}d\xi \, dz_1\dots dz_n
    \displaybreak[0] \notag\\ \notag\\
    &=\frac{1}{(2 \pi i)^n} \sum_{K+L+M=\{1,\dots,n\}}(-1)^{|L|+|M|}N^{|M|} \int_{\mathbb{R}^n} \Phi(\xi_1,\dots,\xi_n)\delta(\sum_{j=1}^n \xi_j) 
 \prod_{m\in M} \int_{(0)} h_m\left(\frac{i z_m}{\mathcal{T}}\right)e^{N z_m \xi_m} dz_m \notag\\ \notag\\
    & \sum_{\genfrac{}{}{0pt}{1}{R\cup R^c=K,\, Q \cup Q^c =L}{|R| = |Q|,\; |R^c|\geq 2,\; |Q^c|\geq 2}}\Bigg[ \sum_{(R:Q)} \prod_{j=1}^{|R+Q|/2} \int_{(\delta)} \int_{(-\delta)} h_{r_j}\left(\frac{i z_{r_j}}{\mathcal{T}}\right) e^{N z_{r_j}\xi_{r_j}} h_{q_j}\left(\frac{i z_{q_j}}{\mathcal{T}}\right) e^{N z_{q_j}\xi_{q_j}} \left(\frac{z'}{z}\right)'(z_{r_j} - z_{q_j})\notag\\ \notag\\
    &\times dz_{r_j} dz_{q_j}\Bigg]    \displaybreak[0]\notag\\ \notag\\
    &\times \Bigg[\sum_{\genfrac{}{}{0pt}{1}{\genfrac{}{}{0pt}{1}{R_1 \cup R_1^c \cup R_2 \cup R_2^c = R^c}{Q_1 \cup Q_1^c \cup Q_2 \cup Q_2^c = Q^c}}{R_1,R_2,Q_1,Q_2\neq \emptyset}} \, \sum_{\genfrac{}{}{0pt}{1}{k_1 \in R_1,\, k_2 \in R_2, k_2>k_1}{l_1 \in Q_1,\, l_2 \in Q_2, l_2>l_1}}  \int_{(\delta_{k_1})}\int_{(\delta_{k_2})} \int_{(-\delta_{l_1})}\int_{(-\delta_{l_2})} h_{k_1}\left(\frac{i z_{k_1}}{\mathcal{T}}\right) e^{N z_{k_1} \xi_{k_1}} h_{k_2}\left(\frac{i z_{k_2}}{\mathcal{T}}\right) e^{N z_{k_2} \xi_{k_2}} \notag\\ \notag\\ 
    &\times h_{l_1}\left(\frac{i z_{l_1}}{\mathcal{T}}\right) e^{N z_{l_1} \xi_{l_1}} h_{l_2}\left(\frac{i z_{l_2}}{\mathcal{T}}\right) e^{N z_{l_2} \xi_{l_2}} e^{-N(z_{k_1}+z_{k_2} - z_{l_1} - z_{l_2})} \notag\\ \notag\\ 
    &\times\frac{z(z_{k_1}-z_{l_1})z(z_{k_1}-z_{l_2})z(z_{k_2}-z_{l_1})z(z_{k_2}-z_{l_2})z(z_{l_1}-z_{k_1})z(z_{l_1}-z_{k_2})z(z_{l_2}-z_{k_1})z(z_{l_2}-z_{k_2})}{z(z_{k_1}-z_{k_2})z(z_{k_2}-z_{k_1})z(z_{l_1}-z_{l_2})z(z_{l_2}-z_{l_1})}\displaybreak[0] \notag\\ \notag\\ 
    &\times  
     \left(\prod_{k \in R_1 \backslash \{k_1\}}\int_{(\delta_k)} h_k\left(\frac{i z_k}{\mathcal{T}}\right) e^{N z_k \xi_k} \frac{z'}{z}(z_k-z_{k_1}) dz_k\right) 
     \left(\prod_{k \in R_1^c} \int_{(\delta_k)} -h_k\left(\frac{i z_k}{\mathcal{T}}\right) e^{N z_k \xi_k} \frac{z'}{z}(z_k-z_{l_1}) dz_k \right) 
     \notag\\ \notag \\ 
     &\times  
     \left(\prod_{k \in R_2 \backslash \{k_2\}}\int_{(\delta_k)} h_k\left(\frac{i z_k}{\mathcal{T}}\right) e^{N z_k \xi_k} \frac{z'}{z}(z_k-z_{k_2}) dz_k\right) 
     \left(\prod_{k \in R_2^c} \int_{(\delta_k)} -h_k\left(\frac{i z_k}{\mathcal{T}}\right) e^{N z_k \xi_k} \frac{z'}{z}(z_k-z_{l_2}) dz_k \right) \displaybreak[0]
     \notag\\ \notag \\ 
    &\times  
     \left(\prod_{l \in Q_1 \backslash \{l_1\}}\int_{(-\delta_l)} h_l\left(\frac{i z_l}{\mathcal{T}}\right) e^{N z_l \xi_l} \frac{z'}{z}(-z_l+z_{l_1}) dz_l\right) 
     \left(\prod_{l \in Q_1^c} \int_{(-\delta_l)} -h_l\left(\frac{i z_l}{\mathcal{T}}\right) e^{N z_l \xi_l} \frac{z'}{z}(-z_l+z_{k_1}) dz_l \right) 
     \notag\\ \notag \\ 
     &\times  
     \left(\prod_{l \in Q_2 \backslash \{l_2\}}\int_{(-\delta_l)} h_l\left(\frac{i z_l}{\mathcal{T}}\right) e^{N z_l \xi_l} \frac{z'}{z}(-z_l+z_{l_2}) dz_l\right) 
     \left(\prod_{l \in Q_2^c} \int_{(-\delta_l)} -h_l\left(\frac{i z_l}{\mathcal{T}}\right) e^{N z_l \xi_l} \frac{z'}{z}(-z_l+z_{k_2}) dz_l \right) 
     \notag\\ \notag \\ 
    &  dz_{l_2}dz_{l_1} dz_{k_2} dz_{k_1} \Bigg] \, d\xi_1 \dots d\xi_n.
\end{align}
Here we have imposed the conditions $k_2>k_1$ and $l_2>l_1$ in order to avoid double counting.  If it is not possible to meet this condition in the inner sum over $k_2, k_1, l_2, l_1$ then the sum is simply zero. 
%%%%%%lemma
\begin{lemma}
\label{lemma_q_3}
Suppose \(R^c = R_1 \cup R_1^c \cup R_2 \cup R_2^c\)  and \(Q^c = Q_1 \cup Q_1^c \cup Q_2 \cup Q_2^c\) are disjoint indexing subsets of \((1, \dots, n)\).  We also have $k_1\in R_1$ and $k_2\in R_2$ with $k_2>k_1$.  Similarly $l_1\in Q_1$ and $l_2\in Q_2$ with $l_2>l_1$.
The sets \(K_1, K_2,L_1\) and \(L_2\) are defined as \(K_1 = R_1 \cup Q_1^c\) and \(K_2 = R_2 \cup Q_2^c\), and the sets \(L_1 = Q_1 \cup R_1^c\) and \(L_2 = Q_2 \cup R_2^c\).
\begin{align}
    \label{last_lemma} 
   & \int_{(\delta_{k_1})}\int_{(\delta_{k_2})} \int_{(-\delta_{l_1})}\int_{(-\delta_{l_2})} h_{k_1}\left(\frac{i z_{k_1}}{T}\right) e^{N z_{k_1} \xi_{k_1}} h_{k_2}\left(\frac{i z_{k_2}}{T}\right) e^{N z_{k_2} \xi_{k_2}} \notag\\ \notag\\ 
    &\times h_{l_1}\left(\frac{i z_{l_1}}{T}\right) e^{N z_{l_1} \xi_{l_1}} h_{l_2}\left(\frac{i z_{l_2}}{T}\right) e^{N z_{l_2} \xi_{l_2}} e^{-N(z_{k_1}+z_{k_2} - z_{l_1} - z_{l_2})} \notag\\ \notag\\ 
    &\times\frac{z(z_{k_1}-z_{l_1})z(z_{k_1}-z_{l_2})z(z_{k_2}-z_{l_1})z(z_{k_2}-z_{l_2})z(z_{l_1}-z_{k_1})z(z_{l_1}-z_{k_2})z(z_{l_2}-z_{k_1})z(z_{l_2}-z_{k_2})}{z(z_{k_1}-z_{k_2})z(z_{k_2}-z_{k_1})z(z_{l_1}-z_{l_2})z(z_{l_2}-z_{l_1})}\displaybreak[0] \notag\\ \notag\\ 
    &\times  
     \left(\prod_{k \in R_1 \backslash \{k_1\}}\int_{(\delta_k)} h_k\left(\frac{i z_k}{T}\right) e^{N z_k \xi_k} \frac{z'}{z}(z_k-z_{k_1}) dz_k\right) 
     \left(\prod_{k \in R_1^c} \int_{(\delta_k)} -h_k\left(\frac{i z_k}{T}\right) e^{N z_k \xi_k} \frac{z'}{z}(z_k-z_{l_1}) dz_k \right) 
     \notag\\ \notag \\ 
     &\times  
     \left(\prod_{k \in R_2 \backslash \{k_2\}}\int_{(\delta_k)} h_k\left(\frac{i z_k}{T}\right) e^{N z_k \xi_k} \frac{z'}{z}(z_k-z_{k_2}) dz_k\right) 
     \left(\prod_{k \in R_2^c} \int_{(\delta_k)} -h_k\left(\frac{i z_k}{T}\right) e^{N z_k \xi_k} \frac{z'}{z}(z_k-z_{l_2}) dz_k \right) 
     \notag\\ \notag \\ 
    &\times  
     \left(\prod_{l \in Q_1 \backslash \{l_1\}}\int_{(-\delta_l)} h_l\left(\frac{i z_l}{T}\right) e^{N z_l \xi_l} \frac{z'}{z}(-z_l+z_{l_1}) dz_l\right) 
     \left(\prod_{l \in Q_1^c} \int_{(-\delta_l)} -h_l\left(\frac{i z_l}{T}\right) e^{N z_l \xi_l} \frac{z'}{z}(-z_l+z_{k_1}) dz_l \right) 
     \notag\\ \notag \\ 
     &\times  
     \left(\prod_{l \in Q_2 \backslash \{l_2\}}\int_{(-\delta_l)} h_l\left(\frac{i z_l}{T}\right) e^{N z_l \xi_l} \frac{z'}{z}(-z_l+z_{l_2}) dz_l\right) 
     \left(\prod_{l \in Q_2^c} \int_{(-\delta_l)} -h_l\left(\frac{i z_l}{T}\right) e^{N z_l \xi_l} \frac{z'}{z}(-z_l+z_{k_2}) dz_l \right) 
     \notag\\ \notag \\ 
    &  dz_{l_2}dz_{l_1} dz_{k_2} dz_{k_1}\displaybreak[0] \notag\\ \notag \\ 
    &= \left(1+O\left(\frac{1}{N}\right)\right)(2\pi i)^{|R^c \cup Q^c|}
 \Bigg\{(NT)^2
    \Bigg[\left(1 + \sum_{j \in L_1} \xi_{j}\right)\left(1 + \sum_{j \in L_2} \xi_{j}\right) \notag\\ \notag\\ 
    &\times \int_{-\infty}^{\infty}\dots \int_{-\infty}^{\infty} \prod_{j \in K_1 \backslash\{k_1\} \cup L_2}g_j(u_j) g_{k_1}\left(NT\left(\sum_{j \in K_1 \cup L_2}\xi_j\right)-\sum_{j \in K_1 \backslash\{k_1\} \cup L_2} u_j\right)d \mathbf{u}\notag\\ \notag\\ 
    &\times\int_{-\infty}^{\infty}\dots \int_{-\infty}^{\infty} \prod_{j \in K_2 \backslash\{k_2\} \cup L_1}g_j(u_j)\, g_{k_2}\left(NT\left(\sum_{j \in K_2 \cup L_1}\xi_j\right)-\sum_{j \in K_2 \backslash\{k_2\} \cup L_1} u_j\right)d \mathbf{u}   \displaybreak[0] \notag\\ \notag\\ 
    &+\int_{-\infty}^{\infty}\dots \int_{-\infty}^{\infty} \prod_{j \in K_1 \backslash\{k_1\} \cup L_1}g_j(u_j)g_{k_1}\left(NT\left(\sum_{j \in K_1 \cup L_1}\xi_j\right)-\sum_{j \in K_1 \backslash\{k_1\} \cup L_1} u_j\right)d \mathbf{u} \notag\\ \notag\\ 
    &\times \int_{-\infty}^{\infty}\dots \int_{-\infty}^{\infty} \prod_{j \in K_2 \backslash\{k_2\} \cup L_2}g_j(u_j)\, g_{k_2}\left(NT\left(\sum_{j \in K_2 \cup L_2}\xi_j\right)-\sum_{j \in K_2 \backslash\{k_2\} \cup L_2} u_j\right)d \mathbf{u}\Bigg] \notag\\ \notag\\ 
    &+2NT\left(-1 + \sum_{j \in K_2} \xi_j\right)\int_{\mathbb{R}^{|R^c \cup Q^c \backslash\{k_1\}|}} \prod_{j \in R^c \cup Q^c \backslash\{k_1\}} g_{j}(u_j) g_{k_1}\left(NT\left(\sum_{j \in R^c \cup Q^c} \xi_j\right) - \sum_{j \in R^c \cup Q^c\backslash\{k_1\}}u_j\right) d\mathbf{u} \notag\\ \notag\\ 
    &-2NT\left(\sum_{j \in K_2 \cup L_1} \xi_j\right)\int_{\mathbb{R}^{|R^c \cup Q^c \backslash\{k_1\}|}} \prod_{j \in R^c \cup Q^c \backslash\{k_1\}} g_{j}(u_j) g_{k_1}\left(NT\left(\sum_{j \in R^c \cup Q^c} \xi_j\right) - \sum_{j \in R^c \cup Q^c\backslash\{k_1\}}u_j\right) d\mathbf{u} \notag\\ \notag\\ 
    &-2NT\left(\sum_{j \in K_2 \cup L_2} \xi_j\right)\int_{\mathbb{R}^{|R^c \cup Q^c \backslash\{k_1\}|}} \prod_{j \in R^c \cup Q^c \backslash\{k_1\}} g_{j}(u_j) g_{k_1}\left(NT\left(\sum_{j \in R^c \cup Q^c} \xi_j\right) - \sum_{j \in R^c \cup Q^c\backslash\{k_1\}}u_j\right) d\mathbf{u}
    \notag\\ \notag\\ 
    &-2NT\left(1 + \sum_{j \in K_2 \cup L_1 \cup L_2} \xi_j\right)\int_{\mathbb{R}^{|R^c \cup Q^c \backslash\{k_1\}|}} \prod_{j \in R^c \cup Q^c \backslash\{k_1\}} g_{j}(u_j) \notag\\ \notag\\ 
    &\times g_{k_1}\left(NT\left(\sum_{j \in R^c \cup Q^c} \xi_j\right) - \sum_{j \in R^c \cup Q^c\backslash\{k_1\}}u_j\right) d\mathbf{u}\Bigg\},
\end{align}
with the conditions that  \(\xi_j < 0\) for \(j \in R_1^{<k_1} \cup R_2^{<k_2} \cup Q_1^{>l_1} \cup Q_2^{>l_2} \cup Q_1^c \cup Q_2^c\) and \(\xi_j >0\) for \(j \in R_1^{>k_1} \cup R_2^{>k_2} \cup Q_1^{<l_1} \cup Q_2^{<l_2} \cup R_1^c \cup R_2^c\). Furthermore \(1 + \sum_{j \in L_2} \xi_{j} < 0\) and \(1 + \sum_{j \in L_1} \xi_{j} < 0\). Each of the last four integrals exists when the factor in brackets preceding the integral is positive.
\end{lemma}

\renewcommand\qedsymbol{$\blacksquare$}
\begin{proof}[Proof of Lemma \ref{lemma_q_3}]
We start first by working out the integrals over sets \(R_1\backslash \{k_1\}, R_1^c\). These follow exactly the same steps as before. Starting with spreading the contours out in the same manner as in the previous chapter (see Figure \ref{contours}). We split the set \(R_1\backslash\{k_1\}\) into two sets \(R_1^{<k_1}\) with indices smaller than \(k_1\) and \(R_1^{>k_1}\) with indices larger than \(k_1\). We can see that if any \(\xi_k < 0\), then closing the contour to the far right, we get a residue from every index in \(R_1^{<k_1}\) with an extra minus sign because of the orientation of the contour around the pole at \(z_k = z_{k_1}\). And vice versa, if \(\xi_k > 0\), then closing the contour to the far left, we get a residue from every index in \(R_1^{>k_1}\). And so we have

\begin{equation}
\label{R_1_int_pt2}
\begin{aligned}
\prod_{k \in R_1 \backslash\{k_1\}} \int_{(\delta_k)}h_k\left(\frac{iz_k}{\mathcal{T}}\right)e^{Nz_k\xi_k} \frac{z'}{z}(z_k-z_{k_1})dz_k &= (2\pi i)^{|R_1 \backslash\{k_1\}|}\prod_{k \in R_1 \backslash\{k_1\}} \underset{z_k = z_{k_1}}{\text{Res}} 
\left(\frac{h_k\left(\frac{iz_k}{\mathcal{T}}\right)e^{Nz_k\xi_k}}{-(z_k-z_{k_1})}\right) \\
&= (2\pi i)^{|R_1 \backslash\{k_1\}|}\prod_{k \in R_1^{<k}}  h_k\left(\frac{iz_{k_1}}{\mathcal{T}}\right)e^{Nz_{k_1}\xi_k} \\ 
&\times\prod_{k \in R_1^{>{k_1}}} - h_k\left(\frac{iz_{k_1}}{\mathcal{T}}\right)e^{Nz_{k_1}\xi_k},
\end{aligned}
\end{equation}
where \(\xi_k<0\) for \(k \in R_1^{<k_1}\) and \(\xi_K >0\) for \(k \in R_1^{>k_1}\). Otherwise the expression is 0. 

For indices in \(R_1^c\) we have,
\begin{equation}
\label{R_1_c_int_pt2}
\begin{aligned}
\prod_{k \in R_1^c} \int_{(\delta_k)}-h_k\left(\frac{iz_k}{\mathcal{T}}\right)e^{Nz_k\xi_k} \frac{z'}{z}(z_k-z_{l_1})dz_k &= (2\pi i)^{|R_1^c|}\prod_{k \in R_1^c} \underset{z_k = z_{l_1}}{\text{Res}} 
\left(\frac{-h_k\left(\frac{iz_k}{\mathcal{T}}\right)e^{Nz_k\xi_k}}{-(z_k-z_{l_1})}\right) \\
&= (2\pi i)^{|R_1^c|} \prod_{k \in R_1^c} h_k\left(\frac{iz_{l_1}}{\mathcal{T}}\right)e^{Nz_{l_1}\xi_k},
\end{aligned}
\end{equation}
if $\xi_k>0$ for all $k\in R_1^c$ and zero otherwise.  This is because  we can see that if \(\xi_k\) is negative then closing the contour to the far right will result in no contribution. However, if \(\xi_k > 0\) then closing the contour to the far left results in a non-zero residue. 

Steps for integrals for sets \(R_2\backslash\{k_2\}\) and \(R_2^c\) are exactly the same. Therefore we have,

\begin{equation}
\label{R_2_int_pt2}
\begin{aligned}
\prod_{k \in R_2 \backslash\{k_2\}} \int_{(\delta_k)}h_k\left(\frac{iz_k}{\mathcal{T}}\right)e^{Nz_k\xi_k} \frac{z'}{z}(z_k-z_{k_2})dz_k &= (2\pi i)^{|R_2 \backslash\{k_2\}|}\prod_{k \in R_2 \backslash\{k_2\}} \underset{z_k = z_{k_2}}{\text{Res}} 
\left(\frac{h_k\left(\frac{iz_k}{\mathcal{T}}\right)e^{Nz_k\xi_k}}{-(z_k-z_{k_2})}\right) \\
&= (2\pi i)^{|R_2 \backslash\{k_2\}|}\prod_{k \in R_2^{<k_2}}  h_k\left(\frac{iz_{k_2}}{\mathcal{T}}\right)e^{Nz_{k_2}\xi_k} \\ 
&\times\prod_{k \in R_2^{>{k_2}}} - h_k\left(\frac{iz_{k_2}}{\mathcal{T}}\right)e^{Nz_{k_2}\xi_k},
\end{aligned}
\end{equation}
where \(\xi_k <0\) for \(k \in R_2^{<k_2}\) and \(\xi_k >0\) for \(k \in R_2^{>k_2}\). Otherwise the expression is 0. 

For indices in \(R_2^c\) we have,
\begin{equation}
\label{R_2_c_int_pt2}
\begin{aligned}
\prod_{k \in R_2^c} \int_{(\delta_k)}-h_k\left(\frac{iz_k}{\mathcal{T}}\right)e^{Nz_k\xi_k} \frac{z'}{z}(z_k-z_{l_2})dz_k &= (2\pi i)^{|R_2^c|}\prod_{k \in R_2^c} \underset{z_k = z_{l_2}}{\text{Res}} 
\left(\frac{-h_k\left(\frac{iz_k}{\mathcal{T}}\right)e^{Nz_k\xi_k}}{-(z_k-z_{l_2})}\right) \\
&= (2\pi i)^{|R_2^c|} \prod_{k \in R_2^c} h_k\left(\frac{iz_{l_2}}{\mathcal{T}}\right)e^{Nz_{l_2}\xi_k},
\end{aligned}
\end{equation}
where \(\xi_k >0\) for \(k \in R_2^c\), otherwise the expression is zero. 

The integrals for sets \(Q_1 \backslash \{l_1\}\), \(Q_1^c\),\(Q_2 \backslash\{l_2\}\), and \(Q_2^c\) are worked out in the same way.

\begin{equation}
\label{Q_1_int_pt2}
\begin{aligned}
\prod_{l \in Q_1 \backslash\{l_1\}} \int_{(\delta_l)}h_l\left(\frac{iz_l}{\mathcal{T}}\right)e^{Nz_l\xi_l} \frac{z'}{z}(-z_l+z_{l_1})dz_l &= (2\pi i)^{|Q_1 \backslash\{l_1\}|}\prod_{l \in Q_1 \backslash\{l_1\}} \underset{z_l = z_{l_1}}{\text{Res}} 
\left(\frac{h_l\left(\frac{iz_l}{\mathcal{T}}\right)e^{Nz_l\xi_l}}{-(-z_l+z_{l_1})}\right) \\
&= (2\pi i)^{|Q_1 \backslash\{l_1\}|}\prod_{l \in Q_1^{<l_1}}  h_l\left(\frac{iz_{l_1}}{\mathcal{T}}\right)e^{Nz_{l_1}\xi_l} \\ 
&\times\prod_{l \in Q_1^{>{l_1}}} - h_l\left(\frac{iz_{l_1}}{\mathcal{T}}\right)e^{Nz_{l_1}\xi_l},
\end{aligned}
\end{equation}
where \(\xi_l >0\) for \(l \in Q_1^{<l_1}\) and \(\xi_l <0\) for \(l \in Q_1^{>l_1}\). Otherwise the expression is 0.
\begin{equation}
\label{Q_1_c_int_pt2}
\begin{aligned}
\prod_{l \in Q_1^c} \int_{(\delta_l)}-h_l\left(\frac{iz_l}{\mathcal{T}}\right)e^{Nz_l\xi_l} \frac{z'}{z}(-z_l+z_{k_1})dz_l &= (2\pi i)^{|Q_1^c|}\prod_{l \in Q_1^c} \underset{z_l = z_{k_1}}{\text{Res}} 
\left(\frac{-h_l\left(\frac{iz_l}{\mathcal{T}}\right)e^{Nz_l\xi_l}}{-(-z_l+z_{k_1})}\right) \\
&= (2\pi i)^{|Q_1^c|} \prod_{l \in Q_1^c} h_l\left(\frac{iz_{k_1}}{\mathcal{T}}\right)e^{Nz_{k_1}\xi_l},
\end{aligned}
\end{equation}
where all \(\xi_l\)s are assumed to be negative, otherwise the expression is zero. .
\begin{equation}
\label{Q_2_int_pt2}
\begin{aligned}
\prod_{l \in Q_2 \backslash\{l_2\}} \int_{(\delta_l)}h_l\left(\frac{iz_l}{\mathcal{T}}\right)e^{Nz_l\xi_l} \frac{z'}{z}(-z_l+z_{l_2})dz_l &= (2\pi i)^{|Q_2 \backslash\{l_2\}|}\prod_{l \in Q_2 \backslash\{l_2\}} \underset{z_l = z_{l_2}}{\text{Res}} 
\left(\frac{h_l\left(\frac{iz_l}{\mathcal{T}}\right)e^{Nz_l\xi_l}}{-(-z_l+z_{l_2})}\right) \\
&= (2\pi i)^{|Q_2 \backslash\{l_2\}|}\prod_{l \in Q_2^{<l}}  h_l\left(\frac{iz_{l_2}}{\mathcal{T}}\right)e^{Nz_{l_2}\xi_l} \\ 
&\times\prod_{l \in Q_2^{>{l_2}}} - h_l\left(\frac{iz_{l_2}}{\mathcal{T}}\right)e^{Nz_{l_2}\xi_l},
\end{aligned}
\end{equation}
where \(\xi_l >0\) for \(l \in Q_2^{<l_2}\) and \(\xi_l <0\) for \(l \in Q_2^{>l_2}\) for the expression not to be zero. 
\begin{equation}
\label{Q_2_c_int_pt2}
\begin{aligned}
\prod_{l \in Q_2^c} \int_{(\delta_l)}-h_l\left(\frac{iz_l}{\mathcal{T}}\right)e^{Nz_l\xi_l} \frac{z'}{z}(-z_l+z_{k_2})dz_l &= (2\pi i)^{|Q_2^c|}\prod_{l \in Q_2^c} \underset{z_l = z_{k_2}}{\text{Res}} 
\left(\frac{-h_l\left(\frac{iz_l}{\mathcal{T}}\right)e^{Nz_l\xi_l}}{-(z_l-z_{k_2})}\right) \\
&= (2\pi i)^{|Q_2^c|} \prod_{l \in Q_2^c} h_l\left(\frac{iz_{k_2}}{\mathcal{T}}\right)e^{Nz_{k_2}\xi_l},
\end{aligned}
\end{equation}
where again all \(\xi\)s are negative, or else the whole product equals zero. 

Inserting these results into the integrals for indices \(l_2, l_1, k_2\) and \(k_1\), we have that the left hand side of  Lemma \ref{lemma_q_3} is
\begin{align}
    \label{four_integrals} 
    &(2 \pi i)^{|R^c \backslash\{k_1,k_2\} \cup Q^c \backslash\{l_1,l_2\}|}\int_{(\delta_{k_1})}\int_{(\delta_{k_2})} \int_{(-\delta_{l_1})}\int_{(-\delta_{l_2})} 
    \left(e^{Nz_{k_1}(-1 + \sum_{j \in K_1} \xi_{j})}  \prod_{j \in K_1}h_{j}\left(\frac{i z_{k_1}}{\mathcal{T}}\right) \right) \notag\\ \notag\\ 
    &\times \left(e^{Nz_{k_2}(-1 + \sum_{j \in K_2} \xi_{j})} \prod_{j \in K_2}h_{j}\left(\frac{i z_{k_2}}{\mathcal{T}}\right) \right)   
    \left( e^{Nz_{l_1}(1 + \sum_{j \in L_1} \xi_{j})} \prod_{j \in L_1}h_{j}\left(\frac{i z_{l_1}}{\mathcal{T}}\right) \right) \notag\\ \notag\\ 
    &\times \left(e^{Nz_{l_2}(1 + \sum_{j \in L_2} \xi_{j})} \prod_{j \in L_2}h_{j}\left(\frac{i z_{l_2}}{\mathcal{T}}\right) \right) \times (-1)^{|R_{1}^{>k_1} \cup R_{2}^{>k_2} \cup Q_{1}^{>l_1} \cup Q_{2}^{>k_2}|}
    \notag\\ \notag\\ 
    &\times\frac{z(z_{k_1}-z_{l_1})z(z_{k_1}-z_{l_2})z(z_{k_2}-z_{l_1})z(z_{k_2}-z_{l_2})z(z_{l_1}-z_{k_1})z(z_{l_1}-z_{k_2})z(z_{l_2}-z_{k_1})z(z_{l_2}-z_{k_2})}{z(z_{k_1}-z_{k_2})z(z_{k_2}-z_{k_1})z(z_{l_1}-z_{l_2})z(z_{l_2}-z_{l_1})} \notag\\ \notag\\ 
    &  dz_{l_2}dz_{l_1} dz_{k_2} dz_{k_1},
\end{align}
where the set \(K_1\) denotes the indices that come with the \(z_{k_1}\) variable in (\ref{R_1_int_pt2})  and (\ref{Q_1_c_int_pt2}), that is \(R_1 \cup Q_1^c\), including \(k_1\). Similarly the set \(K_2\) denotes the indices that come with the \(z_{k_2}\) variable in (\ref{R_2_int_pt2}) and (\ref{Q_2_c_int_pt2}),  that is \(R_2 \cup Q_2^c\), including \(k_2\) . Sets \(L_1\) and \(L_2\) are defined in the same manner. \(L_1 = Q_1 \cup R_1^c\) and \(L_2= Q_2 \cup R_2^c\). Also from the above we have that \(\xi_j < 0\) for \(j \in R_1^{<k_1} \cup R_2^{<k_2} \cup Q_1^{>l_1} \cup Q_2^{>l_2} \cup Q_1^c \cup Q_2^c\), and \(\xi_j >0\) for \(j \in R_1^{>k_1} \cup R_2^{>k_2} \cup Q_1^{<l_1} \cup Q_2^{<l_2} \cup R_1^c \cup R_2^c\) - if these conditions are not met then the entire quantity is zero. 

Next we make a change of variables in \(z_{l_2}, z_{l_2}, z_{k_1},z_{k_2}\) integrals, \(z_j \to z_j/N\) and approximate the function $z(x/N)=N/x +O(1)$.
\begin{align}
    \label{four_integrals_changed} 
    &(1+O(1/N)) \frac{(2 \pi i)^{|R^c \backslash\{k_1,k_2\} \cup Q^c \backslash\{l_1,l_2\}|}}{N^4}\int_{(\delta_{k_1})}\int_{(\delta_{k_2})} \int_{(-\delta_{l_1})}\int_{(-\delta_{l_2})} 
    \left(e^{z_{k_1}(-1 + \sum_{j \in K_1} \xi_{j})}  \prod_{j \in K_1}h_{j}\left(\frac{i z_{k_1}}{N\mathcal{T}}\right) \right) \notag\\ \notag\\ 
    &\times \left(e^{z_{k_2}(-1 + \sum_{j \in K_2} \xi_{j})} \prod_{j \in K_2}h_{j}\left(\frac{i z_{k_2}}{N\mathcal{T}}\right) \right)   
    \left( e^{z_{l_1}(1 + \sum_{j \in L_1} \xi_{j})} \prod_{j \in L_1}h_{j}\left(\frac{i z_{l_1}}{N\mathcal{T}}\right) \right) \notag\\ \notag\\ 
    &\times \left(e^{z_{l_2}(1 + \sum_{j \in L_2} \xi_{j})} \prod_{j \in L_2}h_{j}\left(\frac{i z_{l_2}}{N\mathcal{T}}\right) \right)
    \notag\\ \notag\\ 
    &\times
    \frac{N^8(z_{k_1}-z_{k_2})(z_{k_2}-z_{k_1})(z_{l_1}-z_{l_2})(z_{l_2}-z_{l_1})}{N^4(z_{k_1}-z_{l_1})(z_{k_1}-z_{l_2})(z_{k_2}-z_{l_1})(z_{k_2}-z_{l_2})(z_{l_1}-z_{k_1})(z_{l_1}-z_{k_2})(z_{l_2}-z_{k_1})(z_{l_2}-z_{k_2})}\notag\\ \notag\\ 
    &  dz_{l_2}dz_{l_1} dz_{k_2} dz_{k_1}.
\end{align}
Note that the change of variables also scales the \((\delta)\) contours but the $z_{k_i}$s and $z_{l_j}$s are in different half planes so this makes no difference to the poles at $z_{k_j}=z_{l_i}$.

We continue our calculation with looking more closely at the integral corresponding to \(l_2\)
\begin{equation}
    \label{l_2_integral} 
    \int_{(-\delta_{l_2})} 
    \left(e^{z_{l_2}(1 + \sum_{j \in L_2} \xi_{j})} \prod_{j \in L_2}h_{j}\left(\frac{i z_{l_2}}{N\mathcal{T}}\right) \right)
    \frac{(z_{l_1}-z_{l_2})(z_{l_2}-z_{l_1})}{(z_{k_1}-z_{l_2})(z_{k_2}-z_{l_2})(z_{l_2}-z_{k_1})(z_{l_2}-z_{k_2})}
      dz_{l_2},
\end{equation}
we see that firstly all poles lie on the right side of the plane and if \(1 + \sum_{j \in L_2} \xi_{j} \geq 0\), then moving the integral to the far left will result in no contribution,  but if \(1 + \sum_{j \in L_2} \xi_{j} < 0\) then we can move the contour to the far right and encounter poles at \(z_{l_2} = z_{k_1}\) and  \(z_{l_2} = z_{k_2}\), so we have
\begin{equation}
\begin{aligned}
\label{l_2_k_1_residue}
\underset{z_{l_2} = z_{k_1}}{\text{Res}} &\left[\left(e^{z_{l_2}(1 + \sum_{j \in L_2} \xi_{j})} \prod_{j \in L_2}h_{j}\left(\frac{i z_{l_2}}{N\mathcal{T}}\right) \right)
    \frac{(z_{l_1}-z_{l_2})^2}{(z_{k_1}-z_{l_2})^2(z_{k_2}-z_{l_2})^2} \right] \\ \\
    &=\lim\limits_{z_{l_2} \to z_{k_1}} \frac{d}{dz_{l_2}}\left[\left(e^{z_{l_2}(1 + \sum_{j \in L_2} \xi_{j})} \prod_{j \in L_2}h_{j}\left(\frac{i z_{l_2}}{N\mathcal{T}}\right) \right)
    \frac{(z_{l_1}-z_{l_2})^2}{(z_{k_2}-z_{l_2})^2} \right].
\end{aligned}
\end{equation}
Noting that
\begin{equation}
    \label{b_derivative}
    \frac{d}{dz_{l_2}}\left(\frac{(z_{l_1}-z_{l_2})^2}{(z_{k_2}-z_{l_2})^2}\right) = \frac{2(z_{l_1}-z_{l_2})(z_{l_1}-z_{k_2})}{(z_{k_2}-z_{l_2})^3},
\end{equation}
and that differentiating the $h$ factor results in a term of lower order in $N\mathcal{T}$, we get 
that the first residue is
\begin{equation}
\begin{aligned}
\label{l_2_k_1_residue_cont}
\underset{z_{l_2} = z_{k_1}}{\text{Res}} &\left[\left(e^{z_{l_2}(1 + \sum_{j \in L_2} \xi_{j})} \prod_{j \in L_2}h_{j}\left(\frac{i z_{l_2}}{N\mathcal{T}}\right) \right)
    \frac{(z_{l_1}-z_{l_2})^2}{(z_{k_1}-z_{l_2})^2(z_{k_2}-z_{l_2})^2} \right] \\ \\
    &= \left(1+O\left(\frac{1}{N\mathcal{T}}\right)\right)e^{z_{k_1}(1 + \sum_{j \in L_2} \xi_{j})} \prod_{j \in L_2}h_{j}\left(\frac{i z_{k_1}}{N\mathcal{T}}\right) \\ \\ 
    &\times\left(\left((1 + \sum_{j \in L_2} \xi_{j})\right)\frac{(z_{l_1} - z_{k_1})^2}{(z_{k_2} - z_{k_1})^2} + \frac{2(z_{l_1}-z_{k_1})(z_{l_1}-z_{k_2})}{(z_{k_2}-z_{k_1})^3}\right).
\end{aligned}
\end{equation}
The second residue is worked out in the exact same manner
\begin{equation}
\begin{aligned}
\label{l_2_k_2_residue}
\underset{z_{l_2} = z_{k_2}}{\text{Res}} &\left[\left(e^{z_{l_2}(1 + \sum_{j \in L_2} \xi_{j})} \prod_{j \in L_2}h_{j}\left(\frac{i z_{l_2}}{N\mathcal{T}}\right) \right)
    \frac{(z_{l_1}-z_{l_2})^2}{(z_{k_1}-z_{l_2})^2(z_{k_2}-z_{l_2})^2} \right] \\ \\
    &= \left(1+O\left(\frac{1}{N\mathcal{T}}\right)\right)e^{z_{k_2}(1 + \sum_{j \in L_2} \xi_{j})} \prod_{j \in L_2}h_{j}\left(\frac{i z_{k_2}}{N\mathcal{T}}\right) \\ \\ 
    &\times\left(\left((1 + \sum_{j \in L_2} \xi_{j})\right)\frac{(z_{l_1} - z_{k_2})^2}{(z_{k_1} - z_{k_2})^2} + \frac{2(z_{l_1}-z_{k_2})(z_{l_1}-z_{k_1})}{(z_{k_1}-z_{k_2})^3}\right),
\end{aligned}
\end{equation}
where the only difference between these two residues are the fractions involving \(z\) variables. One can simply swap the \(z_{k_2}\) and \(z_{k_1}\) variables in the first residue to get the second residue. Inserting these back into (\ref{four_integrals_changed}) we get
\begin{align}
    \label{three_integrals} 
    &(1+O(1/N))(2\pi i)^{|R^c \backslash\{k_1,k_2\} \cup Q^c\backslash\{l_1\}|}\int_{(\delta_{k_1})}\int_{(\delta_{k_2})} \int_{(-\delta_{l_1})}
    \left(e^{z_{k_1}(-1 + \sum_{j \in K_1} \xi_{j})}  \prod_{j \in K_1}h_{j}\left(\frac{i z_{k_1}}{N\mathcal{T}}\right) \right) \notag\\ \notag\\ 
    &\times \left(e^{z_{k_2}(-1 + \sum_{j \in K_2} \xi_{j})} \prod_{j \in K_2}h_{j}\left(\frac{i z_{k_2}}{N\mathcal{T}}\right) \right)   
    \left( e^{z_{l_1}(1 + \sum_{j \in L_1} \xi_{j})} \prod_{j \in L_1}h_{j}\left(\frac{i z_{l_1}}{N\mathcal{T}}\right) \right) \notag\\ \notag\\ 
    &\times\frac{(z_{k_1}-z_{k_2})^2}{(z_{k_1}-z_{l_1})^2(z_{k_2}-z_{l_1})^2} \notag\\ \notag\\ 
    &\times \Bigg\{-e^{z_{k_1}(1 + \sum_{j \in L_2} \xi_{j})} \prod_{j \in L_2}h_{j}\left(\frac{i z_{k_1}}{N\mathcal{T}}\right) \left[\left((1 + \sum_{j \in L_2} \xi_{j})\right)\frac{(z_{l_1} - z_{k_1})^2}{(z_{k_2} - z_{k_1})^2} + \frac{2(z_{l_1}-z_{k_1})(z_{l_1}-z_{k_2})}{(z_{k_2}-z_{k_1})^3}\right] \notag\\ \notag\\ 
    &-e^{z_{k_2}(1 + \sum_{j \in L_2} \xi_{j})} \prod_{j \in L_2}h_{j}\left(\frac{i z_{k_2}}{N\mathcal{T}}\right)\left[\left((1 + \sum_{j \in L_2} \xi_{j})\right)\frac{(z_{l_1} - z_{k_2})^2}{(z_{k_1} - z_{k_2})^2} + \frac{2(z_{l_1}-z_{k_2})(z_{l_1}-z_{k_1})}{(z_{k_1}-z_{k_2})^3}\right]\Bigg\} dz_{l_1} dz_{k_2} dz_{k_1},
\end{align}
where the overall minus sign in the final two lines is from the orientation of the contour and we must have  \(1 + \sum_{j \in L_2} \xi_{j} < 0\).

Now we focus on the \(z_{l_1}\) integral. Notice that the curly brackets contain four different terms, therefore we are considering four different integrals where are all the poles lie on the right of \(z_{l_1}\). For simplicity, we are only including terms that contain the \(z_{l_1}\) variable and leaving others unchanged. The first integral is
\begin{equation}
    \label{l_1_integral_pt_1}
    \int_{(-\delta)} e^{z_{l_1}(1+\sum_{j \in L_1} \xi_j)}\prod_{j \in L_1} h_j\left(\frac{iz_{l_1}}{N\mathcal{T}}\right)\frac{1}{(z_{k_2}-z_{l_1})^2}dz_{l_1},
\end{equation}
which corresponds to the first term in the first square brackets. We can see that if \(1+\sum_{j \in L_1} \xi_j \geq 0\), we can close the contour to the far left with no contribution, meaning the entire expression is zero. If \(1+\sum_{j \in L_1} \xi_j < 0\), we can close the contour to the far right, enclosing a double pole at \(z_{l_1} = z_{k_2}\)
\begin{equation}
    \label{zl1_zk2_residue_pt1}
    \begin{aligned}
           \underset{z_{l_1} = z_{k_2}}{\text{Res}} &\left[e^{z_{l_1}(1+\sum_{j \in L_1} \xi_j)}\prod_{j \in L_1} h_j\left(\frac{iz_{l_1}}{N\mathcal{T}}\right)\frac{1}{(z_{k_2}-z_{l_1})^2}\right] \\
           &= \left(1+O\left(\frac{1}{N\mathcal{T}}\right)\right)\left(1+\sum_{j \in L_1}\xi_j\right)e^{z_{k_2}(1 + \sum_{j \in L_1} \xi_{j})} \prod_{j \in L_1}h_{j}\left(\frac{i z_{k_2}}{N\mathcal{T}}\right). 
    \end{aligned}
\end{equation}
The second integral corresponds to the second term in the first square brackets
\begin{equation}
    \label{l_1_integral_pt_2}
    \int_{(-\delta)} e^{z_{l_1}(1+\sum_{j \in L_1} \xi_j)}\prod_{j \in L_1} h_j\left(\frac{iz_{l_1}}{N\mathcal{T}}\right)\frac{2}{(z_{k_2}-z_{k_1})(z_{k_1}-z_{l_1})(z_{k_2}-z_{l_1})}dz_{l_1},
\end{equation}
where there are two simple poles in the right half of the complex plane at \(z_{l_1}= z_{k_1}\) and \(z_{l_1} = z_{k_2}\)
\begin{equation}
    \label{zl1_zk1_residue_pt2}
    \begin{aligned}
           \underset{z_{l_1} = z_{k_1}}{\text{Res}} &\left[e^{z_{l_1}(1+\sum_{j \in L_1} \xi_j)}\prod_{j \in L_1} h_j\left(\frac{iz_{l_1}}{N\mathcal{T}}\right)\frac{2}{(z_{k_2}-z_{k_1})(z_{k_1}-z_{l_1})(z_{k_2}-z_{l_1})}\right] \\ 
           &=-e^{z_{k_1}(1 + \sum_{j \in L_1} \xi_{j})} \prod_{j \in L_1}h_{j}\left(\frac{i z_{k_1}}{N\mathcal{T}}\right) \frac{2}{(z_{k_2}-z_{k_1})^2},
    \end{aligned}
\end{equation}
and 
\begin{equation}
    \label{zl1_zk2_residue_pt2}
    \begin{aligned}
           \underset{z_{l_1} = z_{k_2}}{\text{Res}} &\left[e^{z_{l_1}(1+\sum_{j \in L_1} \xi_j)}\prod_{j \in L_1} h_j\left(\frac{iz_{l_1}}{N\mathcal{T}}\right)\frac{2}{(z_{k_2}-z_{k_1})(z_{k_1}-z_{l_1})(z_{k_2}-z_{l_1})}\right] \\ 
           &=e^{z_{k_2}(1 + \sum_{j \in L_1} \xi_{j})} \prod_{j \in L_1}h_{j}\left(\frac{i z_{k_2}}{N\mathcal{T}}\right) \frac{2}{(z_{k_2}-z_{k_1})^2},
    \end{aligned}
\end{equation}
where we assume that \(1+\sum_{j \in L_1} \xi_j < 0\) (otherwise the contour can be moved to the far left). The third integral corresponds to the first term in the second square brackets
\begin{equation}
\begin{aligned}
    \label{zl1_zk1_integral_pt3}
    \int_{(-\delta)} &e^{z_{l_1}(1+\sum_{j \in L_1} \xi_j)}\prod_{j \in L_1} h_j\left(\frac{iz_{l_1}}{N\mathcal{T}}\right)\frac{1}{(z_{k_1}-z_{l_1})^2}dz_{l_1}, \\
    \end{aligned}
\end{equation}
where, once again, if \(1+\sum_{j \in L_1} \xi_j \geq 0\), we can close the contour to the far left with no contribution, but if the sum is negative, then we close the contour to the far right and encounter a double pole at \(z_{l_1} = z_{k_1}\)
\begin{equation}
\begin{aligned}
    \label{zl1_zk1_residue_pt3}
    \underset{z_{l_1} = z_{k_2}}{\text{Res}} &\left[e^{z_{l_1}(1+\sum_{j \in L_1} \xi_j)}\prod_{j \in L_1} h_j\left(\frac{iz_{l_1}}{N\mathcal{T}}\right)\frac{1}{(z_{k_1}-z_{l_1})^2}\right] \\
    &= \left(1+O\left(\frac{1}{N\mathcal{T}}\right)\right)\left(1+\sum_{j \in L_1}\xi_j\right)e^{z_{k_1}(1 + \sum_{j \in L_1} \xi_{j})} \prod_{j \in L_1}h_{j}\left(\frac{i z_{k_1}}{N\mathcal{T}}\right).
    \end{aligned}
\end{equation}
Lastly the fourth integral corresponds to the second term in the second square brackets
\begin{equation}
    \label{l_1_integral_pt_4}
    \int_{(-\delta)} e^{z_{l_1}(1+\sum_{j \in L_1} \xi_j)}\prod_{j \in L_1} h_j\left(\frac{iz_{l_1}}{N\mathcal{T}}\right)\frac{2}{(z_{k_1}-z_{k_2})(z_{k_1}-z_{l_1})(z_{k_2}-z_{l_1})}dz_{l_1},
\end{equation}
which contains two simple poles in the right half plane and the two residues are
\begin{equation}
    \label{zl1_zk1_residue_pt_4}
     \begin{aligned}
           \underset{z_{l_1} = z_{k_1}}{\text{Res}} &\left[e^{z_{l_1}(1+\sum_{j \in L_1} \xi_j)}\prod_{j \in L_1} h_j\left(\frac{iz_{l_1}}{N\mathcal{T}}\right)\frac{2}{(z_{k_1}-z_{k_2})(z_{k_1}-z_{l_1})(z_{k_2}-z_{l_1})}\right] \\ 
           &=e^{z_{k_1}(1 + \sum_{j \in L_1} \xi_{j})} \prod_{j \in L_1}h_{j}\left(\frac{i z_{k_1}}{N\mathcal{T}}\right) \frac{2}{(z_{k_2}-z_{k_1})^2}.
    \end{aligned}
\end{equation}
and
\begin{equation}
    \label{zl1_zk2_residue_pt_4}
     \begin{aligned}
           \underset{z_{l_1} = z_{k_2}}{\text{Res}} &\left[e^{z_{l_1}(1+\sum_{j \in L_1} \xi_j)}\prod_{j \in L_1} h_j\left(\frac{iz_{l_1}}{N\mathcal{T}}\right)\frac{2}{(z_{k_1}-z_{k_2})(z_{k_1}-z_{l_1})(z_{k_2}-z_{l_1})}\right] \\ 
           &=-e^{z_{k_2}(1 + \sum_{j \in L_1} \xi_{j})} \prod_{j \in L_1}h_{j}\left(\frac{i z_{k_2}}{N\mathcal{T}}\right) \frac{2}{(z_{k_2}-z_{k_1})^2}.
    \end{aligned}
\end{equation}
Inserting these results back into (\ref{three_integrals}) and moving terms around we have
\begin{align}
    \label{two_integrals} 
    &\left(1+O\left(\frac{1}{N}\right)\right)(2\pi i)^{|R^c \backslash\{k_1,k_2\} \cup Q^c|}\int_{(\delta_{k_1})}\int_{(\delta_{k_2})}
    \left(e^{z_{k_1}(-1 + \sum_{j \in K_1} \xi_{j})}  \prod_{j \in K_1}h_{j}\left(\frac{i z_{k_1}}{N\mathcal{T}}\right) \right) \notag\\ \notag\\ 
    &\times \left(e^{z_{k_2}(-1 + \sum_{j \in K_2} \xi_{j})} \prod_{j \in K_2}h_{j}\left(\frac{i z_{k_2}}{N\mathcal{T}}\right) \right)\notag\\ \notag\\ 
    &\times \Bigg\{
    -\left(1+\sum_{j \in L_1}\xi_j\right)\left(1 + \sum_{j \in L_2} \xi_{j}\right)e^{z_{k_1}(1 + \sum_{j \in L_2} \xi_{j})} \prod_{j \in L_2}h_{j}\left(\frac{i z_{k_1}}{N\mathcal{T}}\right)
    e^{z_{k_2}(1 + \sum_{j \in L_1} \xi_{j})} \prod_{j \in L_1}h_{j}\left(\frac{i z_{k_2}}{N\mathcal{T}}\right)   \notag\\ \notag\\ 
    &+e^{z_{k_1}(2 + \sum_{j \in L_1 \cup L_2} \xi_{j})} \prod_{j \in L_1 \cup L_2}h_{j}\left(\frac{i z_{k_1}}{N\mathcal{T}}\right) \frac{2}{(z_{k_2}-z_{k_1})^2}\notag\\ \notag\\ 
    &-e^{z_{k_1}(1 + \sum_{j \in L_2} \xi_{j})} \prod_{j \in L_2}h_{j}\left(\frac{i z_{k_1}}{N\mathcal{T}}\right)
    e^{z_{k_2}(1 + \sum_{j \in L_1} \xi_{j})} \prod_{j \in L_1}h_{j}\left(\frac{i z_{k_2}}{N\mathcal{T}}\right)\frac{2}{(z_{k_2}-z_{k_1})^2} \notag\\ \notag\\ 
    &-\left(1+\sum_{j \in L_1}\xi_j\right)\left(1 + \sum_{j \in L_2} \xi_{j}\right)
    e^{z_{k_1}(1 + \sum_{j \in L_1} \xi_{j})} \prod_{j \in L_1}h_{j}\left(\frac{i z_{k_1}}{N\mathcal{T}}\right)e^{z_{k_2}(1 + \sum_{j \in L_2} \xi_{j})} \prod_{j \in L_2}h_{j}\left(\frac{i z_{k_2}}{N\mathcal{T}}\right) \notag\\ \notag\\ 
    &- e^{z_{k_1}(1 + \sum_{j \in L_1} \xi_{j})} \prod_{j \in L_1}h_{j}\left(\frac{i z_{k_1}}{N\mathcal{T}}\right)e^{z_{k_2}(1 + \sum_{j \in L_2} \xi_{j})} \prod_{j \in L_2}h_{j}\left(\frac{i z_{k_2}}{N\mathcal{T}}\right)\frac{2}{(z_{k_1}-z_{k_2})^2} \notag\\ \notag\\ 
    &+ e^{z_{k_2}(2 + \sum_{j \in L_1 \cup L_2} \xi_{j})} \prod_{j \in L_1 \cup L_2}h_{j}\left(\frac{i z_{k_2}}{N\mathcal{T}}\right)\frac{2}{(z_{k_1}-z_{k_2})^2}
    \Bigg\} dz_{k_2} dz_{k_1},
\end{align}
where we have now accumulated multiple conditions on the \(\xi\) integrals. We require that \(\xi_j < 0\)
for \(j \in R_1^{<k_1} \cup R_2^{<k_2} \cup Q_1^{>l_1} \cup Q_2^{>l_2} \cup Q_1^c \cup Q_2^c\) and \(\xi_j >0\) for \(j \in R_1^{>k_1} \cup R_2^{>k_2} \cup Q_1^{<l_1} \cup Q_2^{<l_2} \cup R_1^c \cup R_2^c\). Also \(1+ \sum_{j \in L_1} \xi_j <0\) and \(1+ \sum_{j \in L_2} \xi_j <0\). If any of these conditions fails, the entire expression is zero. At this point we use the fact that the \(z_{k_2}\) integral lies on the right of the \(z_{k_1}\) integral because $k_2>k_1$. 
Notice that there are now six terms in the curly braces that involve the \(z_{k_2}\) variable, therefore, we will split the \(z_{k_2}\) integral into six separate ones, starting with the following which combines the functions containing \(z_{k_2}\) variable from lines 2 and 3 in (\ref{two_integrals})
\begin{equation}
    \label{zk2_integral_1}
    \begin{aligned} 
        \int_{(\delta_{k_2})} &e^{z_{k_2}(\sum_{j \in K_2 \cup L_1} \xi_{j})} \prod_{j \in K_2 \cup L_1}h_{j}\left(\frac{i z_{k_2}}{N\mathcal{T}}\right)dz_{k_2} \\ \\
        &= (-2\pi i N\mathcal{T})\int_{\mathbb{R}^{|K_2\backslash\{k_2\} \cup L_1|}} \prod_{j \in K_2\backslash\{k_2\} \cup L_1} g_{j}(u_j) \,g_{k_2}\left(N\mathcal{T}\left(\sum_{j \in K_2 \cup L_1} \xi_j\right) - \sum_{j \in K_2\backslash\{k_2\} \cup L_1}u_j\right) d\mathbf{u},
    \end{aligned}
\end{equation}
where we used our result on the generalised version of the convolution theorem (\ref{change_of_vars_result}). Another integral combines lines 2 and 4
\begin{equation}
    \label{zk2_integral_2}
    \begin{aligned} 
        \int_{(\delta_{k_2})} &e^{z_{k_2}(-1 + \sum_{j \in K_2} \xi_{j})} \prod_{j \in K_2}h_{j}\left(\frac{i z_{k_2}}{N\mathcal{T}}\right) \frac{2}{(z_{k_2} - z_{k_1})^2} dz_{k_2}\\ \\
        &= 2\left(1+O\left(\frac{1}{N\mathcal{T}}\right)\right)\left(-1 + \sum_{j \in K_2}\xi_j \right) e^{z_{k_1}\left(-1+\sum_{j \in K_2}\xi_j\right)}\prod_{j \in K_2}h_{j}\left(\frac{i z_{k_1}}{N\mathcal{T}}\right),
    \end{aligned}
\end{equation}
where we have a double pole at \(z_{k_1}\) and assume that \((-1 + \sum_{j \in K_2} \xi_{j}) > 0\). Next we combine lines 2 and 5 to get
\begin{equation}
    \label{zk2_integral_3}
    \begin{aligned} 
        \int_{(\delta_{k_2})} &e^{z_{k_2}(\sum_{j \in K_2 \cup L_1} \xi_{j})} \prod_{j \in K_2 \cup L_1}h_{j}\left(\frac{i z_{k_2}}{N\mathcal{T}}\right) \frac{2}{(z_{k_1}- z_{k_2})^2} dz_{k_2}\\ \\
        &= 2\left(1+O\left(\frac{1}{N\mathcal{T}}\right)\right)\left(\sum_{j \in K_2 \cup L_1}\xi_j \right) e^{z_{k_1}\left(\sum_{j \in K_2 \cup L_1}\xi_j\right)}\prod_{j \in K_2 \cup L_1}h_{j}\left(\frac{i z_{k_1}}{N\mathcal{T}}\right),
    \end{aligned}
\end{equation}
where again we have a double pole at \(z_{k_1}\) and condition \(\sum_{j \in K_2 \cup L_1} \xi_{j} > 0\). Combining lines 2 and 6 we have 
\begin{equation}
    \label{zk2_integral_4}
    \begin{aligned} 
        \int_{(\delta_{k_2})} &e^{z_{k_2}(\sum_{j \in K_2 \cup L_2} \xi_{j})} \prod_{j \in K_2 \cup L_2}h_{j}\left(\frac{i z_{k_2}}{N\mathcal{T}}\right)dz_{k_2} \\ \\
        &= (-2\pi i N\mathcal{T})\int_{\mathbb{R}^{|K_2\backslash\{k_2\} \cup L_2|}} \prod_{j \in K_2\backslash\{k_2\} \cup L_2} g_{j}(u_j) \,g_{k_2}\left(N\mathcal{T}\left(\sum_{j \in K_2 \cup L_2} \xi_j\right) - \sum_{j \in K_2\backslash\{k_2\} \cup L_2}u_j\right) d\mathbf{u},
    \end{aligned}
\end{equation}
and lines 2 and 7 to obtain
\begin{equation}
    \label{zk2_integral_5}
    \begin{aligned} 
        \int_{(\delta_{k_2})} &e^{z_{k_2}(\sum_{j \in K_2 \cup L_2} \xi_{j})} \prod_{j \in K_2 \cup L_2}h_{j}\left(\frac{i z_{k_2}}{N\mathcal{T}}\right) \frac{2}{(z_{k_1}- z_{k_2})^2} dz_{k_2}\\ \\
        &= 2\left(1+O\left(\frac{1}{N\mathcal{T}}\right)\right)\left(\sum_{j \in K_2 \cup L_2}\xi_j \right) e^{z_{k_1}\left(\sum_{j \in K_2 \cup L_2}\xi_j\right)}\prod_{j \in K_2 \cup L_2}h_{j}\left(\frac{i z_{k_1}}{N\mathcal{T}}\right),
    \end{aligned}
\end{equation}
with condition \(\sum_{j \in K_2 \cup L_2} \xi_{j} > 0\) and lastly lines 2 and 8
\begin{equation}
    \label{zk2_integral_6}
    \begin{aligned} 
        \int_{(\delta_{k_2})} &e^{z_{k_2}(1+\sum_{j \in K_2 \cup L_1 \cup L_2} \xi_{j})} \prod_{j \in K_2 \cup L_1 \cup L_2}h_{j}\left(\frac{i z_{k_2}}{N\mathcal{T}}\right) \frac{2}{(z_{k_1}- z_{k_2})^2} dz_{k_2}\\ \\
        &= 2\left(1+O\left(\frac{1}{N\mathcal{T}}\right)\right)\left(1+\sum_{j \in K_2 \cup L_1 \cup L_2}\xi_j \right) e^{z_{k_1}\left(1+\sum_{j \in K_2 \cup L_1 \cup L_2}\xi_j\right)}\prod_{j \in K_2 \cup L_1 \cup L_2}h_{j}\left(\frac{i z_{k_1}}{N\mathcal{T}}\right),
    \end{aligned}
\end{equation}
with condition \(1+\sum_{j \in K_2 \cup L_1 \cup L_2} \xi_{j} > 0\). If for any of these  4 integrals the conditions on the \(\xi\) integrals over \(K_2\) indices fail the conditions specified, the corresponding expression is 0.

Inserting expressions (\ref{zk2_integral_2}) to (\ref{zk2_integral_6}) back into (\ref{two_integrals}) we obtain
\begin{align}
    \label{one_integral} 
    &\left(1+O\left(\frac{1}{N}\right)\right)(2\pi i)^{|R^c \backslash\{k_1\} \cup Q^c|}\int_{(\delta_{k_1})}
    \left(e^{z_{k_1}(-1 + \sum_{j \in K_1} \xi_{j})}  \prod_{j \in K_1}h_{j}\left(\frac{i z_{k_1}}{N\mathcal{T}}\right) \right)\notag\\ \notag\\ 
    &\times \Bigg\{
    \Bigg[N\mathcal{T}\left(1+\sum_{j \in L_1}\xi_j\right)\left(1 + \sum_{j \in L_2} \xi_{j}\right)e^{z_{k_1}(1 + \sum_{j \in L_2} \xi_{j})} \prod_{j \in L_2}h_{j}\left(\frac{i z_{k_1}}{N\mathcal{T}}\right) \notag\\ \notag\\ 
    &\times \int_{-\infty}^{\infty}\dots \int_{-\infty}^{\infty} \prod_{j \in K_2 \backslash\{k_2\} \cup L_1}g_j(u_j)\, g_{k_2}\left(N\mathcal{T}\left(\sum_{j \in K_2 \cup L_1}\xi_j\right)-\sum_{j \in K_2 \backslash\{k_2\} \cup L_1} u_j\right)d \mathbf{u}\Bigg]   \notag\\ \notag\\ 
    &+2\left(-1 + \sum_{j \in K_2} \xi_j\right)e^{z_{k_1}(1 + \sum_{j \in K_2 \cup L_1 \cup L_2} \xi_{j})} \prod_{j \in K_2 \cup L_1 \cup L_2}h_{j}\left(\frac{i z_{k_1}}{N\mathcal{T}}\right) 
    \notag\\ \notag\\ 
    &-2\left(\sum_{j \in K_2 \cup L_1} \xi_j\right)e^{z_{k_1}(1 + \sum_{j \in K_2 \cup L_1 \cup L_2} \xi_{j})} \prod_{j \in K_2 \cup L_1 \cup L_2}h_{j}\left(\frac{i z_{k_1}}{N\mathcal{T}}\right) 
     \notag\\ \notag\\ 
    &+\Bigg[N\mathcal{T}\left(1+\sum_{j \in L_1}\xi_j\right)\left(1 + \sum_{j \in L_2} \xi_{j}\right)
    e^{z_{k_1}(1 + \sum_{j \in L_1} \xi_{j})} \prod_{j \in L_1}h_{j}\left(\frac{i z_{k_1}}{N\mathcal{T}}\right)\notag\\ \notag\\ 
    &\times \int_{-\infty}^{\infty}\dots \int_{-\infty}^{\infty} \prod_{j \in K_2 \backslash\{k_2\} \cup L_2}g_j(u_j)\, g_{k_2}\left(N\mathcal{T}\left(\sum_{j \in K_2 \cup L_2}\xi_j\right)-\sum_{j \in K_2 \backslash\{k_2\} \cup L_2} u_j\right)d \mathbf{u}\Bigg] \notag\\ \notag\\ 
    &- 2\left(\sum_{j \in K_2 \cup L_2} \xi_j\right)e^{z_{k_1}(1 + \sum_{j \in K_2 \cup L_1 \cup L_2} \xi_{j})} \prod_{j \in K_2 \cup L_1 \cup L_2}h_{j}\left(\frac{i z_{k_1}}{N\mathcal{T}}\right)
    \notag\\ \notag\\ 
    &+ 2\left(1 + \sum_{j \in K_2 \cup L_1 \cup L_2} \xi_j\right)e^{z_{k_1}(1 + \sum_{j \in K_2 \cup L_1 \cup L_2} \xi_{j})} \prod_{j \in K_2 \cup L_1 \cup L_2}h_{j}\left(\frac{i z_{k_1}}{N\mathcal{T}}\right)
    \Bigg\} dz_{k_1},
\end{align}
note that the terms with the factor of 2 are all equivalent and only differ in the constant involving the \(\xi\)'s. If the factors in brackets are all positive (meaning all the conditions from those four integrals hold), we can sum up these constants and we see that those four terms in lines 4,5,8, and 9 cancel out. However those conditions might not always hold all at once. The entire expression also carries with it the conditions that \(\xi_j < 0\)
for \(j \in R_1^{<k_1} \cup R_2^{<k_2} \cup Q_1^{>l_1} \cup Q_2^{>l_2} \cup Q_1^c \cup Q_2^c\) and \(\xi_j >0\) for \(j \in R_1^{>k_1} \cup R_2^{>k_2} \cup Q_1^{<l_1} \cup Q_2^{<l_2} \cup R_1^c \cup R_2^c\). Also \(1+ \sum_{j \in L_1} \xi_j <0\) and \(1+ \sum_{j \in L_2} \xi_j <0\). If any of these conditions fails, the expression is 0. Furthermore lines 4, 5, 8, and 9 only exist if the factor in the brackets is positive, otherwise the corresponding line is 0.

Now we deal with the last integral corresponding to the variable \(z_{k_1}\) where, once again, we use the generalised version of the convolution theorem that we have derived earlier at (\ref{change_of_vars_result}), for all six integrals.

\begin{align}
    \label{k1_l2} 
    &\int_{(\delta_{k_1})}
    \left(e^{z_{k_1}(\sum_{j \in K_1 \cup L_2} \xi_{j})}  \prod_{j \in K_1 \cup L_2}h_{j}\left(\frac{i z_{k_1}}{N\mathcal{T}}\right) \right)
     dz_{k_1} \notag\\ \notag\\ 
    &= (-2\pi i N\mathcal{T})\int_{\mathbb{R}^{|K_1\backslash\{k_1\} \cup L_2|}} \prod_{j \in K_1\backslash\{k_1\} \cup L_2} g_{j}(u_j) \,g_{k_1}\left(N\mathcal{T}\left(\sum_{j \in K_1 \cup L_2} \xi_j\right) - \sum_{j \in K_1\backslash\{k_1\} \cup L_2}u_j\right) d\mathbf{u},
\end{align}
where we combined the exponential and product of \(h\) functions from line 1 and line 2 from (\ref{one_integral}). Combining line 1 and line 4, 1 and 5, 1 and 8, and lines 1 and 9, are all the same integrals and we have
\begin{align}
    \label{k1_l2b} 
    &\int_{(\delta_{k_1})}
    \left(e^{z_{k_1}(\sum_{j \in R^c \cup Q^c} \xi_{j})}  \prod_{j \in R^c \cup Q^c}h_{j}\left(\frac{i z_{k_1}}{N\mathcal{T}}\right) \right)
     dz_{k_1} \notag\\ \notag\\ 
    &= (-2\pi i N\mathcal{T})\int_{\mathbb{R}^{|R^c \cup Q^c \backslash\{k_1\}|}} \prod_{j \in R^c \cup Q^c \backslash\{k_1\} } g_{j}(u_j)  \notag\\ \notag\\ 
    &\times g_{k_1}\left(N\mathcal{T}\left(\sum_{j \in R^c \cup Q^c} \xi_j\right) - \sum_{j \in R^c \cup Q^c \backslash\{k_1\} }u_j\right) d\mathbf{u} \notag\\ \notag\\ 
    &= (-2\pi i N\mathcal{T})\int_{\mathbb{R}^{|R^c \cup Q^c \backslash\{k_1\}|}} \prod_{j \in R^c \cup Q^c \backslash\{k_1\}} g_{j}(u_j) g_{k_1}\left(N\mathcal{T}\left(\sum_{j \in R^c \cup Q^c} \xi_j\right) - \sum_{j \in R^c \cup Q^c\backslash\{k_1\}}u_j\right) d\mathbf{u}.
\end{align}
Inserting these back into (\ref{one_integral}) we obtain 

\begin{align}
    \label{no_integral_left} 
    &\left(1+O\left(\frac{1}{N}\right)\right)(2\pi i)^{|R^c \cup Q^c|}
 \Bigg\{
    (N\mathcal{T})^2\left(1 + \sum_{j \in L_1} \xi_{j}\right)\left(1 + \sum_{j \in L_2} \xi_{j}\right) \notag\\ \notag\\ 
    &\times \Bigg[\int_{-\infty}^{\infty}\dots \int_{-\infty}^{\infty} \prod_{j \in K_1 \backslash\{k_1\} \cup L_2}g_j(u_j) g_{k_1}\left(N\mathcal{T}\left(\sum_{j \in K_1 \cup L_2}\xi_j\right)-\sum_{j \in K_1 \backslash\{k_1\} \cup L_2} u_j\right)d \mathbf{u}\notag\\ \notag\\ 
    &\times\int_{-\infty}^{\infty}\dots \int_{-\infty}^{\infty} \prod_{j \in K_2 \backslash\{k_2\} \cup L_1}g_j(u_j)\, g_{k_2}\left(N\mathcal{T}\left(\sum_{j \in K_2 \cup L_1}\xi_j\right)-\sum_{j \in K_2 \backslash\{k_2\} \cup L_1} u_j\right)d \mathbf{u} \notag\\ \notag\\ 
    &+\int_{-\infty}^{\infty}\dots \int_{-\infty}^{\infty} \prod_{j \in K_1 \backslash\{k_1\} \cup L_1}g_j(u_j)g_{k_1}\left(N\mathcal{T}\left(\sum_{j \in K_1 \cup L_1}\xi_j\right)-\sum_{j \in K_1 \backslash\{k_1\} \cup L_1} u_j\right)d \mathbf{u} \notag\\ \notag\\ 
    &\times \int_{-\infty}^{\infty}\dots \int_{-\infty}^{\infty} \prod_{j \in K_2 \backslash\{k_2\} \cup L_2}g_j(u_j)\, g_{k_2}\left(N\mathcal{T}\left(\sum_{j \in K_2 \cup L_2}\xi_j\right)-\sum_{j \in K_2 \backslash\{k_2\} \cup L_2} u_j\right)d \mathbf{u}\Bigg]  \notag\\ \notag\\ 
    &-2N\mathcal{T}\left(-1 + \sum_{j \in K_2} \xi_j\right)\int_{\mathbb{R}^{|R^c \cup Q^c \backslash\{k_1\}|}} \prod_{j \in R^c \cup Q^c \backslash\{k_1\}} g_{j}(u_j) g_{k_1}\left(N\mathcal{T}\left(\sum_{j \in R^c \cup Q^c} \xi_j\right) - \sum_{j \in R^c \cup Q^c\backslash\{k_1\}}u_j\right) d\mathbf{u} \notag\\ \notag\\ 
    &+2N\mathcal{T}\left(\sum_{j \in K_2 \cup L_1} \xi_j\right)\int_{\mathbb{R}^{|R^c \cup Q^c \backslash\{k_1\}|}} \prod_{j \in R^c \cup Q^c \backslash\{k_1\}} g_{j}(u_j) g_{k_1}\left(N\mathcal{T}\left(\sum_{j \in R^c \cup Q^c} \xi_j\right) - \sum_{j \in R^c \cup Q^c\backslash\{k_1\}}u_j\right) d\mathbf{u} \notag\\ \notag\\ 
    &+2N\mathcal{T}\left(\sum_{j \in K_2 \cup L_2} \xi_j\right)\int_{\mathbb{R}^{|R^c \cup Q^c \backslash\{k_1\}|}} \prod_{j \in R^c \cup Q^c \backslash\{k_1\}} g_{j}(u_j) g_{k_1}\left(N\mathcal{T}\left(\sum_{j \in R^c \cup Q^c} \xi_j\right) - \sum_{j \in R^c \cup Q^c\backslash\{k_1\}}u_j\right) d\mathbf{u} 
    \notag\\ \notag\\ 
    &-2N\mathcal{T}\left(1 + \sum_{j \in K_2 \cup L_1 \cup L_2} \xi_j\right)\int_{\mathbb{R}^{|R^c \cup Q^c \backslash\{k_1\}|}} \prod_{j \in R^c \cup Q^c \backslash\{k_1\}} g_{j}(u_j) \notag\\ \notag\\ 
    &\times g_{k_1}\left(N\mathcal{T}\left(\sum_{j \in R^c \cup Q^c} \xi_j\right) - \sum_{j \in R^c \cup Q^c\backslash\{k_1\}}u_j\right) d\mathbf{u}\Bigg\},
\end{align}
with the conditions that \(\xi_j < 0\) for \(j \in R_1^{<k_1} \cup R_2^{<k_2} \cup Q_1^{>l_1} \cup Q_2^{>l_2} \cup Q_1^c \cup Q_2^c\) and  \(\xi_j >0\) for \(j \in R_1^{>k_1} \cup R_2^{>k_2} \cup Q_1^{<l_1} \cup Q_2^{<l_2} \cup R_1^c \cup R_2^c\). Also \(1+ \sum_{j \in L_1} \xi_j <0\) and \(1+ \sum_{j \in L_2} \xi_j <0\). Note that even though the last four integrals are all the same apart from the factors in front of them, they all exist under different conditions. Each one of them exists only when the factor in brackets preceding the integral is positive.
\end{proof}
\renewcommand\qedsymbol{$\blacksquare$}
\begin{proof}[Proof of Theorem \ref{q_3_thm}]
Using Lemmas \ref{lemma_1}, \ref{lemma_2}, and \ref{lemma_q_3} and denoting the six sets of integrals over $d\mathbf{u}$  from Lemma \ref{lemma_q_3} by \(I^1\) to \(I^6\).
\begin{align}
    \label{I_2_2_using_lmms} 
    I&_{2,2}=\left(1+ O(1/N)\right) \sum_{K+L+M=\{1,\dots,n\}}(-1)^{|L|+|M|}N^{|M|}  \int_{\mathbb{R}_{\xi_m}^{|M|}} \prod_{m \in M}(-\mathcal{T}g_m(N\mathcal{T}\xi_m)) \notag\\ \notag\\
    & \sum_{\genfrac{}{}{0pt}{1}{R\cup R^c  = K, \;Q \cup Q^c=L}{|R|=|Q|,\;|R^c|\geq 2,\;|Q^c|\geq 2}} \sum_{(R:Q)} \int_{\mathbb{R}_{\xi_r}^{|R|}} 
    \int_{\mathbb{R}_{\xi_q<0}^{|Q|}} \prod_{q_j \in Q}
    (N\mathcal{T}\xi_{q_j})\int_{\mathbb{R}_{u_q}}g_{q_j}(u_{q_j})g_{r_j}(N\mathcal{T}(\xi_{r_j}-\xi_{q_j})-u_{q_j})du_{q_j}\notag\\ \notag\\
    &\times 
   \sum_{\genfrac{}{}{0pt}{1}{\genfrac{}{}{0pt}{1}{R_1 \cup R_1^c \cup R_2 \cup R_2^c = R^c}{Q_1 \cup Q_1^c \cup Q_2 \cup Q_2^c = Q^c}}{R_1,R_2,Q_1,Q_2\neq \emptyset}} \, \sum_{\genfrac{}{}{0pt}{1}{k_1 \in R_1,\, k_2 \in R_2, k_2>k_1}{l_1 \in Q_1,\, l_2 \in Q_2, l_2>l_1}}  
    (-1)^{|R_{1}^{>k_1} \cup R_{2}^{>k_2} \cup Q_{1}^{>l_1} \cup Q_{2}^{>k_2}|}  \int_{\mathbb{R}_{\xi_j<0}^{|A|}} \int_{\mathbb{R}_{\xi_j>0}^{|B|}}
    \int_{\mathbb{R}_{\xi_{l_2}}}\int_{\mathbb{R}_{\xi_{l_1}}}
    \int_{\mathbb{R}_{\xi_{k_2}}}\int_{\mathbb{R}_{\xi_{k_1}}} \notag\\ \notag\\
    &(I^1 + I^2 + I^3 + I^4 + I^5 + I^6) \, \Phi(\xi_1,\dots,\xi_n) \, \delta\left(\sum_{j=1}^n \xi_j\right) d\xi_1 \dots d\xi_n,
\end{align}
where \(A = R_{1}^{<k_1} \cup R_{2}^{<k_2} \cup Q_{1}^{>l_1} \cup Q_{2}^{>l_2}
\cup Q_{1}^{c} \cup Q_{2}^{c}\)  and \(B = R_{1}^{>k_1} \cup R_{2}^{>k_2} \cup Q_{1}^{<l_1} \cup Q_{2}^{<l_2}
\cup R_{1}^{c} \cup R_{2}^{c}\). At this point the expressions involved are very long and therefore we will define $I_{2,2}^{\epsilon}$ to be the term in (\ref{I_2_2_using_lmms}) that includes $I^\epsilon$, with $\epsilon=1,2,3,4,5$ or $6$.
 For example for \(\epsilon = 1\) defines
\begin{align}
    \label{I_2_2_1} 
    I&_{2,2}^{1} =\left(1+ O(1/N)\right) \sum_{K+L+M=\{1,\dots,n\}}(-1)^{|L|+|M|}N^{|M|}  \int_{\mathbb{R}_{\xi_m}^{|M|}} \prod_{m \in M}(-\mathcal{T}g_m(N\mathcal{T}\xi_m)) \notag\\ \notag\\
    &  \sum_{\genfrac{}{}{0pt}{1}{R\cup R^c  = K, \;Q \cup Q^c=L}{|R|=|Q|,\;|R^c|\geq 2,\;|Q^c|\geq 2}} \sum_{(R:Q)} \int_{\mathbb{R}_{\xi_r}^{|R|}} 
    \int_{\mathbb{R}_{\xi_q<0}^{|Q|}} \prod_{q_j \in Q}
    (N\mathcal{T}\xi_{q_j})\int_{\mathbb{R}_{u_q}}g_{q_j}(u_{q_j})g_{r_j}(N\mathcal{T}(\xi_{r_j}-\xi_{q_j})-u_{q_j})du_{q_j}\notag\\ \notag\\
    &\times 
   \sum_{\genfrac{}{}{0pt}{1}{\genfrac{}{}{0pt}{1}{R_1 \cup R_1^c \cup R_2 \cup R_2^c = R^c}{Q_1 \cup Q_1^c \cup Q_2 \cup Q_2^c = Q^c}}{R_1,R_2,Q_1,Q_2\neq \emptyset}} \, \sum_{\genfrac{}{}{0pt}{1}{k_1 \in R_1,\, k_2 \in R_2, k_2>k_1}{l_1 \in Q_1,\, l_2 \in Q_2, l_2>l_1}}  
    (-1)^{|R_{1}^{>k_1} \cup R_{2}^{>k_2} \cup Q_{1}^{>l_1} \cup Q_{2}^{>k_2}|}   \int_{\mathbb{R}_{\xi_j<0}^{|A|}} \int_{\mathbb{R}_{\xi_j>0}^{|B|}}
    \int_{\mathbb{R}_{l_2}}\int_{\mathbb{R}_{l_1}}
    \int_{\mathbb{R}_{k_2}}\int_{\mathbb{R}_{k_1}} 
     \notag\\ \notag\\
    &\times
    (N\mathcal{T})^2\left(1 + \sum_{j \in L_1} \xi_{j}\right)\left(1 + \sum_{j \in L_2} \xi_{j}\right)\notag\\ \notag\\
    &\times
    \int_{\mathbb{R}_{u}^{K_1\backslash\{k_1\} \cup L_2}} \prod_{j \in K_1\backslash\{k_1\} \cup L_2}g_j(u_j)
    g_{k_1}\left(N\mathcal{T}\left(\sum_{j \in K_1 \cup L_2}\xi_j\right)- \sum_{j \in K_1\backslash\{k_1\} \cup L_2}u_j\right)d\mathbf{u} \displaybreak[0]\notag\\ \notag\\
    &\times \int_{\mathbb{R}_{u}^{K_2\backslash\{k_2\} \cup L_1}}
    \prod_{j \in K_2\backslash\{k_2\} \cup L_1}g_j(u_j)
    g_{k_2}\left(N\mathcal{T}\left(\sum_{j \in K_2 \cup L_1}\xi_j\right)- \sum_{j \in K_2\backslash\{k_2\} \cup L_1}u_j\right)d\mathbf{u}  \notag\\ \notag\\
    &\times 
    \, \Phi(\xi_1,\dots,\xi_n) \, \delta\left(\sum_{j=1}^n \xi_j\right) d\xi_1 \dots d\xi_n.
\end{align}
Now we apply the following change of variables 
\begin{equation}
    \label{I_2_2_change_of_vars_def}
    \begin{aligned}
        N\mathcal{T}\xi_m &= w_m \quad\text{for } m \in M, \\
        u_{q_j} &= w_{q_j} \quad \text{for } q_j \in Q \cup K_1\backslash\{k_1\} \cup L_1 \cup K_2\backslash\{k_2\} \cup L_2, \\
        N\mathcal{T}(\xi_{r_j} - \xi_{q_j}) - u_{q_j} &= w_{r_j} \quad \text{for } r_j \in R, \\
        N\mathcal{T}\left(\sum_{j \in K_1 \cup L_2} \xi_j\right) - \sum_{j \in K_1\backslash\{k_1\} \cup L_2} u_j &= w_{k_1} \quad \text{in line 4 in (\ref{I_2_2_1})}, \\
        N\mathcal{T}\left(\sum_{j \in K_2 \cup L_1} \xi_j\right) - \sum_{j \in K_2\backslash\{k_2\} \cup L_1} u_j &= w_{k_2} \quad \text{in line 5 in (\ref{I_2_2_1})}.
    \end{aligned}
\end{equation}
From this we have
\begin{equation}
    \label{xi_k1_w_k1}
\xi_{k_1} = \sum_{j \in K_1 \cup L_2} \frac{w_j}{N\mathcal{T}} -\left(\sum_{j \in K_1\backslash\{k_1\} \cup L_2} \xi_j\right) 
\end{equation}
and
\begin{equation}
    \label{xi_r_w_r}
\xi_{k_2} = \sum_{j \in K_2 \cup L_1} \frac{w_j}{N\mathcal{T}} -\left(\sum_{j \in K_2\backslash\{k_2\} \cup L_1} \xi_j\right)
\end{equation}
and
\begin{equation}
    \label{xi_k2_w_k2}
\xi_{r_j} = \frac{w_{r_j}+w_{q_j}}{N\mathcal{T}} - \xi_{q_j}
\end{equation}
and we get
\begin{align}
    \label{I_2_2_1_def_chng} 
    I&_{2,2}^1=\left(1+ O(1/N)\right) \sum_{K+L+M=\{1,\dots,n\}}(-1)^{|L|} \int_{\mathbb{R}_{w}^{n}} \prod_{j=1}^{n}g_j(w_j)\, \delta\left(\sum_{j=1}^n \frac{w_j}{N\mathcal{T}}\right)\notag\\ \notag\\
    &  \sum_{\genfrac{}{}{0pt}{1}{R\cup R^c  = K, \;Q \cup Q^c=L}{|R|=|Q|,\;|R^c|\geq 2,\;|Q^c|\geq 2}} \sum_{(R:Q)}
    \int_{\mathbb{R}_{\xi_q<0}^{|Q|}} \prod_{q_j \in Q}
    (\xi_{q_j})\notag\\ \notag\\
    &\times 
    \sum_{\genfrac{}{}{0pt}{1}{\genfrac{}{}{0pt}{1}{R_1 \cup R_1^c \cup R_2 \cup R_2^c = R^c}{Q_1 \cup Q_1^c \cup Q_2 \cup Q_2^c = Q^c}}{R_1,R_2,Q_1,Q_2\neq \emptyset}} \, \sum_{\genfrac{}{}{0pt}{1}{k_1 \in R_1,\, k_2 \in R_2, k_2>k_1}{l_1 \in Q_1,\, l_2 \in Q_2, l_2>l_1}}  
    (-1)^{|R_{1}^{>k_1} \cup R_{2}^{>k_2} \cup Q_{1}^{>l_1} \cup Q_{2}^{>k_2}|} \notag\\ \notag \\
   & \times \int_{\mathbb{R}_{\xi_j<0}^{|A|}} \int_{\mathbb{R}_{\xi_j>0}^{|B|}}
    \int_{\mathbb{R}_{l_2}}\int_{\mathbb{R}_{l_1}}\left(1 + \sum_{j \in L_1} \xi_{j}\right)\left(1 + \sum_{j \in L_2} \xi_{j}\right) \displaybreak[0] \notag\\ \notag\\
    &\times \Phi\Bigg(\sum_{j \in R^c \backslash\{k_1,k_2\} \cup L} \xi_j \mathbf{e}_j + \sum_{m \in M}\frac{w_m}{N\mathcal{T}}\mathbf{e}_m + \sum_{j=1}^{|R|} \left(\left(\frac{w_{r_j}+w_{q_j}}{N\mathcal{T}}\right)-\xi_{q_j}\right) \mathbf{e}_{r_j}\notag\\ \notag\\
    &+ \Bigg(\sum_{j \in K_1 \cup L_2}\frac{w_j}{N\mathcal{T}}- \sum_{j\in K_1\backslash\{k_1\} \cup L_2} \xi_j\Bigg) \mathbf{e}_{k_1} + \Bigg(\sum_{j \in K_2 \cup L_1}\frac{w_j}{N\mathcal{T}}- \sum_{j \in K_2\backslash\{k_2\} \cup L_1} \xi_j\Bigg) \mathbf{e}_{k_2} \Bigg) \notag\\ \notag\\
&  d\xi_1 \dots d\xi_n.
\end{align}
Just as previously, we will Taylor expand the \(\Phi\) function which further simplifies the expression to
\begin{align}
    \label{I_2_2_1_def_chng_pt2} 
    I&_{2,2}^1=\left(1+ O(1/N)\right) \sum_{K+L+M=\{1,\dots,n\}}(-1)^{|L|} \int_{\mathbb{R}_{w}^{n}} \prod_{j=1}^{n}g_j(w_j)\, \delta\left(\sum_{j=1}^n \frac{w_j}{N\mathcal{T}}\right)\notag\\ \notag\\
    &  \sum_{\genfrac{}{}{0pt}{1}{R\cup R^c  = K, \;Q \cup Q^c=L}{|R|=|Q|,\;|R^c|\geq 2,\;|Q^c|\geq 2}} \sum_{(R:Q)}
    \int_{\mathbb{R}_{\xi_q<0}^{|Q|}} \prod_{q_j \in Q}
    (\xi_{q_j})\notag\\ \notag\\
    &\times 
    \sum_{\genfrac{}{}{0pt}{1}{\genfrac{}{}{0pt}{1}{R_1 \cup R_1^c \cup R_2 \cup R_2^c = R^c}{Q_1 \cup Q_1^c \cup Q_2 \cup Q_2^c = Q^c}}{R_1,R_2,Q_1,Q_2\neq \emptyset}} \, \sum_{\genfrac{}{}{0pt}{1}{k_1 \in R_1,\, k_2 \in R_2, k_2>k_1}{l_1 \in Q_1,\, l_2 \in Q_2, l_2>l_1}}  
    (-1)^{|R_{1}^{>k_1} \cup R_{2}^{>k_2} \cup Q_{1}^{>l_1} \cup Q_{2}^{>k_2}|}\notag\\ \notag \\
    &\times \int_{\mathbb{R}_{\xi_j<0}^{|A|}} \int_{\mathbb{R}_{\xi_j>0}^{|B|}}
    \int_{\mathbb{R}_{l_2}}\int_{\mathbb{R}_{l_1}} \left(1+ \sum_{j \in L_1}\xi_j\right)
    \left(1+ \sum_{j \in L_2}\xi_j\right) \notag\\ \notag\\
    &\Phi\Bigg(\sum_{j \in R^c \backslash\{k_1,k_2\} \cup L} \xi_j \mathbf{e}_j + \sum_{j=1}^{|R|} \left(-\xi_{q_j}\right) \mathbf{e}_{r_j}+ \Bigg(- \sum_{j\in K_1\backslash\{k_1\} \cup L_2} \xi_j\Bigg) \mathbf{e}_{k_1} + \Bigg(- \sum_{j \in K_2\backslash\{k_2\} \cup L_1} \xi_j\Bigg) \mathbf{e}_{k_2} \Bigg)\notag\\ \notag\\
    & d\xi \, d\mathbf{w},
\end{align}
where we can once again make use of the delta function as in the previous section (expression \ref{prod_of_gs})
 \begin{align}
        \label{prod_of_gs_as_kappa_pt2}
        \int_{\mathbb{R}^{n}} \prod_{j=1}^n g_j(w_j) \, \delta\left(\sum_{j =1}^n \frac{w_j}{N\mathcal{T}}\right) d\mathbf{w} &= 
         \frac{N\mathcal{T}}{2\pi} \kappa (\mathbf{h}).
    \end{align}
Inserting this back into \(I_{2,2}^1\)  we finally have
\begin{align}
    \label{I_2_2_1_almost_final_result} 
    I&_{2,2}^1=\frac{N\mathcal{T}}{2\pi} \kappa (\mathbf{h}) \sum_{K+L+M=\{1,\dots,n\}}(-1)^{|L|} \sum_{\genfrac{}{}{0pt}{1}{R\cup R^c  = K, \;Q \cup Q^c=L}{|R|=|Q|,\;|R^c|\geq 2,\;|Q^c|\geq 2}} \sum_{(R:Q)}
    \int_{\mathbb{R}_{\xi_q<0}^{|Q|}} \prod_{q_j \in Q}
    (\xi_{q_j})\notag\\ \notag\\
    &\times 
    \sum_{\genfrac{}{}{0pt}{1}{\genfrac{}{}{0pt}{1}{R_1 \cup R_1^c \cup R_2 \cup R_2^c = R^c}{Q_1 \cup Q_1^c \cup Q_2 \cup Q_2^c = Q^c}}{R_1,R_2,Q_1,Q_2\neq \emptyset}} \, \sum_{\genfrac{}{}{0pt}{1}{k_1 \in R_1,\, k_2 \in R_2, k_2>k_1}{l_1 \in Q_1,\, l_2 \in Q_2, l_2>l_1}}  
    (-1)^{|R_{1}^{>k_1} \cup R_{2}^{>k_2} \cup Q_{1}^{>l_1} \cup Q_{2}^{>k_2}|}\notag\\ \notag\\
    &\times \int_{\mathbb{R}_{\xi_j<0}^{|A|}} \int_{\mathbb{R}_{\xi_j>0}^{|B|}}
    \int_{\mathbb{R}_{l_2}}\int_{\mathbb{R}_{l_1}} \left(1+ \sum_{j \in L_1}\xi_j\right)
    \left(1+ \sum_{j \in L_2}\xi_j\right) \notag\\ \notag\\
    &\Phi\Bigg(\sum_{j \in R^c \backslash\{k_1,k_2\} \cup L} \xi_j \mathbf{e}_j + \sum_{j=1}^{|R|} \left(-\xi_{q_j}\right) \mathbf{e}_{r_j}+ \Bigg(- \sum_{j\in K_1\backslash\{k_1\} \cup L_2} \xi_j\Bigg) \mathbf{e}_{k_1} + \Bigg(- \sum_{j \in K_2\backslash\{k_2\} \cup L_1} \xi_j\Bigg) \mathbf{e}_{k_2} \Bigg)
      d\xi \notag\\ \notag\\
    &+ O(\mathcal{T}),
\end{align}
with the conditions  \(1+\sum_{j \in L_1}\xi_j < 0\) and \(1+ \sum_{j \in L_2} \xi_j < 0\).
%   and \(\sum_{j \in K_2 \cup L_1 \cup L_2} \xi_j > 0\).

Lastly we will change the variable \(\xi_j \to -\xi_j\) for \(j \in A \cup Q\), to get
\begin{align}
    \label{I_2_2_1_final_result} 
    I&_{2,2}^1=\frac{N\mathcal{T}}{2\pi} \kappa (\mathbf{h}) \sum_{K+L+M=\{1,\dots,n\}}(-1)^{|L|} \sum_{\genfrac{}{}{0pt}{1}{R\cup R^c  = K, \;Q \cup Q^c=L}{|R|=|Q|,\;|R^c|\geq 2,\;|Q^c|\geq 2}}\sum_{(R:Q)}
    \int_{\mathbb{R}_{\xi_q>0}^{|Q|}} \prod_{q_j \in Q}
    (-\xi_{q_j})\notag\\ \notag\\
    &\times 
     \sum_{\genfrac{}{}{0pt}{1}{\genfrac{}{}{0pt}{1}{R_1 \cup R_1^c \cup R_2 \cup R_2^c = R^c}{Q_1 \cup Q_1^c \cup Q_2 \cup Q_2^c = Q^c}}{R_1,R_2,Q_1,Q_2\neq \emptyset}} \, \sum_{\genfrac{}{}{0pt}{1}{k_1 \in R_1,\, k_2 \in R_2, k_2>k_1}{l_1 \in Q_1,\, l_2 \in Q_2, l_2>l_1}}  
    (-1)^{|R_{1}^{>k_1} \cup R_{2}^{>k_2} \cup Q_{1}^{>l_1} \cup Q_{2}^{>k_2}|}  \notag\\ \notag\\
    &\times\Bigg[ \int_{\mathbb{R}_{\xi_j>0}^{|R^c \cup Q^c\backslash\{k_1,k_2,l_1,l_2\}|}}
    \int_{\mathbb{R}_{l_2}}\int_{\mathbb{R}_{l_1}} \left(1+ \sum_{j \in Q_{1}^{\leq l_1} \cup R_1^c}\xi_j - \sum_{j \in Q_{1}^{>l_1}}\xi_j\right)\left(1+ \sum_{j \in Q_{2}^{\leq l_2} \cup R_2^c}\xi_j - \sum_{j \in Q_{2}^{>l_2}}\xi_j\right) \notag\\ \notag\\
    &\Phi\Bigg\{\sum_{j \in R^c \backslash\{k_1,k_2\} \cup L} \xi_j \mathbf{e}_j + \sum_{j=1}^{|R|} \left(\xi_{q_j}\right) \mathbf{e}_{r_j}+ \Bigg(\sum_{j\in R_{1}^{<k_1} \cup Q_1^c \cup Q_2^{>l_2}} \xi_j - \sum_{j\in R_1^{>k_1} \cup Q_2^{\leq l_2} \cup R_2^c} \xi_j\Bigg) \mathbf{e}_{k_1}\notag\\ \notag\\
    &+ \Bigg(\sum_{j\in R_{2}^{<k_2} \cup Q_2^c \cup Q_1^{>l_1}} \xi_j - \sum_{j\in R_2^{>k_2} \cup Q_1^{\leq l_1} \cup R_1^c} \xi_j\Bigg) \mathbf{e}_{k_2} \Bigg\}  d\xi \Bigg]+ O(\mathcal{T}),
\end{align}
with conditions on the \(\xi\) integrals \(1+\sum_{j \in Q_1^{\leq l_1} \cup R_1^c}\xi_j <\sum_{j \in Q_1^{>l_1}}\),   \(1+\sum_{j \in Q_2^{\leq l_2} \cup R_2^c}\xi_j <\sum_{j \in Q_2^{>l_2}}\).

If we follow these same steps for \(I_{2,2}^2\) to \(I_{2,2}^6\), of course with slight modification to the change of variables, the quantity in square brackets in (\ref{I_2_2_1_final_result}) is replaced with the following.  For $I_{2,2}^2$ the square brackets are replaced with
\begin{align}
    \label{I_2_2_2_final_result} 
   & \int_{\mathbb{R}_{\xi_j>0}^{|R^c \cup Q^c\backslash\{k_1,k_2,l_1,l_2\}|}}
    \int_{\mathbb{R}_{l_2}}\int_{\mathbb{R}_{l_1}}  \times \left(1+ \sum_{j \in Q_{1}^{\leq l_1} \cup R_1^c}\xi_j - \sum_{j \in Q_{1}^{>l_1}}\xi_j\right)\left(1+ \sum_{j \in Q_{2}^{\leq l_2} \cup R_2^c}\xi_j - \sum_{j \in Q_{2}^{>l_2}}\xi_j\right) \notag\\ \notag\\
    &\times \Phi\Bigg\{\sum_{j \in R^c \backslash\{k_1,k_2\} \cup L} \xi_j \mathbf{e}_j + \sum_{j=1}^{|R|} \left(\xi_{q_j}\right) \mathbf{e}_{r_j}+ \Bigg(\sum_{j\in R_{1}^{<k_1} \cup Q_1^c \cup Q_1^{>l_1}} \xi_j - \sum_{j\in R_1^{>k_1} \cup Q_1^{\leq l_1} \cup R_1^c} \xi_j\Bigg) \mathbf{e}_{k_1}\notag\\ \notag\\
    &+ \Bigg(\sum_{j\in R_{2}^{<k_2} \cup Q_2^c \cup Q_2^{>l_2}} \xi_j - \sum_{j\in R_2^{>k_2} \cup Q_2^{\leq l_2} \cup R_2^c} \xi_j\Bigg) \mathbf{e}_{k_2} \Bigg\}  d\xi,
\end{align}
with the same conditions as \(I_{2,2}^1\). 

For $I_{2,2}^3$, we have the expression (\ref{I_2_2_1_final_result}) where the square brackets are replaced with
\begin{align}
    \label{I_2_2_3_final_result} 
    &-2 \int_{\mathbb{R}_{\xi_j>0}^{|R^c \cup Q^c\backslash\{k_1,k_2,l_1,l_2\}|}}
    \int_{\mathbb{R}_{l_2}}\int_{\mathbb{R}_{l_1}} \int_{\mathbb{R}_{k_2}} \left(\sum_{j \in R_2^{\geq k_2}} \xi_j - 1 - \sum_{j \in R_2^{< k_2} \cup Q_2^c} \xi_j\right) \notag\\ \notag\\
    &\Phi\Bigg\{\sum_{j \in R^c \backslash\{k_1\} \cup L} \xi_j \mathbf{e}_j + \sum_{j=1}^{|R|} \left(\xi_{q_j}\right) \mathbf{e}_{r_j}+ \Bigg(\sum_{j\in A} \xi_j - \sum_{j\in B \cup \{k_2, l_1, l_2\}} \xi_j\Bigg) \mathbf{e}_{k_1}\Bigg\}  d\xi ,
\end{align}
with the same condition as \(I_{2,2}^1\) and that \(\sum_{j \in R_2^{\geq k_2}} \xi_j > 1 + \sum_{j \in R_2^{< k_2} \cup Q_2^c} \xi_j\).

$I_{2,2}^4$ looks like (\ref{I_2_2_1_final_result}) but with the following replacing the square brackets
\begin{align}
    \label{I_2_2_4_final_result} 
&2 \int_{\mathbb{R}_{\xi_j>0}^{|R^c \cup Q^c\backslash\{l_1,l_2\}|}}
    \int_{\mathbb{R}_{l_2}}\int_{\mathbb{R}_{l_1}} \int_{\mathbb{R}_{k_2}} \left(\sum_{j \in R_2^{\geq k_2} \cup R_1^c \cup Q_1^{\leq l_1}} \xi_j - \sum_{j \in R_2^{< k_2} \cup Q_2^c  \cup Q_1^{> l_1}} \xi_j\right)  \notag\\ \notag\\
    &\times \Phi\Bigg\{\sum_{j \in R^c \backslash\{k_1\} \cup L} \xi_j \mathbf{e}_j + \sum_{j=1}^{|R|} \left(\xi_{q_j}\right) \mathbf{e}_{r_j}+ \Bigg(\sum_{j\in A} \xi_j - \sum_{j\in B \cup \{k_2, l_1, l_2\}} \xi_j\Bigg) \mathbf{e}_{k_1}\Bigg\}  d\xi,
\end{align}
with the same condition as \(I_{2,2}^1\) and that \(\sum_{j \in R_2^{\geq k_2} \cup R_1^c \cup Q_1^{\leq l_1}} \xi_j> \sum_{j \in R_2^{< k_2} \cup Q_2^c  \cup Q_1^{> l_1}} \xi_j\).

For $I_{2,2}^5$ the square brackets are replaced with
\begin{align}
    \label{I_2_2_5_final_result} 
&2 \int_{\mathbb{R}_{\xi_j>0}^{|R^c \cup Q^c\backslash\{l_1,l_2\}|}}
    \int_{\mathbb{R}_{l_2}}\int_{\mathbb{R}_{l_1}} \int_{\mathbb{R}_{k_2}} \left(\sum_{j \in R_2^{\geq k_2} \cup R_2^c \cup Q_2^{\leq l_2}} \xi_j- \sum_{j \in R_2^{< k_2} \cup Q_2^c  \cup Q_2^{> l_2}} \xi_j\right) \notag\\ \notag\\
    &\Phi\Bigg\{\sum_{j \in R^c \backslash\{k_1\} \cup L} \xi_j \mathbf{e}_j + \sum_{j=1}^{|R|} \left(\xi_{q_j}\right) \mathbf{e}_{r_j}+ \Bigg(\sum_{j\in A} \xi_j - \sum_{j\in B \cup \{k_2, l_1, l_2\}} \xi_j\Bigg) \mathbf{e}_{k_1}\Bigg\}  d\xi,
\end{align}
with the same condition as \(I_{2,2}^1\) and that \(\sum_{j \in R_2^{\geq k_2} \cup R_2^c \cup Q_2^{\leq l_2}} \xi_j> \sum_{j \in R_2^{< k_2} \cup Q_2^c  \cup Q_2^{> l_2}} \xi_j\). 

And lastly
\begin{align}
    \label{I_2_2_6_final_result} 
&-2 \int_{\mathbb{R}_{\xi_j>0}^{|R^c \cup Q^c\backslash\{l_1,l_2\}|}}
    \int_{\mathbb{R}_{l_2}}\int_{\mathbb{R}_{l_1}} \int_{\mathbb{R}_{k_2}} \left(1 +\sum_{j \in R_2^{\geq k_2} \cup R_1^c \cup Q_1^{\leq l_1} } \xi_j- \sum_{j \in R_2^{< k_2} \cup Q_2^c  \cup Q_1^{> l_1} } \xi_j\right) \notag\\ \notag\\
    &\Phi\Bigg\{\sum_{j \in R^c \backslash\{k_1\} \cup L} \xi_j \mathbf{e}_j + \sum_{j=1}^{|R|} \left(\xi_{q_j}\right) \mathbf{e}_{r_j}+ \Bigg(\sum_{j\in A} \xi_j - \sum_{j\in B \cup \{k_2, l_1, l_2\}} \xi_j\Bigg) \mathbf{e}_{k_1}\Bigg\}  d\xi,
\end{align}
with the same condition as \(I_{2,2}^1\) and that \(1 +\sum_{j \in R_2^{\geq k_2} \cup R_1^c \cup Q_1^{\leq l_1}  } \xi_j> \sum_{j \in R_2^{< k_2} \cup Q_2^c  \cup Q_1^{> l_1} } \xi_j\).
\end{proof}

\section{Acknowledgments}
Thanks to Brian Conrey for discussions and encouragement during this project.

%\bibliography{../../../bin/ref.bib}{}

\begin{thebibliography}{10}

\bibitem{kn:aailmz}
L.~Alpoge, N.~Amersi, G.~Iyer, O.~Lazarev, S.J. Miller, and L.~Zhang.
\newblock Maass waveforms and low-lying zeros.
\newblock In {\em Analytic Number Theory: In Honor of {H}elmut {M}aier's 60th
  Birthday}, pages 19--55. Springer, 2015.

\bibitem{kn:alpmil14}
L.~Alpoge and S.J. Miller.
\newblock Low-lying zeros of {M}aass forms {$L$}-functions.
\newblock {\em Int. Math. Res. Not.}, 2015(10):2678--2701, 2015.

\bibitem{kn:balchali}
S.~Baluyot, V.~Chandee, and X~Li.
\newblock Low-lying zeros of a large orthogonal family of automorphic
  $l$-functions.
\newblock ar{X}iv:2310.07606.

\bibitem{kn:bbddm}
O.~Barrett, P.~Burkhardt, J.~De{W}itt, R.~Dorward, and S.J. Miller.
\newblock One-level density for holomorphic cusp forms of arbitrary level.
\newblock {\em Research in Number Theory}, 3, 2017.

\bibitem{kn:betfaz24}
S.~Bettin and A.~Fazzari.
\newblock A weighted one-level density of the non-trivial zeros of the
  {R}iemann zeta-function.
\newblock {\em Mathematische Zeitschrift}, 307(2), 2024.

\bibitem{fourier}
R.~Bracewell.
\newblock {\em The Fourier Transform and Its Applications}.
\newblock McGraw-Hill, New York, third edition, 2000.

\bibitem{chandee}
V.~Chandee and Y.~Lee.
\newblock $n$-level density of the low-lying zeros of primitive {D}irichlet
  $l$-functions.
\newblock {\em Advances in Mathematics}, 369, 2020.

\bibitem{kn:cohen_ea}
P.~Cohen, J.~Dell, O.E. Gonz{\'a}lez, G.~Iyer, S.~Khunger, C.-H. Kwan, S.J.
  Miller, A.~Shashkov, A.~Smith Reina, C.~Sprunger, N.~Triantafillou,
  N.~Truong, R.~Van Peski, S.~Willis, and Y.~Yang.
\newblock Extending support for the centered moments of the low lying zeroes of
  cuspidal newforms.
\newblock ar{X}iv:2208.02625.

\bibitem{kn:cfkrs}
J.B. Conrey, D.W. Farmer, J.P. Keating, M.O. Rubinstein, and N.C. Snaith.
\newblock Integral moments of ${L}$-functions.
\newblock {\em Proc. London Math. Soc.}, 91(1):33--104, 2005.
\newblock ar{X}iv:math.nt/0206018.

\bibitem{kn:cfs05}
J.B. Conrey, P.J. Forrester, and N.C. Snaith.
\newblock Averages of ratios of characteristic polynomials for the compact
  classical groups.
\newblock {\em Int. Math. Res. Notices}, 7:397--431, 2005.

\bibitem{kn:consna06}
J.B. Conrey and N.C. Snaith.
\newblock Applications of the {$L$}-functions ratios conjectures.
\newblock {\em Proc. London Math. Soc.}, 94(3):594--646, 2007.
\newblock ar{X}iv:math.NT/0509480.

\bibitem{kn:consna08}
J.B. Conrey and N.C. Snaith.
\newblock Correlations of eigenvalues and {R}iemann zeros.
\newblock {\em Comm. Number Theory and Phyics}, 2(3):477--536, 2008.
\newblock ar{X}iv:0803.2795.

\bibitem{kn:consna07}
J.B. Conrey and N.C. Snaith.
\newblock Triple correlation of the {R}iemann zeros.
\newblock {\em Journal de Th{\'e}orie des Nombres de Bordeaux}, 20:61--106,
  2008.
\newblock ar{X}iv:math/0610495.

\bibitem{kn:consna14}
J.B. Conrey and N.C. Snaith.
\newblock In support of {$n$}-correlation.
\newblock {\em Comm. Math. Phys.}, 330(2):639--653, 2014.
\newblock DOI:10.1007/s00220-014-1969-1, ar{X}iv:1212.5537.

\bibitem{kn:diasha}
P.~Diaconis and M.~Shahshahani.
\newblock On the eigenvalues of random matrices.
\newblock {\em Journal of Applied Probability}, 31:49--62, 1994.

\bibitem{kn:duemil06}
E.~Due{\~n}ez and S.J. Miller.
\newblock The low lying zeros of a {$GL(4)$} and a {$GL(6)$} family of
  {$L$}-functions.
\newblock {\em Compos. Math.}, 142:1403--25, 2006.

\bibitem{kn:entin14}
A.~Entin.
\newblock Artin-{S}chreier ${L}$-functions and random unitary matrices.
\newblock {\em J. Num. Theory}, 145:340--351, 2014.

\bibitem{kn:entrodrud}
A.~Entin, E~Roditty-Gershon, and Z.~Rudnick.
\newblock Low-lying zeros of quadratic {D}irichlet {$L$}-functions,
  hyper-elliptic curves and random matrix theory.
\newblock {\em Geom. and Func. Analysis}, 23(4):1230--1261, 2013.

\bibitem{kn:fazzari24}
A.~Fazzari.
\newblock A weighted one-level density of families of {$L$}-functions.
\newblock {\em Algebra and Number Theory}, 18(1), 2024.

\bibitem{kn:fiomil15}
D.~Fiorilli and S.~J. Miller.
\newblock Surpassing the ratios conjecture in the 1-level density of
  {D}irichlet {$L$}-functions.
\newblock {\em Algebra and Number Theory}, 9:13--52, 2015.

\bibitem{kn:fouiwa03}
E.~Fouvry and H.~Iwaniec.
\newblock Low-lying zeros of dihedral {$L$}-functions.
\newblock {\em Duke Math. J}, 116(2):189--217, 2003.

\bibitem{kn:gao14}
P.~Gao.
\newblock ${N}$-level density of the low-lying zeros of quadratic {D}irichlet
  ${L}$-functions.
\newblock {\em Int. Math. Res. Not.}, 2014(6):1699--1728, 2014.

\bibitem{kn:goesetal}
J.~Goes, S.~Jackson, S.J. Miller, D.~Montague, K.~Ninsuwan, R.~Peckner, and
  T.~Pham.
\newblock A unitary test of the ratios conjecture.
\newblock {\em J. Number Theory}, 130:2238--2258, 2010.

\bibitem{kn:guloglu05}
A.M. G{\"u}lo{\u g}lu.
\newblock Low-lying zeros of symmetric power ${L}$-functions.
\newblock {\em Int. Math. Res. Not.}, 9:517--550, 2005.

\bibitem{kn:hamwong}
A.~Hamieh and P.-J. Wong.
\newblock Low-lying zeros of {$L$}-functions of ad{\'e}lic {H}ilbert modular
  forms and their convolutions.
\newblock ar{X}iv:2412.03034v1.

\bibitem{kn:hejhal94}
D.A. Hejhal.
\newblock On the triple correlation of zeros of the zeta function.
\newblock {\em Inter. Math. Res. Notices}, 7:293--302, 1994.

\bibitem{kn:hugmil07}
C.P. Hughes and S.J. Miller.
\newblock Low-lying zeros of {$L$}-functions with orthogonal symmetry.
\newblock {\em Duke Math. J.}, 136(1):115--172, 2007.

\bibitem{kn:huymilmor}
D.K. Huynh, S.J. Miller, and R.~Morrison.
\newblock An elliptic curve test of the ${L}$-function ratios conjecture.
\newblock {\em J. Number Theory}, 131:1117--1147, 2011.

\bibitem{kn:ILS99}
H.~Iwaniec, W.~Luo, and P.~Sarnak.
\newblock Low lying zeros of families of ${L}$-functions.
\newblock {\em Inst. Hautes {\'E}tudes Sci. Publ. Math.}, 91:55--131, 2000.

\bibitem{kn:katzsarnak99a}
N.M. Katz and P.~Sarnak.
\newblock {\em Random Matrices, Frobenius Eigenvalues and Monodromy}.
\newblock American Mathematical Society Colloquium Publications, 45. American
  Mathematical Society, Providence, Rhode Island, 1999.

\bibitem{kn:katzsarnak99b}
N.M. Katz and P.~Sarnak.
\newblock Zeros of zeta functions and symmetry.
\newblock {\em Bull. Amer. Math. Soc.}, 36:1--26, 1999.

\bibitem{kn:keasna00b}
J.P. Keating and N.C. Snaith.
\newblock Random matrix theory and ${L}$-functions at $s=1/2$.
\newblock {\em Comm. Math. Phys}, 214:91--110, 2000.

\bibitem{kn:keasna00a}
J.P. Keating and N.C. Snaith.
\newblock Random matrix theory and $\zeta(1/2+it)$.
\newblock {\em Comm. Math. Phys.}, 214:57--89, 2000.

\bibitem{kn:lesesvre}
D.~Lesesvre.
\newblock Low-lying zeros of {$L$}-functions for quaternion algebras.
\newblock ar{X}iv:1810.13257v2.

\bibitem{kn:levmil13}
J.~Levinson and S.J. Miller.
\newblock The $n$-level densities of low-lying zeros of quadratic {D}irichlet
  {$L$}-functions.
\newblock {\em Acta Arithmetica}, 161(2):145--182, 2013.

\bibitem{kn:limil}
J.S. Li and S.J. Miller.
\newblock Bounding vanishing at the central point of cuspidal newforms.
\newblock {\em J. Number Theory}, 244:279--307, 2023.

\bibitem{kn:milliu}
S.-C. Liu and S.J. Miller.
\newblock Low-lying zeros for {$L$}-functions associated to {H}ilbert modular
  forms of large level.
\newblock {\em Acta Arithmetica}, 180(3), 2017.

\bibitem{kn:massna16}
A.M. Mason and N.C.Snaith.
\newblock Orthogonal and symplectic {$n$}-level densities.
\newblock {\em Memoirs of the AMS}, 251(1194), 2018.
\newblock ar{X}iv:1509.05250.

\bibitem{kn:massna16b}
A.M. Mason and N.C. Snaith.
\newblock Symplectic $n$-level densities with restricted support.
\newblock {\em Random Matrices: Theory and Applications}, 5(4), 2016.

\bibitem{kn:maumilmil}
A.~Mauro, J.B. Miller, and S.J. Miller.
\newblock Extending the support of 1- and 2-level densities for cusp form
  {$L$}-functions under square-root cancellation hypothesis.
\newblock {\em Acta Arithmetica}, 214:289--309, 2024.

\bibitem{kn:mehta3}
M.L. Mehta.
\newblock {\em Random Matrices}.
\newblock Elsevier, Amsterdam, third edition, 2004.

\bibitem{kn:mil04}
S.J. Miller.
\newblock One- and two-level densities for rational families of elliptic
  curves: evidence for the underlying group symmetries.
\newblock {\em Compos. Math.}, 140(4):952--992, 2004.

\bibitem{kn:mil07}
S.J. Miller.
\newblock A symplectic test of the {$L$}-functions ratios conjecture.
\newblock {\em Int. Math. Res. Notices}, 2008.
\newblock DOI:10.1093/imrn/rnm146, ar{X}iv:0704.0927.

\bibitem{kn:mil09}
S.J. Miller.
\newblock An orthogonal test of the {$L$}-functions ratios conjecture.
\newblock {\em Proc. London Math. Soc.}, 99:484--520, 2009.
\newblock ar{X}iv:0805.4208.

\bibitem{kn:milmon}
S.J. Miller and D.~Montague.
\newblock An orthogonal test of the {$L$}-functions ratios conjecture, {II}.
\newblock {\em Acta Arithmetica}, 146(1):53--90, 2011.

\bibitem{kn:milpec12}
S.J. Miller and R.~Peckner.
\newblock Low-lying zeros of number field $l$-functions.
\newblock {\em J. Number Theory}, 132:2866--2891, 2012.

\bibitem{kn:mont73}
H.L. Montgomery.
\newblock The pair correlation of the zeta function.
\newblock {\em Proc. Symp. Pure Math}, 24:181--93, 1973.

\bibitem{kn:odlyzko89}
A.M. Odlyzko.
\newblock The $10^{20}$th zero of the {R}iemann zeta function and 175 million
  of its neighbors.
\newblock {\em preprint}, 1992.
\newblock \verb|https://www-users.cse.umn.edu/~odlyzko/unpublished/index.html|.

\bibitem{kn:ozlsny99}
A.~E. {\"O}zl{\"u}k and C.~Snyder.
\newblock On the distribution of the nontrivial zeros of quadratic
  {$L$}-functions close to the real axis.
\newblock {\em Acta Arith.}, 91(3):209--228, 1999.

\bibitem{kn:ricroy11}
G.~Ricotta and E.~Royer.
\newblock Statistics for low-lying zeros of symmetric power {$L$}-functions in
  the level aspect.
\newblock {\em Forum Math.}, 23:969--1028, 2011.
\newblock ar{X}iv:math.NT/0703760.

\bibitem{kn:rub01}
M.O. Rubinstein.
\newblock Low-lying zeros of {$L$}-functions and random matrix theory.
\newblock {\em Duke Math. J.}, 109(1):147--181, 2001.

\bibitem{kn:rudsar}
Z.~Rudnick and P.~Sarnak.
\newblock Zeros of principal ${L}$-functions and random matrix theory.
\newblock {\em Duke Math. J.}, 81(2):269--322, 1996.

\bibitem{kn:shitem16}
S.W. Shin and N.~Templier.
\newblock Sato-{T}ate theorem for families and low-lying zeros of automorphic
  {$L$}-functions.
\newblock {\em Invent.Math.}, 2023:1--177, 2016.

\bibitem{kn:young}
M.~P. Young.
\newblock Low-lying zeros of families of elliptic curves.
\newblock {\em J. Amer. Math. Soc.}, 19(1):205--250, 2006.
\newblock ar{X}iv:math.nt/0406330.

\end{thebibliography}
%\bibliographystyle{plain}
\end{document}